\documentclass[a4paper,11pt,leqno]{amsart}
\usepackage{a4wide,color,array}
\usepackage[dvipsnames]{xcolor}

\usepackage{cite}
\usepackage{hyperref}
\usepackage{amsmath,amssymb,amsthm,color}
\hypersetup{linktocpage,colorlinks, citecolor={magenta}}
\usepackage[english]{babel}
\usepackage[utf8]{inputenc}

\usepackage{xspace}
\usepackage{bm}				
\usepackage{bbm,upgreek}		
\usepackage[e]{esvect}		
\usepackage{esint}			
\usepackage{etoolbox}			
\usepackage{calc}				
\usepackage{eso-pic}			
\usepackage{datetime}			

\usepackage{lmodern}
\numberwithin{equation}{section}

\usepackage{enumitem}
\setlist{nosep}
\setlist{noitemsep}

\newcommand{\Z}{\mathbb{Z}}

\newcommand{\R}{\mathbb{R}}
\newcommand{\C}{\mathbb{C}}

\newtheorem{theo}{Theorem}
\newtheorem{prop}{Proposition}[section]
\newtheorem{lem}[prop]{Lemma}
\newtheorem{coro}[prop]{Corollary}

\newtheorem{remark}[prop]{Remark}

\theoremstyle{definition}
\newtheorem{defi}[prop]{Definition}
\newtheorem{step}{Step}

\def \be{\begin{equation}}
\def \ee{\end{equation}}

\def \t0{\rightarrow 0} 
\def \ti{\rightarrow \infty} 

\def \hal{\frac{1}{2}}

\def \div{\mathrm{div} \,} 
\def \1{\mathbf{1}} 

\def \p{\partial}
\def \ep{\varepsilon}
\def \dist{\mathrm{dist}}

\def\a{\alpha}

\def\cd{\mathsf{c_{\d}}}

\def\({\left(}
\def\){\right)}
\def\nab{\nabla}

\def \KNbeta{K_{N,\beta}}
\def \probas{\mathcal{P}}
\def \PNbeta{\mathbb{P}_{N, \beta}} 

\newcommand{\ed}{}

\def\XXint#1#2#3{{\setbox0=\hbox{$#1{#2#3}{\int}$}
     \vcenter{\hbox{$#2#3$}}\kern-.5\wd0}}

\def \XN{\vec{X}_N}
\def \YN{\vec{Y}_N}

\def \fbeta{\mathcal{F}_{\beta}}
\def \tH{\widetilde{H}}
\def \FN{F_N}

\def \KNbeta{K_{N,\beta}}

\def \Fluct{\mathrm{Fluct}}
\def \fluct{\mathrm{fluct}}

\def \mueq{\mu_0} 

\def\Esp{\mathbf{E}} 
\def \E{\Esp}

\def \KNbeta{K_{N,\beta}}
\def \ZNbetaV{Z_{N,\beta}}

\def \Ani{\mathsf{A}_s}

\def\I{\mathcal{I}}
\def\i{\int_{\R^2}}

\def \id{\mathrm{Id}}
\def \L{\ell}
\def \mueq{\mu_{0}}
\def \mueqt{\mu_{t}}
\def \hmueqt{h^{\mueqt}}
\def \mut{\overline{\mu}_{t}}
\def \hmut{h^{\mut}}

\def\g{\mathsf{g}}
\def\d{\mathsf{d}}

\def \zetaz{\zeta_0}
\def \zetat{\zeta_t}

\def \Ent{\mathrm{Ent}}

\def \Mean{\mathrm{Mean}}
\def \Var{\mathrm{Var}}
\def \Main{\mathrm{Main}}
\def \hD{\widehat{D}}
\def \bD{\overline{D}}

\def \C{\mathcal{C}}

\def\dis{\displaystyle}
\def\indic{\mathbf{1}}
\def \E{\mathcal{E}}

\makeatletter
\def\namedlabel#1#2{\begingroup
    #2%
    \def\@currentlabel{#2}%
    \phantomsection\label{#1}\endgroup
}
\makeatother

\def \epsilon{\varepsilon}
\def \bxN{\bar{x}_N}

\def \D{B} 
\def \C{\mathcal{C}} 
\def \f{\mathsf{f}}

\def \diam{\mathrm{diam}}

\def \Error{\mathsf{Error}}

\def \veta{\vec{\eta}}
\def \tmax{t_{\mathrm{max}}}
\def\Supp{{\mathrm{Supp } \,}}
\def \HN{\mathcal{H}_N}
\def \rr{\mathsf{r}}
\def \PNbetaV{\mathbb{P}_{N,\beta}}
\def \tmut{\tilde{\mu}_t}
\def \tzetat{\tilde{\zeta}_t}
\def \vecr{\vec{\rr}}

\def \coT{}

\def \comTT{}
\begin{document}
\title{Fluctuations of Two Dimensional  Coulomb gases}
\author{Thomas Lebl\'e and Sylvia Serfaty}
\date{Version 3: \today}
\address[Thomas Lebl\'e]{Sorbonne Universit\'es, UPMC Univ. Paris 06, CNRS, UMR 7598, Laboratoire Jacques-Louis Lions, 4, place Jussieu 75005, Paris, France.
\newline current address: Courant Institute, New York University, 251 Mercer st, New York, NY 10012, USA.}
\email{thomasl@cims.nyu.edu}

\address[Sylvia Serfaty]{Sorbonne Universit\'es, UPMC Univ. Paris 06, CNRS, UMR 7598, Laboratoire Jacques-Louis Lions, 4, place Jussieu 75005, Paris, France.
 \newline \& Institut Universitaire de France \newline \& Courant Institute, New York University, 251 Mercer st, New York, NY 10012, USA.}
\email{serfaty@cims.nyu.edu}
\begin{abstract}
We prove a Central Limit Theorem for the linear statistics of two-dimensional Coulomb gases, with arbitrary inverse temperature  and general confining potential, at the macroscopic and mesoscopic scales and possibly near the boundary of the support of the equilibrium measure. This can be stated in terms of convergence of the random electrostatic potential to a Gaussian Free Field. 

Our result is the first to be valid at arbitrary temperature and at the mesoscopic scales, and we recover previous results of Ameur-Hendenmalm-Makarov and Rider-Vir\'{a}g concerning the determinantal case, with weaker assumptions near the boundary. We also prove moderate deviations upper bounds, or rigidity estimates, for the linear statistics and a convergence result for those corresponding to energy-minimizers.

The method relies on a change of variables, a perturbative expansion of the energy, and the comparison of partition functions deduced from our previous work. Near the boundary, we use recent quantitative stability estimates on the solutions to the obstacle problem obtained by Serra and the second author.
\end{abstract}

\maketitle
{\bf keywords:} Coulomb Gas, $\beta$-ensembles, Log Gas, Central Limit Theorem, Gaussian Free Field, Linear statistics, Ginibre ensemble.

\section{Introduction} \label{sec:intro}
\subsection{Presentation of the problem}
Let $\beta > 0$ be fixed. For $N \geq 1$, we are interested in the $N$-point canonical Gibbs measure\footnote{We use $\frac{\beta}{2}$ instead of $\beta$ in order to match  the normalizations in the existing literature. The first sum in \eqref{def:HN} is twice the physical one, but it is more convenient for our analysis.} for a two-dimensional Coulomb (or log-) gas at the \textit{inverse temperature} $\beta$, defined by
\begin{equation}\label{def:PNbeta}
d\PNbetaV(\XN) = \frac{1}{\ZNbetaV} \exp \left( - \frac{\beta}{2} \HN(\XN)\right) d\XN,
\end{equation}
where $\XN = (x_1, \dots, x_N)$ is an $N$-tuple of points in $\R^2$ and $\HN(\XN)$, defined by
\begin{equation} \label{def:HN}
\HN(\XN) := \sum_{1 \leq i \neq  j \leq N} - \log | x_i-x_j| + \sum_{i=1}^N N V(x_i),
\end{equation}
is the energy of the system in the state $\XN$, given by the sum of the pairwise repulsive logarithmic interaction between all particles plus the effect on each particle of an external field or \textit{confining potential} $NV$ whose intensity is proportional to $N$. The constant $\ZNbetaV$ in the definition \eqref{def:PNbeta} is the normalizing constant, called the \textit{partition function}, and is equal to 
\begin{equation*}
\ZNbetaV := \int_{(\R^2)^N} \exp \left( - \frac{\beta}{2} \HN(\XN)\right) d\XN.
\end{equation*}

Under mild assumptions on  $V$, the empirical measure of the particles converges almost surely to a deterministic equilibrium measure $\mueq$ as $N \to \infty$, 
\coT{see e.g. \cite[Chap.2]{serfatyZur} and the references therein.}

For any $N \geq 1$, let us define the \textit{fluctuation measure}
\begin{equation} \label{def:fluctN}
\fluct_N :=  \sum_{i=1}^N \delta_{x_i} - N \mueq.
\end{equation}
It is a random signed measure on $\R^2$. For any real-valued test function $\xi_N$, possibly depending on $N$, we define the \textit{fluctuations of the linear statistics associated to} $\xi_N$ as the real random variable
\begin{equation} \label{def:FluctN}
\Fluct_N(\xi_N) := \int_{\R^2} \xi_N\, d\fluct_N.
 \end{equation}
The main goal of this paper is to prove a Central Limit Theorem (CLT) for the random variable $\Fluct_N(\xi_N)$ under some regularity assumptions on $V$ and $\xi$.  This can be translated into the convergence to a Gaussian Free Field of the random potential
\begin{equation} \label{def:Deltam1}
\Delta^{-1} \fluct_N := \frac{1}{2\pi} \int_{\R^2} - \log | \cdot - x | d\fluct_N(x).
\end{equation}
Here and in the rest of the paper $\Delta$ denotes the usual Laplacian operator on $\R^2$ given by $\Delta f = \partial^2_x f + \partial^2_y f$.

\subsection{Notation}
Throughout the paper, we will work in H\"older spaces $C^{k,1}$, the spaces of real-valued functions with $k$ derivatives on $\R^2$ and a Lipschitz $k$-th derivative. We endow $C^{k,1}$ with the norm 
$$
\|f\|_{C^{k,1}} := |f|_0+ \sum_{j=1}^{k+1}  |f|_j,
$$
where the $|f|_j$ are the semi-norms 
\ed{\begin{equation}
\label{seminorm}
|f|_0:= \sup_x |f(x)| \qquad |f|_j := 
\sup_{x \neq y}  \frac{|f^{(j-1)}(x)-f^{(j-1)}(y)|}  {|x-y|},
\end{equation}}
with the obvious generalization to vector fields. We denote compactly supported function with a subscript $C^{k,1}_c$.

$\D(x, R)$ denotes the disc of center $x$ and radius $R$, and if $A \subset \R^2$ we let $\mathbf{1}_{A}$ be the indicator function of $A$.

For $f \in C^0(\R^2)$ we denote by $f^{\Sigma}$ the harmonic extension of $f$ outside $\Sigma$, namely the unique continuous map which coincides with $f$ on $\Sigma$ (up to $\partial \Sigma$) and is harmonic and bounded in $\R^2 \setminus \Sigma$. It is not hard to see that when $f$ is $C^1$ then $f^{\Sigma}$ is in $C^{0,1}$, but the normal derivative of $f^{\Sigma}$ on $\partial \Sigma$ may present a discontinuity. In \cite[Section 2.5]{ahm}, the \textit{Neumann jump} is defined on $\partial \Sigma$ as the difference between the inner and outer normal derivative of $f^{\Sigma}$.

\coT{If $\mu$ is a measure on $\R^\d$ and $\phi : \R^\d \to \R^d$ we let $\phi \# \mu$ be the push-forward of $\mu$ by $\phi$.}

Finally, we write $a \preceq b$ when $a$ is bounded by $b$ times some universal constant, and $a \approx b$ if $a\preceq b $ and $b \preceq a$.

\subsection{Assumptions}
\label{sec:assumpot}
We will always assume that $\xi_N$ is either independent of $N$, which we call the {\it macroscopic case}, or that $\xi_N$ has the form 
\begin{equation}
\label{def:xiN} \xi_N(x) := \xi \left( \frac{x-\bxN}{\L_N} \right).
\end{equation}
for some sequence $\{\bxN\}_N$ of points in $\R^2$ and a sequence $\{\L_N\}_N$ of positive real numbers, tending to $0$ slower than $N^{-\hal}$, which we call the {\it mesoscopic case}.
 In particular, we have, with the notation of \eqref{seminorm}
\begin{equation}
\label{controlexi_N} |\xi_N|_{k} \leq |\xi|_{k} \L_N^{-k}.
\end{equation}

Let us now describe our assumptions.
\begin{description}
\item[\namedlabel{H1}{(H1)} - Regularity  and growth of $V$] The potential $V$ is in $C^{3,1}(\R^2)$ and satisfies the growth condition
\begin{equation*}
\liminf_{|x| \ti} \frac{V(x)}{2\log |x|} > 1.
\end{equation*}
\end{description}
It is well-known, see e.g. \cite{safftotik} that if $V$ satisfies \ref{H1} then the \textit{logarithmic potential energy} functional defined on the space of probability measures by 
\begin{equation}
\label{MFener}
\mathcal{I}_V (\mu) := \iint_{\R^2 \times \R^2} - \log |x-y| \, d\mu(x)\, d\mu(y)+ \int_{\R^2}V(x)\, d\mu(x)
\end{equation}
has a unique global minimizer $\mueq$,  the \textit{equilibrium measure associated to $V$}. This measure has a compact support, sometimes called the \textit{droplet}, that we will denote by $\Sigma$, and $\mueq$ is characterized by the fact that there exists a constant $c_0$ such that the function $\zeta_0$ defined by
\begin{equation} \label{zeta}
\zeta_0(x) := \int_{\R^2} - \log |x-y| d\mueq(y) + \frac{V(x)}{2} - c_0
\end{equation}
satisfies the Euler-Lagrange conditions
\begin{equation} \label{EulerLagrange}
\zeta_0 \geq 0 \text{ in } \R^2, \quad \zeta_0 = 0 \text{ in } \Sigma.
\end{equation}
The set $\omega$, defined by
\begin{equation}
\label{def:omegavanish} \omega := \{x, \zeta_0(x) = 0\},
\end{equation}
 where $\zeta_0$ vanishes is called the \textit{coincidence set} or \textit{contact set} in the obstacle problem literature (the obstacle being here the function $ c_0-\frac{V}{2}$), for the correspondence see for instance \cite[Section 2.5]{serfatyZur}. 

We make some additional assumptions on $V$, $\mueq$ and $\Sigma$.
\begin{description}
\item[\namedlabel{H3}{(H2)} - Non-degeneracy] We have  $\Delta V > 0$ in the coincidence set $\omega$.
\end{description}
This assumption ensures in particular that the support of the equilibrium measure $\Sigma$  is exactly the coincidence set  of the associated obstacle problem, whereas in general we only have the inclusion $\Sigma \subset \omega$.
\begin{description}
\item[\namedlabel{H4}{(H3)} - Additional regularity for the boundary case] The boundary of $\Sigma$ is a \ed{finite union of } $C^{2,1}$ curves and all its points are \textit{regular} i.e. there are no cusps in the sense of Caffarelli (see e.g.  \cite{MR1658612}), \cite[Definition 3.24]{petrosyan2012regularity}).
\end{description}

\subsection{The central limit theorem}
\begin{theo}[Central limit theorem for fluctuations of linear statistics] \label{theo:CLT} 
Let us distinguish three cases.
\begin{description}
\item[Macroscopic interior case]  $\xi$ is in $C^{2,1}_c(\Sigma)$, $\bar{x}_N = 0, \L_N = 1$ in \eqref{def:xiN},  and \ref{H1}, \ref{H3} hold.
\item[Macroscopic boundary case]  $\xi$ is in $C^{3,1}_c(\R^2)$, $\bar{x}_N = 0, \L_N = 1$, and  \ref{H1}, \ref{H3}, \ref {H4}  hold. Moreover, if $\Sigma$ has several (\ed{finitely many}) connected components, we assume that the following conditions are satisfied:
\ed{\begin{equation}\label{condit}
\int_{ \partial \Sigma_i}\( \frac{\partial \xi^{\Sigma}}{\partial n}(y)\)\, dx=0
\quad  \text{ for  each $i$,}
\end{equation}}
where $\{\Sigma_i\}_i$ are the connected components of the support, \coT{and $\frac{\partial \cdot}{\partial n}$ is the one-sided partial derivative in the normal direction computed outside $\Sigma_i$.}
\end{description}

In either of these two cases,  $\Fluct_N(\xi)$ converges in law to a Gaussian random variable with mean 
\begin{equation}
\label{mean}
\Mean(\xi) := \frac{1}{2\pi} \left( \frac{1}{\beta}-\frac{1}{4}\right) \int_{\R^2} \Delta \xi \left( \mathbf{1}_{\Sigma} + \left( \log \Delta V \right)^{\Sigma}  \right)
\end{equation}
and variance
\begin{equation}
\label{variance}
\Var(\xi) := \frac{1}{2 \pi \beta} \int_{\R^2} |\nabla \xi^{\Sigma} |^2.
\end{equation}

\begin{description}
\item[Mesoscopic case]   $\xi$ is in  $C^{2,1}(\D(0,1))$, $\L_N = N^{-\delta}$ for some $\delta \in (0, \hal)$, and $\bxN$ is in the interior of $\Sigma$ (at distance greater than $4 \L_N$ from $\partial \Sigma$). We also assume that  \ref{H1}, \ref{H3} hold. 
\end{description}
Then $\Fluct_N(\xi_N)$ converges in law to a Gaussian random variable with mean 0 and variance
\begin{equation*}
\displaystyle{\frac{1}{2 \pi \beta} \int_{\R^2} |\nabla \xi |^2}.
\end{equation*}
\end{theo}
One may observe some universality feature in the fact that the variance \eqref{variance} does not depend on $V$. \comTT{The fact that the mean is zero in the mesoscopic cases can be heuristically deduced from \eqref{mean} : if $\xi$ is supported in $\Sigma$ and varies on a much smaller lengthscale than $V$ we see that 
$$
\int_{\R^2} \Delta \xi \left( \mathbf{1}_{\Sigma} + \left( \log \Delta V \right)^{\Sigma}  \right) \approx   \left( \int_{\R^2} \Delta \xi \right) \int_{\R^2} \left( \log \Delta V \right)^{\Sigma}  \approx 0.
$$
}

As in \cite{ahm2}, the convergence of the random potential $\Delta^{-1} \fluct_N$  to a Gaussian Free Field (GFF) is just a direct consequence of the very definition of the GFF (cf. for instance \cite{MR2322706}) \comTT{: more precisely, it means that, for any smooth test function $\xi$ compactly supported in $\Sigma$, the random variable
$$
\langle \xi, \Delta^{-1} \fluct_N \rangle,
$$
(where the pairing is defined as $\langle f, g \rangle := \int \nabla f  \cdot \nabla g$) converges to a Gaussian random variable whose variance is proportional to $\langle \xi, \xi \rangle = \int |\nabla \xi|^2$. The mean is nonzero in general (we get a non-centered Gaussian free field). If the support of $\xi$ intersects the boundary of $\Sigma$, the variance carries an extra term due to the presence of the harmonic extension of $\xi$ in \eqref{variance}, this was described in \cite{ridervirag} as convergence “to the planar Gaussian free field conditioned to be harmonic outside” the support of the equilibrium measure.}

By \textit{interior cases} we will mean the \textit{macroscopic interior case} and the \textit{mesoscopic case}. The same result as Theorem \ref{theo:CLT} has been proven  independently at the same time in \cite{BBNY2}, only for the interior cases, with $V$ and $\xi$ assumed to be $C^4$.
We compare the two approaches in Section \ref{outline}.

\subsection{Comments on the assumptions}
\begin{itemize}
\item It is well-known that if $V$ is \ed{$C^{2}$ for instance} then the density of the equilibrium measure is given by
\begin{equation} \label{mueqfromV}
d\mueq(x) = \frac{1}{4\pi} \Delta V(x) \mathbf{1}_{\Sigma}(x) dx.
\end{equation}
In particular, Assumption \ref{H1} implies that $\mueq$ has a $C^{2}$ density on its support, Assumption  \ref{H4} implies that its boundary is a \ed{$C^{2,1}$} curve, and Assumption \ref{H3} ensures that the density of $\mueq$ is bounded below by a positive constant on $\Sigma$.

\item The points of the boundary $\partial \omega$ of the coincidence set can be either \textit{regular}, i.e. $\partial \omega$ is locally the graph of a $C^{1, \alpha}$ function, or \textit{singular}, i.e. $\partial \omega$ is locally cusp-like (this classification was introduced in \cite{MR1658612}). 
 Singularities are nongeneric and we believe that the assumption that there are none might be bypassed. Let us emphasize that in the macroscopic interior case as well as in the mesoscopic case, we do not require Assumption \ref{H4}.

\item If $V$ is $C^{3,1}$, then $\partial \omega $ is locally $C^{2,1}$ around each regular point, see \cite[Thm. I]{caffarelli1976smoothness}.

\item If $V$ is strictly convex (which implies that $\Sigma$ and $\omega $ coincide) and of class $C^{k+1,\alpha}$ on $\R^2$, it was shown (see \cite[Section 4]{kinderlehrer1978variational}), in the slightly different setting of a bounded domain, that $\Sigma$ is connected and that $\partial \Sigma$ is $C^{k, \alpha}$ with no singular points.

\item 
In the case of  $p\ge 2$ connected components  of the equilibrium measure, then the  $p-1$ additional assumptions \eqref{condit} are needed, similarly as in the {\it multi-cut} one-dimensional case \cite{BorGui2,scherbi1}.  The CLT is not true for a general test function without such assumptions -- for instance one cannot take $\xi$ to count the number of points in or near a given component -- and instead the limit in such cases is the convolution of a Gaussian variable with a discrete Gaussian, \coT{this is known in the one-dimensional case, and expected to hold also in the two-dimensional setting although we do not pursue this goal here}.
\end{itemize}

\subsection{Additional results}
\begin{theo}[Moderate deviations \coT{upper bounds}] \label{theo:ModDev2} Under the same assumptions as Theorem \ref{theo:CLT} (including the same assumptions on $\L_N$), there  exists $c > 0$ such that for any
$1 \ll \tau_N \ll N\L_N^2$ we have
\begin{equation*}
\PNbeta\left(|\Fluct_N(\xi_N)| \geq c \tau_N \right)  \leq \exp\left(- \frac{c^2}{2} \tau_N^2 \right),
\end{equation*}
\end{theo}
This way, we retrieve a rigidity result similar to that of \cite[Theorem 1.2]{BBNY}.

The following result is an elementary consequence of Theorem \ref{theo:CLT}.
\begin{coro} \label{theo:independance} In all the cases where the CLT holds,  let $m \geq 1$ and let $\left\lbrace\xi^{(k)}\right\rbrace_{k =1, \dots, m}$ be $C^{2, 1}$ (resp. $C^{3,1}$ in boundary cases)  compactly supported test functions. Then the joint law of the fluctuations $\left\lbrace\Fluct_N(\xi^{(k)})\right\rbrace_{k=1, \dots, m}$ converges to the law of an $m$-dimensional Gaussian vector, the marginals being as in Theorem \ref{theo:CLT} and the covariance matrix being given by
$$
\left\lbrace \int_{\R^2} \nab (\xi^{(i)})^{\Sigma} \cdot \nab (\xi^{(j)})^{\Sigma} \right\rbrace_{1 \leq i, j \leq m}.
$$
\end{coro}

Finally, we consider minimizers of the energy. Although this formally corresponds to $\beta\to \infty$, the limits $\beta\to \infty$ and $N \to \infty$ cannot be directly commuted but our analysis in fact applies as well, and yields a new rigidity-type result down to the microscopic scale:
\begin{theo}[Fluctuations for energy minimizers]\label{thmini}
Under the same assumptions as~Theorem~\ref{theo:CLT},
assume $\XN$ minimizes the energy $\HN$ as in \eqref{def:HN}.  Let $\Fluct_N$ be defined\footnote{In this context it is of course not  random.} as in \eqref{def:FluctN}. In the two macroscopic cases we have
\begin{equation*}
\lim_{N\to \infty} \Fluct_N(\xi) = \frac{-1}{8\pi} \int_{\R^2} \Delta \xi \left( \indic_{\Sigma} + (\log \Delta V)^\Sigma)\right).
\end{equation*}
In the mesoscopic case, with $\xi$ in $C^{2,1}_c(\D(0,1))$, $\bar{x}_N$  in the interior of $\Sigma$ at distance greater than $4\L_N$ from $\partial \Sigma$ and $\L_N$ such that $\L_N = o_N(1)$ and\footnote{Let us observe that the second condition is weaker than the assumption $\L_N = N^{-\delta}$ with $\delta \in (0, \hal)$ as in the mesoscopic case of Theorem \ref{theo:CLT}, it allows to consider test functions living at a large, microscopic scale.} $N^{-1/2} = o(\L_N)$, then
\begin{equation*}
\lim_{N\to \infty}\Fluct_N(\xi_N) = 0.
\end{equation*}
\end{theo}

\comTT{
\begin{remark}[Corrections to the mean-field approximation]
A  consequence of the main result is that 
$$\lim_{N\to \infty} \Esp_{\PNbeta} (\Fluct_N(\xi))= \Mean(\xi)$$ 
which, after spelling out the definition of $\Fluct_N$, can be rephrased in terms of the first marginal of the Gibbs measure as
\begin{equation}\label{firstmarg}
\lim_{N\to \infty} N \int_{\R^2} \xi d(\PNbeta^{(1)} -\mueq)  = \Mean(\xi),\end{equation} giving the order $1/N$ correction to the mean-field  approximation $\PNbeta^{(1)} \sim \mueq$.
Expanding in the same way higher order moments of the fluctuations could in principle give access to corrections to the mean-field approximation for   higher order marginals $\PNbeta^{(k)}$ of the Gibbs measure.
\end{remark}
}

\subsection{Motivation and existing literature}
The model described by \eqref{def:PNbeta} and \eqref{def:HN} is known in statistical physics as a \textit{two-dimensional Coulomb gas}, \textit{two-dimensional log-gas} or \textit{two-dimensional one-component plasma}, we refer e.g. to \cite{alastueyjancovici}, \cite{jlm}, \cite{sm} for a physical treatment of its main properties. Such ensembles have been the object of interest for statistical mechanics, the fractional quantum Hall effect (\comTT{as pioneered by Laughlin in \cite{laughlin1983anomalous}}, see also e.g. \cite{rougerie2015incompressibility}, \cite{stormer1999fractional} and the references therein),  and also due to their connection with random matrices: when $\beta=2$ and $V(x)=|x|^2$, the Gibbs measure \eqref{def:PNbeta} coincides with the law of the eigenvalues of the Ginibre ensemble (see \cite{ginibre,mehta}). The case $\beta=2$ (for general $V$) is one that happens to be determinantal, allowing the use of exact formulas. We refer to \cite{forrester} for a survey of the connection between log-gases and random matrix theory, and in particular to \cite[Chap.15]{forrester} for the two-dimensional, non-Hermitian case.

Systems of particles with a logarithmic interaction as in \eqref{def:HN}, called \textit{log-gases}, have  also (and mostly) been studied on the real line, motivated by their link with Hermitian random matrix theory. There has been a lot of attention to the phenomenon of “universality” in such ensembles, which consists in showing that a large part of the behavior of the system, in particular microscopic statistics, is independent of the exact form of the  potential $V$. Universality of the point processes at the microscopic scale and rigidity\footnote{Let us emphasize that the “rigidity” mentioned here is  different from the notion of “rigid” point processes as studied e.g. in \cite{Ghosh:2012qd}.} estimates for the positions of the particules  were established in \cite{bourgade1d1,bourgade1d2,bfg}. A Central Limit Theorem for the fluctuations (at macroscopic scale) was proved in the pioneering paper of \cite{MR1487983} for polynomial potentials, followed by generalizations in \cite{scherbi1} for real analytic $V$ and in the multi-cut case (in which case the CLT does not hold for all smooth test functions), see also \cite{Lambert:2017kq} for a recent new method of proof with quantitative estimates on the convergence rate in the one-cut case.  A CLT for the fluctuations of linear statistics \textit{at mesoscopic scales}  was recently obtained in \cite{Bekerman:2016qf}. Expansions of the partition function (which in particular imply CLT's) are established in \cite{BorGui1,BorGui2,scherbi1}. These CLT results are also extended in \cite{Bekerman:2017fj} to more general so-called  critical cases  with weaker regularity assumptions by adapting the method of the present paper.

The literature is less abundant in the two-dimensional case, which is the object of the present study. The CLT for fluctuations at the macroscopic scale and convergence to a Gaussian Free Field was obtained in \cite{ridervirag} in the case of the Ginibre ensemble, in \cite{ahm} for $\beta=2 $ and general $V$ in the bulk case, in \cite{ahm2} for $\beta=2$ and analytic $V$ in the boundary case.
In our previous work \cite{ls1} (extended to the mesoscopic scales in \cite{loiloc}), we proved a Large Deviations Principle for the empirical fields associated to these ensembles, i.e. for the point processes seen at the microscopic scale and averaged. Our result also contained a next order expansion of the partition function, and allows one to derive large deviations bounds for the fluctuations of linear statistics. Local laws and moderate deviations bounds at any mesoscopic scale have appeared in \cite{BBNY}, derived from the so-called \textit{loop equations} associated to the problem. Non-asymptotic concentration bounds for two- and higher-dimensional Coulomb gases have been derived in \cite{chafai2016concentration}. The fluctuations of a particular observable and the associated fourth order phase transition are also studied in \cite{cunden2015fluctuations}. 

The study of minimizers of $\HN$ (without temperature) has  been  of interest recently. A next to leading order expansion of $\min \HN$, together with the convergence of minimizers to those of a ``Coulomb renormalized energy" was obtained in \cite{SS2d}. Additional rigidity
  of the minimizing configurations   down to the microscopic scale were shown   in \cite{aoc,nodari2014renormalized} (see also \cite{petrache2016equidistribution} for the higher-dimensional Coulomb cases). An explicit upper bound on the particle density was also derived with a different approach in \cite{lieb2016rigidity} with applications to the Fractional Quantum Hall Effect, see also previous ``incompressibility estimates" in \cite{rougerie2014quantum,rougerie2015incompressibility}.
\smallskip

A very remarkable feature of these CLT's is that no $\frac{1}{\sqrt{N}}$ normalization is needed to obtain a Gaussian limit (in contrast with the usual CLT for i.i.d random variables).  As expressed e.g. in \cite{MR1487983, ahm}, this must result from effective cancelations caused by the repulsive behavior of the particles.
Theorem \ref{theo:CLT} recovers results of \cite{ahm,ahm2,ridervirag} in the $\beta =2$, macroscopic scale case but with a different method.  
We need stronger regularity assumptions  than  in \cite{ridervirag} where test functions are only assumed to be $C^1$ (for quadratic $V$); and weaker assumptions than in \cite{ahm2} where $V$ and  $\p \Sigma$  are analytic while $\xi \in C^\infty_c(\R^2)$. The optimal regularity needed for $\xi$ in order for the CLT to hold seems to be an interesting open question and we believe our result can be improved in that direction. 

In \cite[Section 2.6]{ahm2} (which is $\beta = 2$) the value of the mean of the limiting Gaussian random variable is expressed as
\begin{equation*}
\frac{1}{8 \pi} \left( \int_{\Sigma} \left(\Delta \xi + \xi \Delta (\log \Delta V)\right) + \int_{\partial \Sigma} \xi \mathcal{N} \left[ (\log \Delta V)^{\Sigma} \right] \right),
\end{equation*}
where $\mathcal{N}$ is the “Neumann jump” on $\partial \Sigma$ of the harmonic extension.
An integration by parts allows to easily check that it coincides with the expression of $\Mean(\xi)$ given in Theorem \ref{theo:CLT}. The expression for the variance in \cite{ahm2} is the same as ours.

In \cite{ridervirag} (which corresponds to $\beta =2$ and $\Sigma = \D(0,1)$), the variance is expressed as
\begin{equation} \label{varRV}
\frac{1}{4 \pi} \int_{\D(0,1)} | \nabla \xi |^2 + \frac{1}{2} \sum_{k \in \Z} |k| |\hat{\xi}(k)|,
\end{equation}
where $\hat{\xi}(k)$ is the $k$-th Fourier coefficient \comTT{of the map restricted to the unit circle}. The second term can be viewed as $\|\xi\|_{\dot{H}^{1/2}(\partial \D(0,1))}^2$, where $\dot{H}^{1/2}$ is the homogeneous fractional Sobolev space on $\partial \D(0,1)$ with exponent $1/2$. It is not hard to check that
\begin{equation*}
\frac{1}{4 \pi} \int_{\R^2 \setminus \Sigma} | \nabla \xi^{\Sigma}|^2 = \frac{1}{2} \| \xi \|^2_{\dot{H}^{1/2}(\partial \Sigma)},
\end{equation*}
and we recover \eqref{varRV} from $\Var(\xi)$. Let us emphasize that, in contrast with \cite{ridervirag}, in \eqref{def:FluctN} we do not substract the expectation but $N$ times the limit and we find that the expectation of the fluctuations is given, as in \cite{ahm2}, by
\begin{equation*}
\frac{1}{8\pi} \int_{\D(0,1)} \Delta \xi.
\end{equation*}

\subsection{Open questions}
A first natural open question  is to know the minimal regularity that needs to be assumed on $\xi$ for the CLT (or the order of magnitude $O(1)$ of the fluctuations) to hold. It is expected (see \cite{jlm}) that when $\xi$ is not continuous -- for instance when $\xi$ is the indicator function  of a set -- then its fluctuations (i.e. the fluctuations of the number of points in the set)
are of order $N^{1/4}$, hence much larger  than $O(1)$, but still much smaller than that of i.i.d. points, thus  still exhibiting a rigidity phenomenon.  This was proven up to logarithmic corrections, and also in dimensions 1 and  3, for a ``hierarchical" Coulomb gas model (in which the interaction is modified in such a way that the system naturally gets ``coarse-grained" when changing scales) in the  recent work \cite{Chatterjee:2017ly}.
On this aspect, and many others, much more is known for other two-dimensional models of particles with strong repulsion, given by the zeroes of random polynomials or random series -- we refer e.g. to \cite{nazarov2010fluctuations, Sodin:kk} and the references therein.

It is also natural to search  for a generalization of Theorem~\ref{theo:CLT} in the mesoscopic case under the condition $N^{-1/2} \ll \L_N$ instead of the actual constraint that $\L_N = N^{-\delta}$ with $\delta > \hal$, as can be done for minimizers (see Theorem \ref{thmini}). The existence of a limit point process for $\PNbeta$, i.e. of a two-dimensional analogue of the sine-$\beta$ process appearing for one-dimensional $\beta$-ensembles, is not known for $\beta \neq 2$ (in the Ginibre case $\beta =2$, explicit expressions are known e.g. for all the correlation functions), but it would also be interesting to study the asymptotic normality of fluctuations in such hypothetical infinite ensembles.

\comTT{Recently, there has been a surge of  interest in studying the extreme values (in a certain sense) of Gaussian Free Fields (GFF's) and fields that ressemble a GFF, for example the characteristic polynomial of a matrix in the Circular Unitary Ensemble, see \cite{paquette2017maximum} and references therein. In our setting, the field $\Delta^{-1} \fluct_N$  is the characteristic polynomial in the $\beta=2$, $V$ quadratic, “Ginibre ensemble” case, and can be thought of as the “characteristic polynomial” of a hypothetic non-Hermitian matrix whose eigenvalues are the $\XN$ for $\beta, V$ general. Since it converges to a GFF, it would be natural to compare the maximal values of this field to the corresponding quantities for a GFF.}

\subsection{Outline of the proof and the paper and further remarks}\label{outline}
Our approach is based on the energy approach initiated in \cite{SS2d,ss1d,RougSer,PetSer}, which consists in expressing the interaction energy in terms of the electrostatic potential generated by the point configuration.
We are able to leverage on the result of our previous papers on the LDP \cite{ls1,loiloc} which provided a next-order expansion of the partition function which is  explicit in terms of the equilibrium measure, \ed{presented in Section \ref{sec2}}. We show here that such an expansion allows to quickly obtain a CLT: our method is conceptually simple and flexible. \comTT{As mentioned above, our method applies in the one-dimensional logarithmic case \cite{Bekerman:2017fj} and provides a rather simple proof of the previously known result, although  the treatment of  the mesoscopic case by this method is still an open question and seems more difficult than in the two-dimensional case due to the nonlocal nature of the half-Laplacian operator (of which $-\log$ is the fundamental solution in dimension one). In fact, since the one-dimensional macroscopic setting is significantly easier, we encourage the reader interested in the details of the proof to consult \cite{Bekerman:2017fj,serfaty2017microscopic} for a first reading. 
Our method can also be  extended to  higher dimensions, this is the object of future work. }
\smallskip

 \ed{Let us now outline the proof and the paper.
 The first step is to split the energy into \be\label{splitfirst}\HN(\XN)\sim N^2 \I(\mueq)+ F_N(\XN, \mueq)\ee
where $F_N$ is the Coulomb   interaction of the system formed by the point charges at the $x_i$'s and the negative background charge  $-N\mueq$.  This is presented in  Section \ref{sec2}, where we also show that thanks to the known expansion of the partition function,  $F_N+ \hal N\log N $ is of order $N$ \coT{and we control its exponential moments}. Since $F_N$ controls the fluctuation measure, we can deduce first concentration bounds on it. }
\comTT{After splitting the energy  as in \eqref{splitfirst}, we may simplify out the contribution of $N^2 \I(\mueq)$ from the partition function, and define a next-order partition function involving only $F_N$ and denoted for now $K_{N,\beta}(\mueq)$.}

\comTT{A random variable is a Gaussian with mean $m$ and variance $v$  if and only if its Laplace transform is $e^{mt+\hal vt^2}$, hence, as is well-known, to prove the convergence  in law of a random variable to a Gaussian, it suffices to show that the logarithm of its   Laplace transform converges to a quadratic function. This is the starting point of the proof, as in   \cite{MR1487983} and all other previous works: we wish to  compute the large $N$ limit of the (logarithm of) the Laplace transform of the fluctuation in the form  
 $ \Esp_{\PNbeta} \( \exp(Nt \Fluct_N(\xi_N)) \right).
$
  Some straightforward explicit computations detailed in Proposition \ref{prop:LaplaceTransform} show that (in the interior case, for simplicity)}
 \begin{equation}\label{espf}  \Esp_{\PNbeta} \( \exp(Nt \Fluct_N(\xi_N)) \right)= e^{\hal N^2 t^2  \Var(\xi)}
 \frac{K_{N,\beta}(\mueqt)}{K_{N,\beta} (\mueq)}\qquad \Var(\xi)= \frac{1}{2\pi \beta} \int_{\R^\d} |\nab \xi|^2,
 \end{equation}\comTT{where we already see the variance $\Var(\xi)$ formally appear.  In order to obtain the Laplace transform of the fluctuations, one needs to take $t= \frac{\tau}{N}$ in \eqref{espf}, where $\tau $ is fixed and let $N \to \infty$.} \footnote{\comTT{Implementing this strategy, we compute asymptotics of various quantities in $t$ or $\tau$, where $\tau = Nt$ should be thought of as being order~$1$. In particular, we may  discard all lower-order terms of order $Nt^2$ or $\sqrt{N}t$ since they have a vanishing contribution in the limit $N \to \infty$, but of course not the terms in $N^2 t^2$ appearing in the variance.}}

\comTT{ Evaluating \eqref{espf}  thus reduces  to understanding 
  the ratio of the next order  partition functions associated to a Coulomb gas with “perturbed” potential  $V-\frac{2t}{\beta} \xi$ and equilibrium measure $\mueqt$, and the original one.  
Our previous work \cite{ls1,loiloc} provides us  with an expansion of  $\log K_{N,\beta}(\mu)$, which leads to an expansion of the ratio  of the form 
\begin{equation}\label{logkm}
\log  \frac{K_{N,\beta}(\mueqt)}{K_{N,\beta} (\mueq)}= \left( 1-\frac{\beta}{4}\right)  N\left( \int \mueq \log \mueq-  \int \mueqt \log \mueqt\right) +o(N).
\end{equation}
 It has an explicit Lipschitz dependence in $\mu$, plus a $o(N)$ error term  (see Section \ref{sec:partitionfunctionsexpansions}). Using this expansion with $t = \frac{\tau}{N}$ we  obtain 
$$
\log \frac{K_{N,\beta}(\mu_{\tau/N})  }{K_{N,\beta} (\mueq)} = \Mean(\xi) \tau  + \Error_N(\tau) + o(N), 
$$ 
where  the explicit mean $\Mean(\xi)$ now appears as the linearization  with respect to $t \to 0$ of the entropy terms in the right-hand side  \eqref{logkm}, and where  $\Error_N(\tau)$ goes to zero as $N \to \infty$ for fixed $\tau$. Since we do not know that the error terms in \eqref{logkm} have a Lipschitz dependence in the equilibrium measure, we can only bound their difference by their sum, leading to the  $o(N)$ error instead of the  $o(1)$ error that we need, and preventing us from directly concluding.
Our way to circumvent this is  to combine this estimate with a second way of computing $ \frac{K_{N,\beta}(\mueqt)}{K_{N,\beta} (\mueq)}$, discussed below.}
 As it turns out, the second approach will still not directly yield an $o_N(1)$ error in the comparison of partition functions, but we will be able to combine both estimates in order to get the result, see Corollary \ref{coro:Aniso}.

 \comTT{The second way of computing    is to look for a change of variables  (as is fairly common in this topic, see for example  \cite{MR1487983,BorGui1,scherbi2}) that will exactly map the old Coulomb gas to the new one (with perturbed potential).} This leads to the question of inverting  an operator and the loop (or Schwinger-Dyson) equations.
Instead, we   use  a change of variables \ed{$\id + t\psi$} which is a transport map between the equilibrium measure $\mueq$ and  an approximation  of the equilibrium measure $\mueqt$ for the perturbed potential. 

The conditions \eqref{condit} ensure that the mass  of the perturbed equilibrium measure carried by each connected component remains unchanged to order $t$, so that \ed{ a regular enough such $\psi$ exists. We then let $\tmut= (\id + t\psi)\# \mueq$ be the approximate equilibrium measure. 
Section \ref{sec:transport} is devoted to the construction of $\psi$ and to proving that 
$$\log \frac{K_{N,\beta}(\tmut)}{K_{N,\beta}(\mueqt)}=o(1)$$ 
which allows to replace $\mueqt$ by the approximate measure $\tmut$ in \eqref{espf}.}

The construction of the transport map $\psi$ is easy in the interior cases. In the boundary cases, it requires a precise understanding on  how the support of $\mueqt$ varies with $t$, which is  a question of quantitative stability for the coincidence set of the obstacle problem under perturbation of the obstacle which we believe to be of independent interest. Such a result was missing in the literature and  is proven in all dimensions in a separate paper \cite{SerSer}. Our approach is in 
contrast to what was previously found in the literature where  the  analyticity of $V, \xi, \p \Sigma$ is often assumed --- in particular it was assumed in the only paper  \cite{ahm2} that previously treated the boundary case, and was required  in order to  be able to apply the sophisticated Sakai's theory, which is anyway restricted to two \smallskip 
dimensions. 

\ed{
The next step, presented in  Section \ref{secani} is to use the  change of variables $\phi_t=\id + t\psi$ in the integral that defines $K_{N,\beta}(\tmut)$. This leads us to evaluating (roughly)
\begin{equation}
\label{eav}
\Esp_{\PNbeta}\( \exp\( -\frac{\beta}{2}\( F_N(\phi_t (\XN), \tmut)- F_N(\XN, \mueq)\)+ \sum_{i=1}^N \log |\det D\phi_t|\) \).\end{equation}
Then, we linearize the exponent in $t$.  A large part of the analysis is to linearize the difference of energies $F_N$ before and after transport, this is done in Proposition \ref{prop:comparaison2} and relies on our energetic approach, which allows to use regularity  estimates from potential theory and elliptic PDEs. This somehow replaces the loop equation terms. 
What we find is that roughly
\begin{multline}\label{expaF} 
-\frac{\beta}{2}\(F_N(\phi_t (\XN), \tmut)- F_N(\XN, \mueq)\)
+ \sum_{i=1}^N \log |\det D\phi_t| \\ \simeq 
 \Ani(t\psi, \XN, \mueq)+  \( 1- \frac{\beta}{4}\) N\( \int \mueq \log \mueq-  \int \tmut \log \tmut\) +O(t^2 N)
 \end{multline}
i.e. the linearization gives rise to explicit terms which are the same as in \eqref{logkm} plus an additional explicit but rather unknown term  (which we call the {\it anisotropy}) $\Ani(t\psi, \XN, \mueq)$ whose important features are that it   is {\it linear} in $t$ and controlled by the energy $F_N+\hal N\log N$.
 \comTT{In other words, using the transport and linearizing the relevant quantities along this transport opens the possibility of obtaining an expansion for the  {\it relative partition function} in \eqref{logkm} which is now Lipschitz in the equilibrium measure, as opposed to the previous expansion obtained by substracting those obtained for each equilibrium measure, but involve the new unknown term $\Ani(t\psi, \XN, \mueq)$.
 }

To conclude, the key is to to compare the results \eqref{logkm} and \eqref{eav}--\eqref{expaF}   obtained by the two approaches, for $t$ fixed but possibly small. This   yields 
\begin{equation}
\label{espan}
\log \Esp_{\PNbeta}\left[ \exp\left( \Ani(t\psi, \XN,\mueq)\right)\right] = o(N) 
\end{equation}
i.e.  the anistropy is small, with respect to $N$, in exponential moments. 
Using then \comTT{crucially} its linear character \comTT{(and not so much its precise expression)}, we can transfer this information from $t$ fixed to $t=\frac{\tau}{N}$ by simply using H\"older's inequality in \eqref{espan}, obtaining
$$\log \Esp_{\PNbeta}\left[ \exp\left( \Ani(\frac{\tau}{N}\psi, \XN,\mueq)\right)\right] = o(1) .$$
Inserting this into \eqref{eav}--\eqref{expaF} we thus conclude the evaluation of \eqref{espf} by finding 
\begin{equation}\label{esperance}\lim_{N\to \infty} \Esp_{\PNbeta}\( \exp(\tau \Fluct_N(\xi))\) = \exp\(\Mean(\xi) \tau  +\hal \Var(\xi) \tau^2 \),\end{equation} where  the mean and variance are as in \eqref{mean} and \eqref{variance}.}
This  concludes at the  end of Section~\ref{secani} the proof of the main theorem. 
In  Section \ref{concl} we give the proofs of the other theorems. 

In Appendix \ref{sec:proofcomparaison}, we prove  Proposition \ref{prop:comparaison2}, and in Appendix \ref{app2} we gather the  proofs  of the preliminary results of Section \ref{sec2}.
Finally, in Appendix \ref{sec:comparaisonmieux} we provide additional detail for the reader interested in the precise import of the results of \cite{ls1,loiloc}.   
\medskip

In \cite{BBNY2}, instead of transport, the method relies on loop equations and strong rigidity estimates (taken from \cite{BBNY}).  As we do, authors of \cite{BBNY2} use an expansion of the partition function  which is explicit in terms of $\mu_0$  and they  transfer information from large $t$'s to smaller $t$'s to show that the contribution of some anisotropy-type  term is small.
 Because they cannot take $t$ as large as order $1$, they instead rely on an expansion of the log of the partition function with  a quantitative bound on the error term, which in turn is obtained by comparing   the  Coulomb gas with logarithmic interaction  to one with  a
short-range (screened) Yukawa interaction,  for which rigidity and the existence of a thermodynamic limit can be proven, effectively replacing the screening procedure used in \cite{ls1}.

\medskip

\section{Preliminaries}\label{sec2}
For future reference  we will sometimes work in general dimension $\d$ (here $\d=2$) and denote the logarithmic potential $-\log |x|$ by $\g(x)$, with 
$$-\Delta \g= \cd \delta_0,$$ 
here  $\cd=2\pi$ in dimension $2$.

\subsection{The next-order energy} \label{sec:nextorder}
We start by presenting the electric formulation to the energy.
\begin{defi}[Next-order energy] 
Let $\mu$ be a bounded, compactly supported probability density on $\R^{\d}$. We define an energy functional on $(\R^{\d})^N$ by
\begin{equation}
\label{def:FN} \FN(\XN, \mu) :=  \iint_{(\R^\d \times \R^\d) \setminus \triangle} \g (x-y)\, \left(\sum_{i=1}^N \delta_{x_i} - Nd \mu\right) (x) \left(\sum_{i=1}^N \delta_{x_i} - Nd \mu\right) (y), 
\end{equation}
where $\triangle$ is the diagonal in $\R^\d \times \R^\d$.
\end{defi}

\begin{lem}[Splitting formula] \label{lem:splitting}
Assume $\mueq $, the minimizer of $\mathcal {I}_V$ (as in \eqref{MFener}), is absolutely continuous with respect to the Lebesgue measure. 
For any $N$ and any $\XN \in (\R^{\d})^N$ we have
\begin{equation}\label{split0}
\HN(\XN) =  N^2 \mathcal{I}_{V}(\mueq) + 2N \sum_{i=1}^N \zeta_0(x_i) +F_N(\XN, \mueq).
\end{equation}
\end{lem}
The proof of Lemma \ref{lem:splitting} is given in Section \ref{sec:preuvelemsplitting}.

\def \PNbetamuzeta{\mathbb{P}_{N, \beta}^{(\mu, \zeta)}}
Using  \eqref{split0}, we may re-write $\PNbetaV$ as
\begin{equation} \label{PNbetaVdeux}
d\PNbetaV(\XN) = \frac{1}{\KNbeta(\mueq, \zeta_0)} \exp\left( - \frac{\beta}{2} \left( \FN(\XN, \mueq) + 2N \sum_{i=1}^N \zeta_0(x_i) \right) \right) d\XN,
\end{equation}
with a next-order partition function $\KNbeta(\mueq, \zeta_0)$ defined by
\begin{equation} \label{def:KNbeta}
\KNbeta(\mueq, \zeta_0) := \int_{(\R^\d)^N} \exp\left( - \frac{\beta}{2} \left( \FN (\XN, \mueq) + 2N \sum_{i=1}^N \zeta_0(x_i) \right) \right) d\XN.
\end{equation}
We extend this notation to $\KNbeta(\mu, \zeta)$ where $\mu$ is a (bounded, compactly supported) probability density and $\zeta$ a “confining term”. We also define 
\begin{equation} \label{def:PNbetamuzeta}
d\PNbetamuzeta(\XN) := \frac{1}{\KNbeta(\mu, \zeta)} \exp\left( - \frac{\beta}{2} \left( \FN (\XN, \mu) + 2N \sum_{i=1}^N \zeta(x_i) \right) \right) d\XN.
\end{equation}

\subsection{Electric fields and truncation}
Contrarily to our previous works, we can only afford here total errors that are $o(1)$ as $N \to \infty$ since the Laplace transform of the fluctuation is of order $1$. Thus we need improved versions of the previous results, which allow to have exact formulas.  To do so,  we use the rewriting of the energy via truncation as in \cite{RougSer,PetSer} but using the  nearest-neighbor distance truncation as in \cite{2D2CP}.

For any $N$-tuple $\XN=(x_1, \dots, x_N)$ of points in the space $\R^\d$, and any bounded, compactly supported probability density $\mu$, we define the electrostatic potential generated by $\XN$ and $\mu$ as
\begin{equation}\label{def:HNmuABC}
H_N^{\mu}(x) :=\int_{\R^\d}\g(x-y)\, \left(\sum_{i=1}^N \delta_{x_i} - N d\mu\right)(y).
\end{equation}
For $\eta > 0$, let us define the truncation at distance $\eta$ by
\begin{equation*}\label{def:truncation}
\f_\eta(x) := \min \left(\g(x)-\g(\eta), 0\right).
\end{equation*}
If $\XN = (x_1, \dots, x_N)$ is a $N$-tuple of points in $\R^\d$ we denote for all $i=1, \dots, N$,
\begin{equation}\label{def:trxi}
\rr(x_i)= \frac{1}{4} \min\left(\min_{j \neq i} |x_i-x_j|, N^{-1/\d}\right)
\end{equation}
which we will think of as the \textit{nearest-neighbor distance} for $x_i$.  Let $\vec{\eta}$ be a $N$-tuple of small distances $\vec{\eta} = (\eta_1, \dots, \eta_N)$.
We define the truncated potential $H_{N,\vec{\eta}}^{\mu}$ as \comTT{
\begin{equation}\label{def:HNmutrun}
H_{N,\vec{\eta}}^{\mu} (x)= H_N^{\mu}(x)-\sum_{i=1}^N \f_{\eta_i}(x-x_i).
\end{equation}}
This amounts to truncating the singularity of $H_{N}^{\mu}$ near each particle $x_i$ at distance $\eta_i$.
We note that since $\g$ is a multiple of the Coulomb kernel in $\R^\d$, $H_N^{\mu}$ satisfies 
\begin{equation} \label{rem:aproposdeHNeta}
-\Delta H_{N}^{\mu}= \cd\, \left(\sum_{i=1}^N \delta_{x_i} - N d\mu\right) \text{ in $\R^\d$,}
\end{equation}
with $\cd = 2\pi$ for $\d=2$. 
\comTT{Also, denoting $\delta_x^{(\eta)}$ the uniform measure of mass $1$ on $\p B(x, \eta)$, we note that we have 
$$
\f_{\eta}(x)= \int_{\R^\d} \g(x-y) \( \delta_0- \delta_0^{(\eta)}\)(y),$$
hence  in view of \eqref{def:HNmuABC} and \eqref{def:HNmutrun}, we may write}
\begin{align}
\label{eqhne} H_{N, \vec{\eta}}^{\mu} (x) & = \int_{\R^\d} \g(x-y) \left(\sum_{i=1}^N \delta_{x_i}^{(\eta_i)}  -N d\mu\right)(y), \\
\label{dhne} -\Delta H_{N, \vec{\eta}}^{\mu}  & =\cd\left(\sum_{i=1}^N \delta_{x_i}^{(\eta_i)}  -N d\mu\right).
\end{align}

The (truncated) electric fields are defined as the gradient of the (truncated) electric potentials. The main point of introducing these objects is that we may express the next-order energy $\FN(\XN, \mu)$ in terms of the (truncated) electric fields.
\begin{prop} \label{prop:monoto}
Let $\mu$ be a bounded probability density on $\R^{\d}$ and $\XN$ be in $(\R^\d)^N$. We may re-write $\FN(\XN, \mu)$ as  
\begin{equation*}
\label{def:FNbis}
\FN(\XN, \mu) := \frac{1}{\cd} \lim_{\eta \to 0} \left(\int_{\R^\d}|\nab H_{N, \vec{\eta}} ^{\mu}|^2  - \cd \sum_{i=1}^N \g(\eta_i)  \right).
\end{equation*}
If $\vec{\eta} = (\eta_1, \dots, \eta_N)$ is such that $0 < \eta_i \le \rr(x_i)$ for each $i = 1, \dots, N$ we have
 \begin{equation}
\label{fnmeta}
 \FN(\XN,\mu) = \frac{1}{\cd} \left(\int_{\R^\d}|\nab H_{N, \vec{\eta}} ^{\mu}|^2  -\cd \sum_{i=1}^N  \g(\eta_i)  \right)  
 + 2N \sum_{i=1}^N \int_{\R^{\d}} \f_{\eta_i}(x - x_i) d\mu(x)
\end{equation}
and for general $\vec{\eta}$ we have the bounds
 \begin{multline}
\label{fnmeta2}
 \sum_{i\neq j} \Big(\g(x_i-x_j) -\min ( \g(\eta_i), \g(\eta_j))\Big) \indic_{|x_i-x_j|\le \eta_i + \eta_j}  \\ \le  \FN(\XN,\mu) - \frac{1}{\cd} \left(\int_{\R^\d}|\nab H_{N, \vec{\eta}} ^{\mu}|^2  -\cd \sum_{i=1}^N  \g(\eta_i)  \right)  
 - 2N \sum_{i=1}^N \int_{\R^{\d}} \f_{\eta_i}(x - x_i) d\mu(x) \\ \le   \sum_{i \neq j} \g(x_i-x_j) \indic_{|x_i-x_j|\le \eta_i + \eta_j},
\end{multline}
where the error terms in \eqref{fnmeta} and \eqref{fnmeta2} may be  bounded as follows
\begin{equation} \label{ecartaeta}
\left|  2N \sum_{i=1}^N \int_{\R^{\d}} \f_{\eta_i}(x - x_i) d\mu(x)  \right|\le C N \|\mu\|_{L^\infty} \sum_{i=1}^N \eta_i^{\d},
\end{equation}
for some constant $C$ depending only on $\d$. 
\end{prop}
The proof of Proposition \ref{prop:monoto} is given in Section \ref{sec:preuvepropmonoto}.

\ed{Choosing in particular $\eta_i =N^{-1/\d}$, we deduce from \eqref{fnmeta2} and \eqref{ecartaeta}
that 
\begin{coro}\label{corominoe} For any $\XN$, we have
\begin{equation}
F_N(\XN,\mu) \ge - N \g(N^{-\frac{1}{\d}})-CN \|\mu\|_{L^\infty}\end{equation}
for some constant depending only on $\d$.\end{coro}}
\coT{In particular, for $\d=2$ we obtain
$$
F_N(\XN, \mu) \geq \hal N \log N - CN \| \mu \|_{L^{\infty}}.
$$
}

\subsection{The electric energy controls the fluctuations}
\coT{In this section, we explain how to derive a priori, deterministic bounds on the fluctuations of a smooth test function in terms of the electric energy of the points. The basic idea is to use the fact that $\sum_{i=1}^N \delta_{x_i} - N \mu$ is, up to constant, minus the Laplacian of $H_N^{\mu}$ and to write
$$
\Fluct_N[\varphi] = \int \varphi \left(\sum_{i=1}^N \delta_{x_i} - N \mu \right) \simeq \int \varphi \Delta  H_N^{\mu} \simeq \int \nabla \varphi \nabla H_N^{\mu}
$$
and to apply Cauchy-Schwarz's inequality.} \comTT{This is again similar to  previous works \cite{SS2d,RougSer,PetSer}, but with more explicit dependence in the test-functions.}
\ed{
\begin{prop}\label{prop:fluctenergy}
Let  $\varphi$ be a compactly supported Lipschitz function on $\R^\d$ and $\mu$ be a bounded probability density on $\R^\d$.  Let $U_N$ be an open set   containing a $\delta$-neighborhood of the support of $\varphi$, with $\delta \ge 2  N^{-1/\d}$. Let $\vec{\eta}$ be a $N$-tuple of distances such that $\eta_i \le N^{-1/\d}$, for each $i = 1, \dots, N$. For any configuration $\XN$, we have
\begin{multline} \label{controlfluctuations}
\left|\int_{\R^\d} \varphi \, \left(\sum_{i=1}^N \delta_{x_i} - N d\mu \right) \right|
\le C \|\nab \varphi\|_{L^2(U_N)} \|\nab H_{N,\vec{\eta}}^{\mu} \|_{L^2(U_N )}\\
  + C \|\nab\varphi\|_{L^\infty} \(   \delta^{-\hal}|\p U_N|^{\hal} N^{-\frac{1}{\d}} \|\nab H_{N,\vec{\eta}}^{\mu} \|_{L^2(U_N )}+  N^{1-\frac{1}{\d}}|U_N|  \|\mu\|_{L^\infty(U_N)}\right)
\end{multline}
where $C$ depends only on $\d$.

Let $S_N$ be a compact subset of $\R^\d$ and $U_N$ containing its $\delta$-neighborhood, with $\delta\ge 2N^{-1/\d}$, and let $\#I_{S_N} $ denote the number of balls $B(x_i, N^{-1/\d})$ intersecting  $S_N$. We have
\begin{equation}
\label{contrnbpoints}
\# I_{S_N} \le  N \int_{U_N}  d\mu + C \delta^{-\hal} |\p U_N|^{\hal}      \|\nab H_{N,\vec{\eta}}^{\mu} \|_{L^2(U_N)},
\end{equation}where $C$ depends only on $\d$.
\end{prop}
In particular, for $\d=2$ and $\mu = \mueq$, we obtain
\begin{multline}\label{controlfluctuations2}
 \left|\Fluct_N(\varphi) \right|
  \\ 
  \le C \|\nab\varphi\|_{L^\infty} \left(( |U_N|^{\hal}  +\delta^{-\hal}|\p U_N|^{\hal} N^{-\frac{1}{2}}) \|\nab H_{N,\vec{\eta}}^{\mueq} \|_{L^2(U_N)}+  N^{\hal}|U_N|  \|\mueq\|_{L^\infty(U_N)}\right),  
\end{multline}
and an estimate on the number of points as in \eqref{contrnbpoints}
\begin{equation}
\label{contrnbpoints2} 
\# I_{S_N} \le  N \|\mueq\|_{\infty} |U_N| +   C\delta^{-\hal}|\p U_N|^{\hal} \|\nab H_{N,\vec{\eta}}^{\mueq} \|_{L^2(U_N)}.
\end{equation}

\begin{coro} \label{coro:nombredepointspresdubord}
Let $\mu$ be a bounded probability density with compact support $\Sigma$ such that $\partial \Sigma$ is a piecewise $C^{1}$ curve. For any configuration $\XN$, letting $I_\partial^r $ denote the set of points such that $\dist(x_i, \partial \Sigma) \le r$ or $ x_i \notin \Sigma$, if $r>0$ is smaller than a constant depending only on $\partial \Sigma$ and for any $\vec{\eta}$, we have
\begin{equation}\label{nbpbord}
 \# I_{\partial}^r \le  C  \min\( N^{\frac{1}{6}}, r^{-\hal}\) \|\nab H_{N,\vec{\eta}}^\mu\|_{L^2(\Sigma)} + C \max (N^{\frac{2}{3}}, Nr),
 \end{equation} where $C$ depends only on $\mu$ and $\d$.
\end{coro}
The proofs of Proposition \ref{prop:fluctenergy}  and Corollary \ref{coro:nombredepointspresdubord} are  given in Section \ref{sec:preuvefluctenergy}.}

\subsection{Local control of the distances}
\def \EnergieLoc{\mathsf{F}^{\bar x_N, \L_N}}
We  will  need  the following result, that shows that the electric energy 
locally controls  the nearest neighbor distances.
Let $\vec{\eta}$ be such that $\eta_i \le \rr(x_i) $ for all $i=1, \dots , N$. \ed{For a given background measure $\mu$,} let us introduce 
\begin{equation}\label{deflocen}
\EnergieLoc_{\vec{\eta}}(\XN) := N^{\frac{2}{\d}-1}  \int_{B(\bar{x}_N,2\L_N)} |\nab H_{N,\veta}^{\mu}|^2 - \cd \sum_{i, x_i \in B(\bar x_N, 2 \L_N) } \g(\eta_i N^{1/\d}  ).
\end{equation}

\begin{lem}
 \label{lem:contrdist}
For any configuration, there exists $\vec{\alpha}$  with  $\alpha_i \le N^{-1/\d}$, such that 
\begin{equation*}
\sum_{i \in B(\bar{x}_N, \L_N)} \g (\rr(x_i) N^{1/\d}) \le C \( \EnergieLoc_{\vec{\alpha}} (\XN) + 
 \EnergieLoc_{\vec{\eta}} (\XN) +  N\L_N^\d + \sum_{i , x_i\in B(\bar{x}_N, \L_N)}\g(1) \),
 \end{equation*}
 where  $\eta_i= N^{-1/\d}$ for each $i$, and 
  $C$ depends only on $\|\mueq\|_{L^\infty}$ and $\d$.
\end{lem}
Lemma \ref{lem:contrdist} is proven in Section \ref{sec:preuvecontrdist}.

\subsection{Perturbed quantities}
\label{sec141}
\begin{defi} \label{def:mueqt}
For any $t \in \R$ and $N \geq 1$, we define
\begin{itemize}
\item The perturbed potential $V_t$ as $V - \frac{2t\xi_N}{\beta}$.
\item The perturbed equilibrium measure $\mueqt$ as the equilibrium measure associated to $V_t$. 
\item The perturbed droplet $\Sigma_t$ as the compact support of $\mueqt$.
\item The next-order confinement term $\zetat$ as in \eqref{zeta}.
\item The next-order energy $\FN(\XN, \mueqt)$ as in \eqref{def:FN}.
\item The next-order partition function $\KNbeta(\mueqt, \zetat)$ as in \eqref{def:KNbeta}.
\end{itemize}
\end{defi}

 We also define $\mut$ as the signed measure
\begin{equation} \label{def:mut}
\mut := \mueq- \frac{t}{\cd \beta} \Delta \xi_N.
\end{equation}
\coT{For $\d =2$, it corresponds to
$$
\mut := \mueq - \frac{t}{2\pi \beta} \Delta \xi_N.
$$
}
\begin{remark} \label{rem:mut}
Let $N \geq 1$ be fixed. If $\xi_N$ is supported in $\Sigma_0$ and $t$ is such that 
\begin{equation} \label{tmax}
|t| \leq \tmax := \frac{\cd \beta \min_{\Sigma_0} \mueq}{2  \|\Delta \xi_N\|_{L^{\infty}} }
\end{equation}
then $\mueqt= \mut$ and $\Sigma_t=\Sigma_0$.
\end{remark}
\begin{proof}
From \ref{H3} we know that the density $\mueq$ is bounded below by a positive constant. Thus for $t$ as in \eqref{tmax}, $\mut$ is a probability density whose support is $\Sigma_0$. On the other hand, we can check that $\mut$ satifies the Euler-Lagrange equations \eqref{EulerLagrange} associated to $V_t$. \comTT{It is known since \cite{frostman}} that these equations characterize the equilibrium measure, hence $\mueqt = \mut$ and in particular $\Sigma_t = \Sigma_0$.
\end{proof}
Let us observe that if $\xi$ is $C^2$ and $\xi_N$ is as in \eqref{def:xiN}, then $\tmax$ is of order
$$
\tmax \approx \L_N^2.
$$

\subsection{The Laplace transform of fluctuations as a ratio of partition functions}
\coT{The following proposition expresses the Laplace transform of the fluctuations of a linear statistic as a ratio of partition functions. A deterministic term appears, which will later be identified as the variance of the fluctuations. For the sake of simplicity of the presentation, we take the log of the Laplace transform.}
\begin{prop} \label{prop:LaplaceTransform}
In the interior cases, for $|t| \leq \tmax$ as in \eqref{tmax}, we have the identity
\begin{multline} \label{LTttot0facile}
\log \Esp_{\PNbetaV} \left[\exp( Nt\,  \Fluct_N(\xi_N) ) \right] \\
= \log \KNbeta(\mueqt,\zeta_t) - \log \KNbeta(\mueq,\zeta_0) + \frac{N^2t^2 \L_N^{\d-2}}{2\cd \beta} \int_{\R^\d} |\nabla \xi |^2.
\end{multline} 

In the macroscopic boundary case, as $t \to 0$ we have
\begin{multline} \label{LTttot0}
\log \Esp_{\PNbetaV} \left[\exp( Nt\,  \Fluct_N(\xi_N) ) \right] \\ = \log \KNbeta(\mueqt,\zeta_t) - \log \KNbeta(\mueq,\zeta_0) + \frac{N^2t^2}{2\cd \beta} \int_{\R^\d} |\nabla \xi^{\Sigma} |^2 \\ +N^2 \left( \|\xi\|_{C^{0,1}}(1+\|\xi\|_{C^{1,1}} ) +\|\xi\|_{C^{0,1}}^2 \right) O(t^3),
\end{multline}
\end{prop}
\begin{proof}[Proof of Proposition \ref{prop:LaplaceTransform}]
\setcounter{step}{0}
For any bounded, compactly supported probability density $\mu$ we define the electrostatic potential generated by $\mu$ as
\begin{equation} \label{def:hmu}
h^{\mu}(x) := \int_{\R^\d} \g(x-y) d\mu(y),
\end{equation}
let us observe that this quantity appears in the definition of $\zeta$ as in \eqref{zeta}.

\coT{The proof relies on simple algebra.}
\begin{step}[Reexpressing fluctuations]
We have, with the notation of Definition \ref{def:mueqt} and \eqref{def:mut},
\begin{multline}\label{relfinale}
\frac{2Nt}{\beta} \Fluct_N(\xi_N) =  F_N(\XN, \mueq) - F_N(\XN, \mueqt)
-  2N \sum_{i=1}^N (\zetat(x_i)-\zetaz(x_i))\\ - 2N^2 \int_{\R^\d} \zetaz \, d\mueqt
+ \frac{N^2t^2}{\cd \beta^2} \int_{\R^\d} |\nab \xi_N|^2 - \frac{N^2}{\cd}\int_{\R^\d} \left| \nabla \hmueqt - \nabla \hmut \right|^2 .
\end{multline}
\end{step}
\begin{proof}
Letting $\nu=-\frac{1}{\cd}\Delta \xi_N$, \coT{and using the fact that 
$$
\xi_N(x) = - \int \frac{\g(x-y)}{\cd} \Delta \xi_N,
$$
}
we may write 
\begin{multline*}
F_N(\XN, \mueq)-  \frac{2N t}{\beta} \int \xi_N d\fluct_N\\
=\iint_{\triangle^c} \g (x-y) d\fluct_N(x)d\fluct_N(y) - \frac{2Nt}{\beta}\iint \g(x-y) d\nu(x) d\fluct_N(y).
\end{multline*}
Completing the square, we obtain
\begin{multline*}
F_N(\XN, \mueq)-  \frac{2N t}{\beta} \int \xi_N d\fluct_N
\\= \iint_{\triangle^c} \g(x-y) d\left( \fluct_N- \frac{Nt}{\beta} \nu\right) (x) d\left( \fluct_N- \frac{Nt}{\beta} \nu\right) (y) - \frac{N^2t^2}{\beta^2} \iint \g(x-y) d\nu(x) d\nu(y).
\end{multline*}
From the definitions we see that
$$
\fluct_N- \frac{Nt}{\beta}\nu  \coT{   = \sum_{i=1}^N \delta_{x_i} - N \mueq + \frac{Nt}{\cd \beta} \Delta \xi_N} = \sum_{i=1}^N \delta_{x_i}  -N \mut 
$$
thus we may write 
\begin{multline*}
\iint_{\triangle^c} \g(x-y) \left( \fluct_N- \frac{Nt}{\beta} d\nu\right) (x) \left( \fluct_N- \frac{Nt}{\beta}d \nu\right) (y) \\
= \iint_{\triangle^c} \g(x-y) \left( \sum_{i=1}^N \delta_{x_i} - N d\mueqt\right)(x)  \left( \sum_{i=1}^N \delta_{x_i} -N  d\mueqt\right)(y) 
\\+N^2 \iint \g(x-y) \, d( \mueqt-\mut)(x) \, d(\mueqt-\mut)(y)
\\ +2N \iint \g(x-y) \left( \sum_{i=1}^N \delta_{x_i} - Nd \mueqt\right)(x) \, d(\mueqt-\mut)(y)\\
= F_N(\XN, \mueqt) + N^2 \iint \g(x-y) \, d( \mueqt-\mut)(x) \, d(\mueqt-\mut)(y) \\ 
+ 2N\int (h^{\mueqt} - h^{\mut}) \left( \sum_{i=1}^N \delta_{x_i} - N d\mueqt \right) .
\end{multline*}
\coT{The first term in the right-hand side corresponds to $\FN(\XN, \mueqt)$. The second term in the right-hand side can be integrated by parts, yielding
$$
N^2 \iint \g(x-y) \, d( \mueqt-\mut)(x) \, d(\mueqt-\mut)(y) = N^2 \int |\nabla h^{\mueqt} - \nabla h^{\mut}|^2.
$$}
Finally, from \eqref{zeta} for $V_0$ and $V_t$ and from \eqref{def:mut}, we see that
$$
h^{\mueqt} - h^{\mut}= \zetat-\zeta_0+c_t-c_0
$$
and we also use that $\zetat$ vanishes on the support of $\mueqt$. It yields 
$$
\int \left( h^{\mueqt} - h^{\mut} \right) \left( \sum_{i=1}^N \delta_{x_i} - N d\mueqt \right)= \sum_{i=1}^N  \zetat(x_i)- \sum_{i=1}^N \zeta_0(x_i) + N \int \zeta_0 \,d \mueqt.
$$ 
Combining these successive identities, we obtain \eqref{relfinale}.
\end{proof}
\begin{step}[Expressing the Laplace transform - I]
We have the identity
\begin{multline}  \label{LTrewri1}
\Esp_{\PNbeta} \left[\exp( N t\,  \Fluct_N(\xi_N) ) \right] \\ = \frac{\KNbeta(\mueqt,\zeta_t)}{\KNbeta(\mueq,\zeta_0)} \exp \left(\frac{N^{2} t^2}{2\cd \beta} \int_{\R^{\d}} |\nabla \xi_N|^2 - \frac{\beta N^{2} }{2\cd } \int_{\R^{\d}} | \nabla \hmueqt - \nabla \hmut |^2  - \beta N^{2} \int_{\R^{\d}} \zetaz d\mueqt \right).
\end{multline}
\end{step}
\begin{proof}
Using the expression \eqref{PNbetaVdeux} of the Gibbs measure, we may compute the Laplace transform of the fluctuations 
\begin{multline*}
\Esp_{\PNbetaV} \left[\exp( N t\,  \Fluct_N(\xi_N) ) \right] \\ = \frac{1}{\KNbeta(\mueq, \zeta_0)} \int_{(\R^\d)^N} \exp\left( N t\,  \Fluct_N(\xi_N) - \frac{\beta}{2} \left( \FN(\XN, \mueq) + 2N \sum_{i=1}^N \zeta_0(x_i) \right) \right) d\XN
\end{multline*} 
Inserting \eqref{relfinale} and using the definition of $\KNbeta(\mueqt, \zeta_t)$ as in \eqref{def:KNbeta} we obtain \eqref{LTrewri1}.
\end{proof}

\begin{step}[Expressing the Laplace transform - II] 
Let us note that in the \textit{interior cases}, by Remark \ref{rem:mut},  \eqref{LTrewri1} simplifies into
\begin{equation*} \label{LTrewri1b}
\Esp_{\PNbetaV} \left[\exp( N t\,  \Fluct_N(\xi_N) ) \right]  = \frac{\KNbeta(\mueqt,\zeta_t)}{\KNbeta(\mueq,\zeta_0)} \exp \left(\frac{N^{2}   t^2}{2\cd \beta} \int_{\R^{\d}} |\nabla \xi_N|^2\right),
\end{equation*}
as soon as $|t| \leq \tmax$ as in \eqref{tmax}, and thus yields the result \eqref{LTttot0facile} in this case.

In the macroscopic boundary case, we rely on the following result.
\begin{lem} \label{lem:termesfactorises}
As $t \to 0$, we have
\begin{multline}
\label{251}
 \frac{t^2}{2\cd \beta}\int_{\R^\d} |\nab \xi_N|^2 - \frac{\beta}{2\cd } \int_{\R^\d} \left| \nabla \hmut  - \nabla \hmueqt \right|^2  =  \frac{t^2}{2\cd \beta} \int_{\R^\d}  |\nabla \xi^{\Sigma }|^2  \\ + \left( \|\xi\|_{C^{0,1}}(1+\|\xi\|_{C^{1,1}} ) +\|\xi\|_{C^{0,1}}^2 \right) O(t^3)
 \end{multline}
 \ed{and 
\begin{equation} \label{252}
 \int_{\R^2} \zeta_0 d\mueqt  = O(t^3 \|\xi\|_{C^{0,1}}^2).
\end{equation}}
\end{lem}
Lemma \ref{lem:termesfactorises} is proven in Section \ref{sec:preuvetermesfactorises}. Combining \eqref{LTrewri1} and Lemma \ref{lem:termesfactorises}, we obtain \eqref{LTttot0}.
\end{step}
\end{proof}

\begin{coro} \label{coro:LaplaceTransform}
In particular, for any fixed $\tau$, taking $t = \frac{\tau}{N}$ in the previous result, we obtain:
\begin{equation} \label{LTtto0deuxieme}
\log \Esp_{\PNbetaV} \left[\exp( \tau  \Fluct_N(\xi_N) ) \right] = \log \frac{\KNbeta(\mu_{\tau/N},\zeta_{\tau/N})}{\KNbeta(\mueq,\zeta_0)} + \frac{\tau^2}{2 \cd \beta}  \int_{\R^\d} |\nabla \xi^{\Sigma} |^2 + o_N(1),
\end{equation}
and the convergence as $N \to \infty$ is uniform for bounded $\tau$.
\end{coro}
Our goal will thus be to estimate the ratio of partition functions appearing in the right-hand side of \eqref{LTtto0deuxieme} in order to compute the limit of the left-hand side, i.e. of the (log-)Laplace transform of the fluctuations.

\subsection{Expansion of the partition function}
\label{sec:partitionfunctionsexpansions}
\coT{We first recall the expansions of the partition function $\ZNbetaV$ (as in \eqref{def:PNbeta}) obtained in \cite{ls1,loiloc}. They improve on previously known estimates up to order $o(N \log N)$, by giving some explicit expression for the next order in terms of the free energy functional $\fbeta$ introduced in \cite{ls1}. This functional depends on the equilibrium measure $\mueq$ in a simple way, thanks to the logarithmic nature of the interaction, and it allows us to obtain relative expansions up to an error term $r_N = o(N \L_N^2)$. Since the allowed errors in the computation of Laplace transforms are $o(1)$, these expansions are not precise enough to yield the CLT directly.}

If $\mu$ is a probability density, we denote by $\Ent(\mu)$ the entropy\footnote{\comTT{Our choice of $\Ent(\mu)$ is really the opposite of the physical entropy.}} of $\mu$ given by 
$$\Ent(\mu) : = \int_{\R^2} \mu \log \mu.$$
The following asymptotic expansion is proven in  \cite[Corollary 1.5]{ls1}, see also \cite[Remark 4.3]{ls1}. 
\begin{prop}[Partition function expansion: macroscopic case] \label{lem:comparaisonmacro} Let $\mu$ be a probability density on $\R^2$ supported in $\Sigma$. Assume that $\mu$ is $C^{0,1}$ and bounded below on $\Sigma$ and that $\Sigma$ is a finite union of compact connected sets with $C^{0,1}$ boundary. Let $\zeta$ be some Lipschitz function on $\R^2$ satisfying 
$$
\zeta = 0 \text{ in } \Sigma, \quad \zeta > 0 \text{ in } \R^2 \setminus \Sigma,\quad  \int_{\R^2} \exp\left(-\beta N \zeta(x)\right) dx<\infty \ \text{for $N$ large enough}.
$$ 
Then, with the notation of \eqref{def:KNbeta} and for some $C_\beta$ depending only on $\beta$, we have
\begin{equation}\label{comparaisonmacro}
\log \KNbeta(\mu, \zeta) = \frac{\beta}{4} N\log N + C_\beta N- N \left(1 - \frac{\beta}{4}\right) \Ent(\mu)  + N r_N,
\end{equation} 
with $r_N$ satisfying \begin{equation} \label{rNnegligible}
\lim_{N\to \infty } r_N=0.
\end{equation}
\end{prop}
\comTT{The constant $C_{\beta}$ is obtained by minimizing the free energy functional introduced in \cite{ls1}, and is not purely of entropic origin. Its precise value,  however,  will not matter as we examine only {\it differences} of logarithms of partition functions.} 

As stated, the error term $N r_N$ depends on $\mu$. In fact, as explained in Section \ref{sec:comparaisonmieux}, the convergence $r_N \to 0$ is uniform on certain subsets of probability densities. In particular, if $\{\mu_s\}_{s}$ is a family of  probability densities whose supports are contained in a uniform compact, whose densities are uniformly bounded by a positive constant on their respective supports and whose $C^{0,1}$ norms are uniformly bounded, then the convergence \eqref{rNnegligible} is uniform in $s$. This is the case for the family $\{\tmut\}_{t\in [0,1]}$ of approximate equilibrium measures constructed below in Section \ref{sec:transport}. Proposition \ref{lem:comparaisonmacro} will be used in the macroscopic cases (interior and boundary).

In the mesoscopic case, we will instead use a \textit{relative} expansion of the partition function, as follows.
\begin{prop}[Partition function expansion: mesoscopic case] \label{prop:comparaisonmeso}
For $N \geq 1$ and $t$ as in \eqref{tmax}, with $\mu_t = \bar{\mu}_t$ as in \eqref{def:mut}, we have
\begin{equation} \label{comparaisonmeso}
\log \KNbeta(\mueqt, \zeta_t) - \log \KNbeta(\mueq, \zeta_0) =  N \left(1 - \frac{\beta}{4}\right) \left(\Ent(\mueq) - \Ent(\mueqt)\right) + N \L_N^2 r_N,
\end{equation}
with $\lim_{N\to \infty} r_N=0$ uniformly for $|t| \leq \tmax$. 
\end{prop}
\begin{proof}
The expansion \eqref{comparaisonmeso} follows from the analysis of \cite{loiloc}, however some inspection of the proof is needed. Following the argument of \cite{loiloc}, we may split the energy between an interior part (corresponding to $U_N$, the support of $\xi_N$) and an “exterior part” ($\R^2 \backslash U_N$). Since the background measures $\mueq$ and $\mueqt$ coincide on the exterior part, the ratio of the associated partition functions depends only on the interior part (the partition function associated to the interior part is denoted by $K_{N,z, \delta_1}^{\beta}$ in \cite{loiloc}). Using the limit \cite[Eq. (6.3)]{loiloc} and the scaling properties of the functional $\fbeta$ defined there we may then recover \eqref{comparaisonmeso}. We refer to Section \ref{sec:comparaisonmieux} for more detail.
\end{proof}

\subsection{Exponential moments of the energy and concentration bound}
\coT{Many error terms in our computations involve the energy of a given configuration, and since we work with Laplace transforms, we will need to control the exponential moments of the energy.}
\begin{lem} \label{lem:expmoment}
Let $(\mu, \zeta)$ satisfy the assumptions of Proposition \ref{lem:comparaisonmacro}.   For any  $\vec{\eta}$ such that $\eta_i \leq N^{-1/2}$ for $i = 1, \dots, N$, we have
\begin{equation} \label{controleexpmoment}
\left|\log \Esp_{\PNbetamuzeta}\left( \exp\left(\frac{\beta}{4} \left( 
\EnergieLoc_{\vec{\eta}}(\XN) \right)\right) \right)  \right| \le C N \L_N^2
\end{equation}
with $\EnergieLoc$ as defined in \eqref{deflocen}, 
\ed{
\begin{equation}\label{expmomzeta}\left|\log \Esp_{\PNbetamuzeta}\left( \exp\left(\beta N\sum_{i=1}^N \zeta(x_i)\right)\right)\right|\le C N,
\end{equation}}
\comTT{where $C$ depends only on $\beta$ and $V$.}
\end{lem}
\begin{proof}
In the macroscopic case, i.e. $\L_N=1$, it follows from \eqref{fnmeta}, \eqref{ecartaeta} that
$$
\frac{1}{2\pi} \int_{\R^2} |\nab H_{N,\veta}^{\mu}|^2 + \sum_{i = 1}^N \log (\eta_i N^{1/2}) = \left( \FN(\XN, \mu) + \frac{N \log N}{2} \right) + O(N)\|\mu\|_{L^{\infty}}, 
$$
and it remains to control the exponential moments of $\FN(\XN, \mu) + \frac{N \log N}{2}$ under $\PNbetamuzeta$.  This follows e.g. from the analysis of \cite{SS2d}, but we can also deduce it from Proposition \ref{lem:comparaisonmacro}.  We may indeed write   
\begin{multline} \label{deuxzeta}
\Esp_{\PNbetamuzeta} \left[ \exp\left(\frac{\beta}{4} F_N(\XN, \mu) \right)\right] \\ = \frac{1}{\KNbeta(\mu, \zeta)} \int \exp \left( -\frac{\beta}{4} \left( F_N(\XN, \mu) + 2N \sum_{i=1}^N 2 \zeta(x_i)\right) \right) d\XN 
\\
=  \frac{K_{N, \frac{\beta}{2}}(\mu, 2\zeta)}{\KNbeta(\mu, \zeta)}.
\end{multline} 
Taking the $\log$ and using \eqref{comparaisonmacro} to expand both \ed{right-hand side terms up to order $N$ yields that 
\begin{equation}\label{expmF}
\left| \log \Esp_{\PNbetamuzeta} \left[ \exp\left(\frac{\beta}{4}( F_N(\XN, \mu) + \hal N \log N )+\beta N \sum_{i=1}^N \zeta(x_i)\right)\right]\right| \le C N,
\end{equation}
and using  $\zeta\ge 0$ and Corollary \ref{corominoe},  the result follows. The constant $C$ is uniform for the cases mentioned after Proposition \ref{lem:comparaisonmacro}.}

In the mesoscopic case, \eqref{controleexpmoment} follows from the \coT{“good control on the energy” as stated in} \cite[Section 4.4]{loiloc}, up to replacing the truncation at fixed distance $\eta$ by the truncation used here, which depends on the points.
\end{proof}

Similarly, we can bound the exponential moments of the \textit{electric energy} \coT{(the $L^2$, positive part of the energy)}.
\begin{lem} \label{lem:expmoment2}
Let $(\mu, \zeta)$ satisfy the assumptions of Proposition \ref{lem:comparaisonmacro}, and let $s \in (0,1)$. We have
\begin{equation} \label{controleexpmoment2}
\left|\log \Esp_{\PNbetamuzeta}\left(\exp\( \frac{\beta}{4} \int_{B(\bar x_N,2\L_N)} |\nab H_{N,s\vecr}^{\mu}|^2  \)\right)\right| \le C N \L_N^2(1 +|\log s|),
\end{equation}
\comTT{with $C$ depending only on $\beta$ and $V$.}
\end{lem}
\begin{proof}[Proof of Lemma \ref{lem:expmoment2}]
\coT{Using the result of Lemma \ref{lem:expmoment} with $\vec{\eta} = s \vec{\rr}$ we obtain
$$
\left|\log \Esp_{\PNbetamuzeta}\left( \exp\left(\frac{\beta}{4} \left( 
\EnergieLoc_{s\vec{\rr}}(\XN) \right)\right) \right)  \right| \le C N \L_N^2, 
$$
where $\EnergieLoc_{s\vec{\rr}}(\XN)$ has been defined in \eqref{deflocen} and is equal to
$$
\EnergieLoc_{s\vec{\rr}}(\XN) = \int_{B(\bar x_N,2\L_N)} |\nab H_{N,s{\vec{\rr}}}^{\mu}|^2 - 2\pi \sum_{i, x_i \in B(\bar{x}_N,2\L_N) } -\log |s \rr(x_i) N^{1/2} |.
$$}
Using \eqref{contrnbpoints} to control the number of points in $B(\bar x_N, 2 \L_N)$, and the Cauchy-Schwarz's inequality, we see that 
\begin{multline*}
\log \Esp_{\PNbetamuzeta}\left(\exp\left( \frac{\beta}{4}\int_{B(\bar x_N,2\L_N)} |\nab H_{N,s\vecr}^{\mu}|^2 \right)  \right)  \leq C N \L_N^2 ( 1+| \log s|) \\ 
+ C \log \Esp_{\PNbetamuzeta}\left( \exp\left( - \frac{\beta}{4} \sum_{i, x_i \in B(\bar x_N,2\L_N)}  \log (\rr(x_i) N^{1/2} ) \right)\right).
\end{multline*}
Since $\rr(x_i)$ is the nearest-neighbor distance, we expect it to be of order $N^{-1/2}$ and thus the second-term in the right-hand side should be of the same order as the number of points in $B(\bar x_N,2\L_N)$, namely $O(N \L_N^2)$ according to \eqref{contrnbpoints}.
In the macroscopic case, we may use \eqref{fnmeta2} with $\eta_i =N^{-1/2}$ for all $i$, this yields
$$- \sum_{i\neq j}  \log ( |x_i-x_j| N^{1/2})   \le  F_N(\XN, \mu)+ \hal N \log N + CN.$$
Since the left-hand side is bounded below by $-\sum_i \log (\rr(x_i) N^{1/2})- CN$, in view of \eqref{expmF}, we conclude in this case. In the mesoscopic case, we  may combine the result of Lemmas \ref{lem:contrdist} and \ref{lem:expmoment} to conclude as well.
\end{proof}

\comTT{Combined with Proposition \ref{prop:fluctenergy} or  \eqref{controlfluctuations2} applied with $\vec{\eta}= \vecr$, this result directly implies a first concentration bound on the fluctuations.}

\begin{coro}
\comTT{
Assume  $\xi$ is a  Lipschitz function supported in a ball $\D(\bar{x_N}, \L_N)$. Then
\begin{equation}
\left|\log \Esp_{\PNbeta}  \(\exp (\Fluct_N(\xi) )\)\right|\le C \|\nab \xi\|_{L^2(\D(\bar{x_N}, \L_N))} \sqrt{N} \L_N+ C \|\nab \xi\|_{L^\infty} \(1+ \sqrt{N} \L_N^2 \|\mueq\|_{L^\infty}\)
 \end{equation}
with $C$ depending only  on $\beta$ and $V$.}
\end{coro}
\comTT{Of course, this bound is less precise than that provided by the main theorem, however it is valid for much less regular  functions $\xi$.}

\section{Approximate transport and approximation error}
\label{sec:transport}

\def \ttmax{\tilde{t}_{\mathrm{max}}}

In this section, we again work in arbitrary dimension $\d$. We are looking to define \coT{a transport map $\phi_t := \id + \frac{t}{\beta} \psi$, for some $\psi$, and }
\ed{an approximate equilibrium measure 
\begin{equation}\label{tildemueqt}
\tilde \mueqt:= \left(\id + \frac{t}{\beta} \psi\right) \# \mueq,
\end{equation}
\coT{as well as an associated approximate confining term $\tilde{\zeta}_t$.} The goal of this section is to show that if we choose $\psi$ properly, then in the evaluation of the right-hand side of \eqref{LTtto0deuxieme} we may replace $\mueqt$ by $\tilde{\mu}_t$, while making only a small error as shown by the following result and its corollary.

\coT{The quantity $\tmax$ has been introduced in \eqref{tmax}, and we define $\ttmax$ as
\begin{equation}
\label{ttmax} 
\ttmax := \beta \left( 2 \|\psi\|_{C^{0,1}(\R^2)} \right)^{-1}.
\end{equation}}

\begin{prop} \label{pro36}Denote by $U_N$ an open cube containing the support of $\xi_N$ (in the mesoscopic case, of sidelength $O(\L_N)$). Let $N \geq 1$ and $\tau$ be such that $\frac{\tau}{N} \leq \min(\tmax, \ttmax)$. We may choose $\psi$ in such a way that with the definition \eqref{tildemueqt} we have
\begin{multline}\label{424}
\log \frac{\KNbeta(\tilde{\mu}_{\tau/N}  , \tilde{\zeta}_{\tau/N} )}{\KNbeta(\mu_{\tau /N }, \zeta_{\tau/N} )}  =
O\left( \tau^4 N^{-2}| U_N|(\diam \, U_N)^2M_{\xi}^2\)
\\ + O\left(\tau^2 N^{-1/2} (\diam \, U_N) |U_N| M_{\xi} \right)   + O \(  \frac{ \tau^2}{N^{1/2}} +  \frac{\tau^4}{N^2}  \) \indic_{\partial}.
\end{multline}
\coT{where $\indic_{\partial}$ is $1$ in the boundary case and $0$ otherwise,} and
\begin{equation}
\label{def:Mxi}
M_{\xi}:=  \|\xi_N\|_{C^{1,1}} + \|\xi_N\|_{C^{1,1}}^2 + \|\xi_N\|_{C^{2,1}} \|\xi_N\|_{C^{0,1}}.
\end{equation}
In the interior cases, the implicit constants depend only on $V$, and in the boundary case they depend on $V$ and $\xi$.
\end{prop}

\begin{coro} \label{coro:tmuetmut} For any fixed $\tau$, we have
\begin{equation}
\log \frac{\KNbeta(\tilde{\mu}_{\tau/N}  , \tilde{\zeta}_{\tau/N} )}{\KNbeta(\mu_{\tau /N }, \zeta_{\tau/N} )}  = o_N(1),
\end{equation}
and the convergence is uniform for $\tau$ bounded.
\end{coro}}

\subsection{Approximate transport and equilibrium measure} \label{sec:defpsi}
\setcounter{step}{0}

\ed{We first define $\psi$ and then evaluate the errors made by replacing $\mueqt$ by the approximate equilibrium measure \eqref{tildemueqt}.
The idea is to pick $\psi$   such that $\id + \frac{t}{\beta} \psi$ approximately  solves (i.e. solves to order $t$) the equation
$$
\det (D(\id + \frac{t}{\beta} \psi)) \approx \frac{\mueq}{\mueqt \circ (\id + \frac{t}{\beta} \psi)},
$$}
\coT{which expresses the fact that $\id + \frac{t}{\beta} \psi$ approximately transports $\mueq$ on $\mueqt$.}

\begin{step}[The interior cases]
In all interior  cases (mesoscopic and macroscopic), we set
\begin{equation}\label{defpsimeso}
\psi= - \frac{\nab \xi_N}{\cd \mueq}
\end{equation}
which is well-defined in view of the lower bound for $\mueq$ provided by \ref{H3}. This way $\psi$ is compactly supported in $U_N$ and solves
$$ 
\div (\mueq \psi)= - \frac{1}{\cd}\Delta \xi_N,
$$
and we have the obvious bounds, 
\begin{equation}
\label{regpsimeso} 
|\psi|_{k}   \leq C_k |\xi_N|_{k+1}  \leq C |\xi|_{k+1} \frac{1}{\L_N^{k+1}} \text{ for } k = 0, 1, 2.
\end{equation}
The constant $C_k$ in \eqref{regpsimeso} depends  on $|\mueq|_{k}$ and on $\min_{\Sigma} \mu_0$, hence on $V$.
\end{step}

\begin{step}[Boundary case]
The boundary case is significantly more difficult than the interior one, since it requires to understand the dependence on $t$ of the support $\Sigma_t$, for $t$ small. This is provided by the analysis of \cite{SerSer}, and here we quote the results that we will need (which are valid in all dimensions), with the notation of this paper. It  follows from the main theorem in \cite{SerSer} after noting that $V^0$ there corresponds to $\frac{1}{\beta}(\xi - \xi^\Sigma)$ for us and that $-\Delta h^0$ there is equal to $\cd \mueq$ here. We recall that $\xi^\Sigma$ denotes the harmonic extension of $\xi$ outside $\Sigma$ and $[\cdot]$ denotes the jump across the interface $\Sigma$.
\begin{prop}
\label{proserser}
Under our assumptions the following holds. 
Let $\mueqt$ be the equilibrium measure associated to $V- \frac{2t}{\beta}\xi$ as in Definition \ref{def:mueqt}, and let $\Sigma_t$ be its support. Letting $\vec{n}$ denote the outer unit normal vector to $\partial \Sigma$, there exists a family $\{\rho_t\}_{t}$ of maps from $\p \Sigma $ to $\R^\d$ such that, for $|t|$ small enough \ed{depending on $\d, \alpha, V, \Sigma$}, we have
\begin{equation}
\label{bst} \partial \Sigma_t = \left\lbrace x +\frac{ t}{\beta} \rho_t(x) \vec{n}(x), x \in \partial \Sigma \right\rbrace,
\end{equation}
with 
\begin{equation}\label{esq}
\|\rho_t - \rho\|_{L^\infty(\p \Sigma)} \le Ct,
 \end{equation}  where 
 \begin{equation}\label{rho}
 \rho(x) =  \frac{2}{\Delta V(x)} [\nab \xi^\Sigma]\cdot \vec{n}.
 \end{equation}
 Finally, we also have 
 \begin{equation}\label{zetateq}
 \|\zeta_t-\zeta_0- \frac{t}{\beta} (\xi-\xi^\Sigma) \|_{L^\infty(\R^\d)}\le Ct^2, \qquad   \|\zeta_t-\zeta_0-  \frac{t}{\beta} (\xi-\xi^\Sigma) \|_{C^{0,1} (\R^\d\backslash \Sigma_t\cup \Sigma)  }\le Ct^2.
 \end{equation}
Here, the constants  $C$ depend on the $C^{3,1}$ norms of $V$ and $\xi$ and the lower bound in \ref{H3}.
\end{prop}

Based on this result, and recalling that $\mueq(x)= \frac{\Delta V(x)}{2\cd}$ on its support, we choose to construct $\psi$ as follows. 
\begin{lem} 
There exists a map $\psi: U \to \R^\d$ where $U$ is a  neighborhood of $\Sigma$ in which $\Delta V>0$, such that 
\begin{equation}\label{defpsi}
\left\lbrace
\begin{array}{ll}
\div (\mueq \psi)= - \frac{1}{\cd} \Delta \xi& \text{in } \ \Sigma\\ [2mm]
\psi \cdot \vec{n}=\dis\frac{1}{\cd \mueq} [\nab \xi^\Sigma]\cdot \vec{n}& \text{on} \ \partial \Sigma.
\end{array}
\right.
\end{equation}
and 
\begin{equation}\label{psiext}
\psi= (\xi^{\Sigma}-\xi)\frac{\nab\zeta_0}{|\nab \zeta_0|^2} + \psi^\perp\quad \text{in} \ U\backslash  \Sigma ,\end{equation}
for some vector field $\psi^\perp$ in $C^{1,1}(U\backslash \Sigma)$  and perpendicular to $\nab \zeta_0$.
 Moreover, $\psi$  is globally  Lipschitz and satisfies 
\begin{equation}\label{regpsi}
\|\psi\|_{C^{1,1}(\Sigma)}+ \|\psi\|_{C^{1,1} (U\backslash \Sigma)} \le C \|\xi\|_{C^{2,1}}
\end{equation}
and 
\begin{equation}\label{eqdiv}
\div (\mueq \psi \indic_{\Sigma}) = -\frac{1}{\cd}\Delta \xi^\Sigma\quad \text{in} \ U.
\end{equation}
in the sense of distributions.
\end{lem}
\begin{proof}
First, we note that by assumption \ref{H3} of $V$ and standard results  in the analysis of the obstacle problem (see for instance \cite{MR1658612} or also \cite[Sec. 3]{SerSer}), we know that, as $x \to \partial \Sigma$ from the outside, 
\begin{align} 
\label{bzetaA} & \zeta_0(x) \sim \frac{\Delta V(x)}{4} \dist (x, \partial \Sigma)^2 \\
 \label{bzetaB} & \nab \zeta_0(x) \sim \frac{\Delta V(x)}{2} \dist(x, \partial \Sigma) \vec{n}.
\end{align}
Thus, if the neighborhood $U$ is chosen small enough, $\nab \zeta_0$  does not vanish in $U\backslash \Sigma$ and  \eqref{psiext} makes sense. 

Let us solve on each connected component $\Sigma_i$ of $\Sigma$
\begin{equation}
\label{oneachconnected}
\left\lbrace
\begin{array}{ll}
\div (\mueq \nab u_i)= -\frac{1}{\cd} \Delta \xi& \text{in } \ \Sigma_i\\ [2mm]
\displaystyle\frac{\partial u_i}{\partial n}=\dis\frac{1}{\cd\mueq} [\nab \xi^\Sigma]\cdot \vec{n}& \text{on} \ \partial \Sigma_i.
\end{array}\right.
\end{equation}
This is an elliptic PDE in divergence form with regular and bounded below coefficient $\mueq$, which is solvable since 
the compatibility condition holds, following \eqref{condit}
\ed{$$
-\int_{\Sigma_i} \Delta \xi- \int_{\partial \Sigma_i} [\nab \xi^\Sigma] \cdot \vec{n} = \int_{\partial \Sigma_i}\frac{\partial \xi^\Sigma}{\partial n}=0.
$$}
Elliptic regularity estimates \cite[Theorems 6.30, 6.31]{gilbarg2015elliptic}
yield that 
$$\|u_i\|_{C^{2,1}(\Sigma_i)}\le C \|\xi\|_{C^{2,1}}.$$
Indeed, $\xi \in C^{2,1}$ and $\xi^{\Sigma}$ is $C^{2,1}$ in the exterior of $\Sigma$.

We then let $\psi= \nab u_i$ in each $\Sigma_i$. 
It satisfies \eqref{defpsi} because of \eqref{oneachconnected}. In addition, in view of \eqref{bzetaB}, we have 
\begin{equation}
 \frac{(\xi^{\Sigma}-\xi)\nab\zeta_0}{|\nab \zeta_0|^2}-\frac{2}{\Delta V(x)} \left[\frac{\partial \xi^\Sigma}{\partial n}\right] \vec{n}\to 0 \quad \text{as} \ x\to \partial\Sigma\,  \text{ from the outer side},
\end{equation}
hence $\psi\cdot \vec{n}$ is built continuous across $\partial \Sigma$. To make $\psi$ itself continuous, we consider the trace of $\psi- (\psi \cdot \vec{n} )\vec{n}$  on $\partial \Sigma$ from the inside, and extend it arbitrarily to a regular vector field. We may then subtract off the projection of that vector field onto $\nab \zeta_0$ to obtain a vector field $\psi^\perp$ which remains perpendicular to $\nab \zeta_0$.
The relation \eqref{eqdiv} can easily be checked by using test-functions.
\end{proof}
\end{step}

We define $\psi$ on the whole $\R^\d$ by multiplying it by a cutoff $\chi$ equal to $1$ in a neighborhood of $\Sigma$.
{\bf In the sequel, we will write $\|\psi\|_{C^{1,1}}$ to mean $\|\psi\|_{C^{1,1}(\Sigma)}+ \|\psi\|_{C^{1,1}(\R^\d\backslash \Sigma)}$.} We now introduce several “approximate” quantities defined with $\psi$.

\begin{defi} \label{def:approxquant}
Let $\psi$ be as above. For $|t| \leq \ttmax$ as in \eqref{ttmax}:
\begin{itemize}
\item We let $\phi_t$ be the approximate transport, defined by $\phi_t := \id +\frac{ t}{\beta} \psi$
\item We let $\tmut$ be the approximate equilibrium measure, defined by $\tmut := \phi_t \# \mueq$.
\item We let $\tzetat$ be the approximate confining potential $\tzetat := \zeta_0 \circ \phi_t^{-1}$.
\item We let $\PNbeta^{(t)}$ be the approximate Gibbs measure
\begin{equation} \label{def:approxGibbsmeas}
d\PNbeta^{(t)}(\XN) = \frac{1}{\KNbeta(\tmut, \tzetat)} \exp\left(- \frac{\beta}{2} \left( \FN(\XN, \tmut) + 2N \sum_{i=1}^N \tzetat(x_i) \right) \right) d\XN,
\end{equation}
which corresponds to $= d\PNbeta^{(\tmut, \tzetat)}(\XN)$ with the notation of \eqref{def:PNbetamuzeta}.
\end{itemize}
\end{defi}
We may check that if $|t|$ is smaller than $\ttmax$, the map $\phi_t$ is a $C^{0,1}$-diffeomorphism on $\R^2$. In view of \eqref{regpsimeso}, the order of magnitude of $\ttmax$ is (as for $\tmax$)
$$
\ttmax \approx \L_N^2.
$$

\subsection{Comparison of partition functions}
\ed{
In order to prove Proposition~\ref{pro36}, the main point is to show  that the energies of a given configuration computed with respect to a background $\mueqt$ or $\tmut$ are typically very close to each other, as stated in the following}
\begin{lem} \label{lemmutomut}
For any $\XN$, we have 
 \begin{equation} \label{mutomut}
\left| \left(F_N(\XN, \mueqt)+ 2N \sum_{i=1}^N \zeta_t (x_i)\right) -\left(F_N(\XN, \tmut)+2N \sum_{i=1}^N\tilde \zeta_t(x_i) \right)\right|  \le \Error_1
\end{equation}
with an $\Error_1$ term bounded by \ed{
\begin{multline}
\label{err1}
|\Error_1|\le  Ct^4 N^2 (\diam\, U_N)^2   |U_N|  M_{\xi}^2 \\
+ C N \left( ( |U_N|^{1/2} + |U_N|^{\frac{\d-2}{2\d} }N^{-\frac{1}{\d}})\|\nab H_{N,\vec{\rr}}^{\tilde \mueqt}\|_{L^2(U_N)} + N^{1-\frac{1}{\d}}|U_N|\|\tmut\|_{L^\infty(U_N)}    \right)  (\diam\, U_N)    t^2 M_{\xi}
\\ +  C\( N t^2\(N^{-\frac{1}{6}}  \|\nab H_{N,\vec{\eta}}^{\tilde \mueqt}\|_{L^2}    + N^{\frac{1}{3}} + N^{\frac{2}{3}}\sum_{i=1}^N \zeta_t(x_i) \) + N^2 t^4\)\indic_{\partial},
\end{multline} where $\indic_{\partial }$ is zero in the interior cases and $1$ otherwise,   and $C$ depends only on $V$.}
\end{lem}
\coT{ We recall that the error term $M_{\xi}$ was introduced in \eqref{def:Mxi} as 
$$
M_{\xi}:=  \|\xi_N\|_{C^{1,1}} + \|\xi_N\|_{C^{1,1}}^2 + \|\xi_N\|_{C^{2,1}} \|\xi_N\|_{C^{0,1}}.
$$}
In particular, in our setting with $\d=2$, $U_N$ of size $\L_N$ and $M_{\xi}$ of size $\L_N^{-4}$, we obtain an error term in \eqref{mutomut} as follows
\begin{multline} \label{err1Explicite}
|\Error_1| \preceq t^4 N^2 \L_N^{-4} + N \L_N^{-2} t^2 \|\nab H_{N,\vec{\rr}}^{\tilde \mueqt}\|_{L^2(U_N)} + N^{3/2} \L_N^{-1} t^2 \\
+ \indic_{\partial}\( N t^2\(N^{-\frac{1}{6}}  \|\nab H_{N,\vec{\eta}}^{\tilde \mueqt}\|_{L^2}    + N^{\frac{1}{3}} + N^{\frac{2}{3}}\sum_{i=1}^N \zeta_t(x_i) \) + N^2 t^4\).
\end{multline}

Lemma \ref{lemmutomut} is proven in Section \ref{sec:preuvelemmutomut}. Now, we explain how to prove Proposition \ref{pro36} using Lemma \ref{lemmutomut}.

\ed{
\begin{proof}[Proof of Proposition \ref{pro36}] Assuming Lemma \ref{lemmutomut}, letting  $t$ be such that $|t| \leq \min(\tmax, \ttmax)$, by definition we may write the ratio of the partition functions as
\begin{multline}
\label{423}
\frac{\KNbeta(\mu_t, \zeta_t)}{\KNbeta(\tilde{\mu}_t, \tzetat)} 
\\ =\Esp_{\PNbeta^{(t)}} \left( \exp \left(-\frac{\beta}{2}\left(F_N(\XN, \mueqt)- F_N(\XN, \tilde{\mu}_t) + 2 N \sum_{i=1}^N  \left(\zeta_t(x_i) - \tzetat (x_i)\right)  \right) \right)\right).
\end{multline}
Setting $t=\tau /N$ and remembering that $N^{-1/2}\ll \L_N$,  we then combine \eqref{err1Explicite},  Lemma  \ref{lem:expmoment2} to control the exponential moments of $\|\nab H_{N,\vec{\rr}}^{\tmut}\|_{L^2}$, \eqref{expmomzeta} to control those of $\sum \zeta_t(x_i)$, and we obtain \eqref{424}.
\end{proof}

\subsection{Approximation error}
     In order to prove Lemma \ref{lemmutomut}, we show that the various quantities involved in the energy computation are close for $\mueqt$ and $\tmut$, as stated in the following.}

\begin{lem}\label{nabht} 
Let $|t| \leq \min(\tmax, \ttmax)$. In the interior cases, we have  
\begin{align}
\label{nabhntm} & \|\nab (h^{\tilde \mueqt}   - h^{\mueqt})\|_{L^\infty} \le C(\diam\, U_N)    M_{\xi} t^2, \\
\label{nabh2}  & \int_{\R^\d} |\nab h^{\tilde \mueqt - \mueqt}|^2 \le    C (\diam\, U_N)^2   |U_N| M_{\xi}^2 t^4 , 
\end{align}
where the constant $C$ depends only on $\|\mueq\|_{C^{0,1}} $ and  $\min_{\Sigma_0} \mueq$.

In the macroscopic boundary case, we have
\begin{align} 
\label{nabhtmv} & \|\nab (h^{\tilde \mueqt}   - h^{\mueqt})\|_{L^\infty} \le C t^2, \qquad \int_{\R^\d} |\nab h^{\tilde \mueqt- \mueqt}|^2 \le C t^4 \\
\label{nabzeta} & \|\tilde \zeta_t-\zeta_t\|_{L^\infty} \le  Ct^2,\quad |\tilde \zeta_t(x)- \zeta_t(x)|\le C t^2 \max(\dist(x,\Sigma_t) , t^2), 
\\ 
\label{nabzeta2bis} & 0 \leq \int_{\R^\d} \zeta_t d\tilde \mueqt + \int_{\R^\d} \tilde \zeta_t d\mueqt\le C t^4,
\end{align}
where the constant $C$ depends only on $\|\mueq\|_{C^{0,1}} $ and  $\min_{\Sigma_0} \mueq$  and on $\|\xi\|_{C^{3,1}}$.
\end{lem}
In particular, for $\xi_N$ as in \eqref{def:xiN}, we have $\diam \, U_N \approx \L_N$ and $M_{\xi} \approx \L_N^{-4}$. We thus obtain, in all cases
\begin{equation} \label{nabhtExplicite}
\|\nab (h^{\tilde \mueqt}   - h^{\mueqt})\|_{L^\infty} \le C \L_N^{-3} t^2 \quad \int_{\R^\d} |\nab h^{\tilde \mueqt - \mueqt}|^2 \le C\L_N^{-4} t^4.
\end{equation}

\begin{proof}[Proof in the interior cases]
By Remark \ref{rem:mut} we know  in the interior cases that  if $|t| \leq  \tmax$ then $\mueqt= \mueq-\frac{t}{\cd \beta} \Delta \xi_N$, and in particular $\mueqt$ and $\mueq$ have the same support.  Moreover, by definition of a transport we have
\begin{equation} \label{tmutasttransport}
 \tilde \mueqt= \frac{\mueq\circ \phi_t^{-1}}{\det\left(I+\frac{t}{\beta}D\psi\right)\circ \phi_t^{-1}}.
\end{equation}
A Taylor expansion yields
 \begin{equation}\label{deno}
 \frac{1}{\det\left(I+\frac{t}{\beta}D\psi\right)\circ \phi_t^{-1}}= 1-\frac{t}{\beta} \div \psi+ u
 \end{equation}
 where $u$ can be checked to satisfy
 $$\|u\|_{L^\infty} \le  C t^2 \|\psi\|_{C^{0,1}}^2+ C t^2 \|D\psi\|_{C^{0,1}}  \|\psi\|_{L^\infty},
$$ 
while, thanks to the $C^{1,1}$ regularity of $\mueq$ on its support we have
 \begin{equation}
 \label{Numerator}
 \|\mueq \circ \phi_t^{-1} - \mueq -t  \nab \mueq \cdot \psi\|_{L^\infty} \le t^2 \|\mueq\|_{C^{1,1}(\Sigma)} \|\psi\|_{C^{0,1}}^2.
 \end{equation}
Finally, using that by definition of $\psi$ we have
\begin{equation}
\label{usefulABC}
\div (\mueq \psi)= - \frac{1}{\cd}\Delta \xi_N \quad \text{in} \ \Sigma,
\end{equation}
 we conclude, combining \eqref{deno}, \eqref{Numerator} and \eqref{usefulABC} that 
 \begin{equation}\label{eq3}
 \|\tilde \mueqt- \mueq + \frac{t}{\cd\beta}\Delta \xi_N  \|_{L^\infty(U_N )} \le  C   t^2\left(  \|\psi\|_{C^{0,1}}+ \|\psi\|_{C^{0,1}}^2 + \|\psi\|_{C^{1,1}} \|\psi\|_{L^\infty}\right) ,
\end{equation}
with $C$ depending only on  $\|V\|_{C^{3,1}}$. We deduce that 
\begin{align*}
& \|h^{\tilde \mueqt}   - h^{\mueqt}\|_{L^\infty} \le C(\diam\, U_N)    t^2\left(  \|\psi\|_{C^{0,1}}+ \|\psi\|_{C^{0,1}}^2 + \|\psi\|_{C^{1,1}} \|\psi\|_{L^\infty}\right), \\
& \|\nab ( h^{\tilde \mueqt}   - h^{\mueqt})\|_{L^\infty} \le C(\diam\, U_N)    t^2\left(  \|\psi\|_{C^{0,1}}+ \|\psi\|_{C^{0,1}}^2 + \|\psi\|_{C^{1,1}} \|\psi\|_{L^\infty}\right).
\end{align*}
Indeed, if $\diam \, U_N$ is of order $1$  then it follows by standard elliptic regularity estimates, and if not it can be deduced by rescaling by $\L_N$.

Consequently, we may write, integrating by parts and using the Cauchy-Schwarz inequality
\begin{multline*}
\int_{\R^\d} |\nab (h^{\tilde \mueqt}   - h^{\mueqt})|^2 =\cd \int_{\R^\d} (h^{\tilde \mueqt}   - h^{\mueqt})(\tilde \mueqt -\mueqt) \\
\le C (\diam\, U_N)^2   |U_N| t^4\left(  \|\psi\|_{C^{0,1}}+ \|\psi\|_{C^{0,1}}^2 + \|\psi\|_{C^{1,1}} \|\psi\|_{L^\infty}\right)^2 .
\end{multline*} Combining with \eqref{regpsimeso}, we have obtained
  \eqref{nabhntm}, \eqref{nabh2}.\end{proof}

\begin{proof}[Proof in the boundary case]
 First, a Taylor expansion as in \eqref{tmutasttransport} yields again that 
 \begin{equation}\label{tmut}
 \tilde \mueqt= \left(\frac{1}{2\cd}\Delta V - \frac{t}{\beta\cd}\Delta \xi + u\right) \indic_{\phi_t(\Sigma)}
 \end{equation}
 with $u$ such that 
 $$
 \|u\|_{L^\infty} \le C t^2 \left( \|\psi\|_{C^{0,1}}^2+ \|\psi\|_{C^{1,1}}  \|\psi\|_{L^\infty}\right).
 $$
 
On the other hand, thanks to the result of Proposition \ref{proserser}, we know that 
$$
\mueqt = (\frac{1}{2\cd} \Delta V- \frac{t}{\beta \cd }\Delta \xi ) \indic_{\Sigma_t},
$$
with $\phi_t(\Sigma)$ and $\Sigma_t$ at Hausdorff distance $O(t^2)$ from each other. It implies, testing against a smooth test-function for instance,  that $\frac{\tilde \mueqt -\mueqt}{t^2}$  behaves asymptotically as $t \to 0$ like the uniform measure on $\p \Sigma$ multiplied by a bounded function, plus a bounded function in $\Sigma$. Single-layer potential estimates (see for instance \cite[Chap. II]{dautray2012mathematical} or \cite[Theorem A.1]{SerSer}) allow to deduce that 
 $$\|\nab (h^{\tilde \mueqt}-h^{\mueqt})\|_{L^\infty} \le C t^2(1+ \|\xi\|_{C^{2,1}})
 $$ where $C$ depends on $V$,  which yields \eqref{nabhtmv}.

Let us turn to  \eqref{nabzeta}. Using that $\zeta_0$ is as regular as $V$ outside $\Sigma$ i.e. in $C^{3,1}( \R^\d\backslash \Sigma)$  and \eqref{psiext}, we may write that
$$\tilde \zeta_t-\zeta_0=\zeta_0\circ \phi_t^{-1}- \zeta_0=  - \frac{t}{\beta}\nab \zeta_0\cdot \psi +O(t^2)=  -\frac{t}{\beta} (\xi^\Sigma-\xi)+O(t^2) \quad \text{in} \  (\Sigma  \cup \Sigma_t)^c$$ 
with a $O$ in $L^\infty$. Combining with \eqref{zetateq}, we deduce that 
\begin{equation}\label{esex}
\|\tilde \zeta_t-\zeta_t\|_{L^\infty( (\Sigma  \cup \Sigma_t)^c)}\le C t^2.
\end{equation}
In addition, $\zeta_0$ satisfies \eqref{bzetaA}, so by the regularity of $\phi_t$ and definition of $\tilde \zeta_t$, it follows that 
$$|\tilde \zeta_t(x) |\le C\dist(x, \p( \phi_t(\Sigma)))^2$$
while the analogue of \eqref{bzetaA} holds for $\zeta_t$, i.e.
\begin{equation}
\label{bzetat}
\zeta_t(x) \sim \( \frac{1}{4} \Delta V(x) - \frac{t}{2\beta} \Delta \xi\) \dist (x, \p \Sigma_t)^2 \quad \text{as} \ x\to \p \Sigma_t.
\end{equation}
 Since $\p( \phi_t(\Sigma))$ and $\p \Sigma_t$ are at distance  $O(t^2)$ from each other  and $\p \Sigma$ and $\p \Sigma_t$ at distance $O(t)$ from each other, both from \eqref{bst}--\eqref{esq}, we deduce that  $|\tilde \zeta_t- \zeta_t|\le C t^2$ in $\Sigma \cup \Sigma_t$ (note that both functions are zero in $\Sigma \cap \Sigma_t$). Combining with \eqref{esex}, the first relation in  \eqref{nabzeta} follows.  The second relation is based on the comparison of \eqref{bzetat} and \eqref{bzetaA}, and the fact that $\Sigma_t$ and $\phi_t(\Sigma)$ are at distance $t^2$ from each other. The third relation follows from the same facts as well.
   \end{proof}
 
\subsection{Energy comparison: from $\mueqt$ to $\tmut$}
\label{sec:preuvelemmutomut}
\ed{We may now conclude with the }
\begin{proof}[Proof of Lemma \ref{lemmutomut}]
Returning to the definition \eqref{def:FN}, we may write 
 \begin{multline}\label{partF}
F_N(\XN, \mueqt)-F_N(\XN, \tmut)=
N^2 \int_{\R^\d\times \R^\d} \g(x-y)\, d\left(\tilde{\mueqt}- \mueqt\right)(x) d\left( \tilde{\mueqt}-\mueqt\right) (y)\\ +  2 N\int_{\R^\d\times \R^\d}  \g(x-y) d( \tilde{\mueqt} - \mueqt )(x)\, \left( \sum_{i=1}^N \delta_{x_i} -N d\tilde{\mueqt}\right) (y)\\
= N^2 \int_{\R^\d} |\nab h^{\mueqt-\tilde \mueqt}|^2 + 2 N \int_{\R^\d} h^{\tilde{\mueqt} - \mueqt }  \left( \sum_{i=1}^N \delta_{x_i} -N d\tilde{\mueqt}\right) (y).
\end{multline}
On the other hand,
\begin{multline}\label{partZ}
  \sum_{i=1}^N \left( \zeta_t (x_i) -\tilde \zeta_t(x_i) \right) =  N\int_{\R^\d} (\zeta_t-\tilde \zeta_t)d\tilde{\mueqt} +\int_{\R^\d} (\zeta_t-\tilde \zeta_t)\, \Big( \sum_{i=1}^N \delta_{x_i} -N d\tilde{\mueqt}\Big) \\
 =N  \int_{\R^\d} \zeta_t d\tilde{\mueqt}   +\int_{\R^\d} (\zeta_t-\tilde \zeta_t)\, \Big( \sum_{i=1}^N \delta_{x_i} -N d\tilde{\mueqt}\Big),\end{multline}
 using that $\tilde \zeta_t $ vanishes on the support of $\tilde\mueqt$.
 Combining \eqref{partF} and \eqref{partZ}, we obtain 
  \begin{multline}\label{evalf}
 \left(F_N(\XN, \mueqt)+ 2N \sum_{i=1}^N \zeta_t (x_i)\right) -\left(F_N(\XN, \tmut)+2N \sum_{i=1}^N\tilde \zeta_t(x_i) \right)
 \\
 = N^2 \int_{\R^\d} |\nab h^{\mueqt-\tilde \mueqt}|^2 + 2N^2  \int_{\R^\d} \zeta_t d\tilde{\mueqt} 
 +2N  \int_{\R^\d}( h^{\tilde \mueqt -\mueqt} + \zeta_t - \tilde \zeta_t) \left( \sum_{i=1}^N \delta_{x_i} -N d \tilde{\mueqt} \right).
 \end{multline}
The last term in the right-hand side can be seen as a fluctuation. Using \coT{the a priori bounds on the fluctuations} given by Proposition \ref{prop:fluctenergy} and the control on $\nab h^{\tilde \mueqt-\mueqt}$ given by Lemma \ref{nabht}, we find  
  \begin{multline}\left| \int  h^{\tilde \mueqt -\mueqt}   \Big( \sum_{i=1}^N \delta_{x_i} -N d \tilde{\mueqt}\Big)\right|\\ 
  \le    
C  \left( \left( |U_N|^{1/2} + |U_N|^{\frac{\d-2}{2\d} }N^{-\frac{1}{\d}} \right)\|\nab H_{N,\vec{\rr}}^{\tilde \mueqt}\|_{L^2(U_N)} + N^{1-\frac{1}{\d}}|U_N|\|\tmut\|_{L^\infty(U_N)}    \right) \\
\times (\diam\, U_N)    t^2 M_{\xi}.
\end{multline}  

In the interior cases,  $\tilde \zeta_t=\zeta_t$ and this concludes the evaluation of the last term in \eqref{evalf}.
In the boundary case, \ed{ in view  of \eqref{bzetat} and the fact that $\partial \Sigma_t $ and $\partial\phi_t(\Sigma)$ are $O(t^2)$ apart, we have $|\zeta_t-\tilde\zeta_t|\le  C t^2 r $ for $\dist(x_i,\partial \Sigma_t ) \le r$. \coT{In particular 
$$
\int \left(\zeta_t - \tilde \zeta_t\right) N d \tilde{\mueqt}  = N \int (\zeta_t - 0)  d \tilde{\mueqt} \leq N C t^2 t^2 \leq C Nt^4.
$$}
On the other hand, in view of \eqref{bzetat} again, the number of points $x_i$ such that $\dist(x_i, \partial \Sigma_t) \ge N^{-1/3}$ 
is bounded by $N^{2/3} \sum_{i=1}^N \zeta_t(x_i)$. We may thus write 
 \begin{multline*}\left| \int  (\zeta_t-\tilde \zeta_t)  \Big( \sum_{i=1}^N \delta_{x_i} -N d \tilde{\mueqt}\Big)\right|\\ 
 \le \left(C  t^2 N^{-\frac{1}{3}} \sum_{i, \dist (x_i, \p \Sigma_t) \le N^{-\frac{1}{3}} } 1 \right) + C  t^2 \left(  \sum_{i, \dist (x_i, \p \Sigma_t)\ge N^{-\frac{1}{3}}  } 1 \right)  + CNt^4  \\
 \le  Ct^2 N^{-\frac{1}{3}}  \( N^{\frac{1}{6}}  \|\nab H_{N,\vec{\eta}}^{\tilde \mueqt}\|_{L^2} +N^{\frac{2}{3}}\)  +  C t^2  N^{\frac{2}{3}} \sum_{i=1}^N  \zeta_t(x_i)  + CNt^4 
 \end{multline*}
 where we used 
  \eqref{nbpbord} to control the first sum.
  We are thus led to 
   \begin{equation*}\left| \int  (\zeta_t-\tilde \zeta_t)  \Big( \sum_{i=1}^N \delta_{x_i} -N d\tilde{\mueqt}\Big)\right|
\le C t^2\( N^{-\frac{1}{6}}    \|\nab H_{N,\vec{\eta}}^{\tilde \mueqt}\|_{L^2}    + N^{\frac{1}{3}} + N^{\frac{2}{3}}\sum_{i=1}^N \zeta_t(x_i) \) + CNt^4,
\end{equation*} 
with a constant depending on $V$ and $\|\xi\|_{C^{2,1}}$.
    Combining with \eqref{nabh2}, resp. \eqref{nabhtmv} and \eqref{nabzeta}, the result follows.}
\end{proof}

\subsection{Proof of Corollary \ref{coro:tmuetmut}}
\begin{proof}[Proof of Corollary \ref{coro:tmuetmut}]
We use the fact that $\min(\tmax, \ttmax)$ is of order $\L_N^2$, $M_{\xi}$ is of order $\L_N^{-4}$, $\diam \, U_N$ is of order $\L_N$ and $|U_N|$ of order $\L_N^2$. All the terms in the right-hand side of \eqref{424} tend to $0$ as $N \to \infty$, uniformly for $\tau$ bounded.
\end{proof}

\section{Study of the Laplace transform and conclusion}\label{secani}
As seen in Section \ref{sec2}, computing the Laplace transform of fluctuations amounts to computing the ratio of partition functions associated to the perturbed and unperturbed potentials. According to the results of Section \ref{sec:transport}, one can replace the perturbed equilibrium with the approximate perturbed measure, which is the push-forward of the measure $\mueq$ by $\phi_t$. The ratio of partition functions can then be evaluated by using the transport map $\phi_t$ as a change of variables: 
using the definition \eqref{def:KNbeta} and changing variables by $\Phi_t$ where $\Phi_t(\XN)= (\phi_t (x_1), \dots, \phi_t (x_N))$  we have
\begin{multline} \label{acv}
\KNbeta(\tmut, \tzetat) = \int_{(\R^2)^N} \exp\left(- \frac{\beta}{2} \left( \FN(\XN, \tmut) + 2N \sum_{i=1}^N \tzetat(x_i) \right) \right) d\XN \\
=  \int_{(\R^2)^N} \exp\left(- \frac{\beta}{2} \left( \FN(\Phi_t(\XN), \tmut) + 2N \sum_{i=1}^N \tzetat(\phi_t(x_i)) \right)  + \sum_{i=1}^N \log |\det  D \phi_t(x_i)| \right) d\XN.
\end{multline}
Evaluating the   ratio of the partition  functions thus involves evaluating the (exponential moments of) the difference of the energies
$ \FN(\Phi_t(\XN), \tmut) $ and $\FN(\XN, \mueq)$. 
This is the object of the next two subsections.
\subsection{The anisotropy}
\begin{defi}[Anisotropy]
\label{defani}
Let $\psi$ be a $C^{0,1}$ map from $\R^2$ to $\R^2$.
Let $s \in (0, \hal)$ be a parameter. We define the \textit{anisotropy of $\XN, \mu$ with respect to $\psi$} as the following quantity
\begin{equation}
\label{def:Ani}
\Ani(\psi, \XN, \mu) := \frac{1}{2\pi} \int_{U_N} \langle \nab H_{N, s\vecr}^{\mu},  \mathcal A \nab H_{N,s\vecr}^{\mu} \rangle,
\end{equation}
where $\mathcal A(x)= 2D\psi - (\div \psi) \id$. 
\end{defi}
\comTT{We recall that truncated potentials of the form $H^{\mu}_{N, \vec{\eta}}$ have been defined in \eqref{def:HNmutrun}, here $\vecr$ is as in \eqref{def:trxi} and $s$ is an additional parameter that will eventually be sent to $0$ for technical reasons.}

One may observe that $\mathcal A$ is a trace-free matrix, so we are integrating terms of the form $(\p_1 H)^2-(\p_2 H)^2$ in some moving coordinate frame, hence the term \textit{anisotropy}.

\subsection{Energy comparison along a transport}
We now state a result that allows to linearize the variation of  the energy along a  general transport. This proposition is a crucial step, and its proof is postponed to Appendix \ref{sec:proofcomparaison}. It relies very much on our ``electric" formulation of the energy, which allows to better take advantage of the charge compensations, and also can be generalized to higher dimensions. 

\begin{prop} \label{prop:comparaison2}Let $\mu$ be a probability measure with a bounded density and compact support $\Sigma$ with a $C^1$ boundary.  
Let $\psi$ be a $C^{0,1}(\R^2) \cap C^{1,1}(\Sigma)\cap C^{1,1}   (\R^2 \backslash \Sigma)$ map from $\R^2$ to $\R^2$ (possibly depending on $N$), with $\|\psi\|_{C^{0,1}} \leq \hal$, and assume that there is a union of cubes $U_N$ containing an $N^{-1/2}$-neighborhood of the support of $\psi$.

Let finally $\Phi= \id + \psi$ and $ \nu = \Phi \# \mu$.  Let $s \in (0, \hal)$.
For any $\XN$ in $(\R^2)^N$ we let $I_N$ be the set of indices $i \in \{1, \dots, N\}$ such that $x_i \in U_N$, and $I_{\partial,\psi} $ be the set of indices such that $B(x_i, N^{-1/2})$ intersects both $\partial \Sigma $ and the support of $\psi$.
We also let  $\Phi(\XN) = (\Phi(x_1), \dots, \Phi(x_N))$. 

We have
\begin{equation} \label{transingle}
\FN\left( \Phi (\XN) , \nu \right) - \FN(\XN, \mu) \leq \Ani(\psi, \XN, \mu) + \hal \sum_{i \in I} \div \psi (x_i) + \Error,
\end{equation}
with the error term bounded as follows
\begin{multline}  \label{errortransingle}
|\Error| \preceq
\|  \psi\|_{C^{0,1} }^2  \int_{U_N} |\nab H_{N,s\vecr}^{\mu}|^2
+ s^2 \|\psi\|_{C^{0,1}}  \left(\# I_N +  \int_{U_N} |\nab H_{N, s\vecr}^{\mu}|^2 \right)
\\ + \# I_N \frac{s}{\sqrt{N}} \|\psi\|_{C^{1,1}} + \# I_{\partial,\psi } \,  \|\psi\|_{C^{0,1}}.
\end{multline}
\end{prop}
\comTT{Let us emphasize that, in the statement of Proposition \ref{prop:comparaison2}, we considered an “abstract” small function $\psi$, and the associated transport $\id + \psi$, however we will apply this Proposition to $t \psi$ and $\id + t \psi$ respectively, where $\psi$ is the \textit{ad hoc} function constructed above.}

In particular, we obtain, with $\phi_t$ as in Definition \ref{def:approxquant} and $\Phi_t$ as above 
\begin{equation} \label{comparisonenergyt}
\FN(\Phi_t(\XN), \tmut) - \FN(\XN, \mueq) \le \frac{ t}{\beta}\Ani(\psi, \XN ,\mueq) + \frac{t}{2\beta }\sum_{i \in I} \div \psi (x_i) + \Error_t,
\end{equation}
with an $\Error_t$ term bounded by
\begin{multline}  \label{errortransinglepourt}
|\Error_t| \preceq
t^2 \|  \psi\|_{C^{0,1} }^2  \int_{U_N} |\nab H_{N,{s\vecr}}^{\mueq}|^2
+ s^2 t \|\psi\|_{C^{0,1}}  \left(\# I_N +  \int_{U_N} |\nab H_{N, s\vecr}^{\mueq}|^2 \right)
\\ 
+ \# I_N \frac{st}{\sqrt{N}} \|\psi\|_{C^{1,1}} +  \# I_{\partial,\psi } \, t  \|\psi\|_{C^{0,1}}.
\end{multline}
In our context, $\|\psi\|_{C^{0,1}} \approx \L_N^{-2}$ and $\|\psi\|_{C^{1,1}} \approx \L_N^{-3}$ (see \eqref{regpsimeso}), also $U_N$ is of size $\L_N$ and the electric energy scales like $N |U_N| = N \L_N^2$. We thus obtain an error term of order
$$
|\Error_t| \preceq N t^2 \L_N^{-2} + s^2 tN + st \sqrt{N} \L_N^{-1} + \# I_{\partial,\psi }  t \L_N^{-2}.
$$
In the end, we shall take $t$ of order $1/N$. We readily see that the first and the third terms in the right-hand side will then vanish as $N \ti$. The second one gives a $O(s^2)$ contribution, and we will then take $s \to 0$ (see below). For the last term, let us recall that it only appears in the macroscopic boundary case, when $\L_N = 1$, and that $\# I_{\partial,\psi }$ is of course $\ll N$, which ensures that $\# I_{\partial,\psi }  t \L_N^{-2}$ is negligible as $N \ti$.

 \subsection{Comparison of partition functions by transport and smallness of the anisotropy}
 \def \PNbetaz{\mathbb{P}_{N, \beta}^{(0)}}
\def \PNbetat{\mathbb{P}_{N, \beta}^{(t)}}
\def \ErrorLin{\mathsf{Error}_{\mathsf{Lin}}}
Our next goal is to show that the exponential moments of the anisotropy term appearing in \eqref{transingle} are small. This is achieved by evaluating the ratio of the partition functions in two different ways: one by Lemma \ref{lem:comparaisonmacro}, and one by the change of variables outlined above, and comparing the results.  
\begin{prop} \label{prop:comparaisontransport1} Let $N \geq 1 $ and $|t| \leq \min(\tmax, \ttmax)$, we have
\begin{multline} \label{comparaisontransport1A}
\log \KNbeta(\tmut, \tzetat) - \log \KNbeta(\mueq, \zeta_0) \geq \left(1 - \frac{\beta}{4} \right) N \left(\Ent(\mueq) - \Ent(\tmut)\right)   \\
+ \log \Esp_{\PNbetaz} \left[ \exp\left(  - \frac{1}{2} t \Ani(\psi, \XN, \mueq) \right) \right] + \ErrorLin(t)
\end{multline}
and similarly
\begin{multline} \label{comparaisontransport1B}
\log \KNbeta(\tmut, \tzetat) - \log \KNbeta(\mueq, \zeta_0) \leq \left(1 - \frac{\beta}{4} \right) N \left(\Ent(\mueq) - \Ent(\tmut)\right)   \\
- \log \Esp_{\PNbetat} \left[ \exp\left(  - \frac{1}{2} t \Ani(-\psi, \XN, \tmut) \right) \right] + \ErrorLin(t)
\end{multline}
with an $\ErrorLin$ term satisfying
\begin{multline} \label{ErrortransportI}
 \ErrorLin(t)  \preceq t^2 \|  \psi\|_{C^{0,1}}^2 N\L_N^2 (1+|\log s|) + t s^2 \|\psi\|_{C^{0,1}} N \L_N^2 (1+ |\log s|) \\
  + t \sqrt{N} \L_N^2 \|\psi\|_{C^{1,1}} + tN^{\frac{2}{3}}   \indic_{\partial}\|\psi\|_{C^{0,1}}    , 
\end{multline}
 where $\indic_{\partial }$ is $1$ in the boundary case, and $0$ in the interior cases.
\end{prop}
In particular, in view of \eqref{regpsimeso}, we have (up to a fixed multiplicative constant depending on~$\xi$)
\begin{equation} \label{ErrortransportIExplicite}
|\ErrorLin(t)|  \leq  t^2 \L_N^{-2} N (1 + |\log s|) + ts^2 N(1 + |\log s|) + t \sqrt{N} \L_N^{-1} + tN^{\frac{2}{3}} \indic_{\partial}.
\end{equation}
In the end, we will take $t$ of order $1/N$ and we will send $s \to 0$ as $N \ti$. We can observe that $\ErrorLin$ will then vanish as $N \ti$.
\begin{proof}
Starting from \eqref{acv}, 
inserting  \eqref{comparisonenergyt} and the fact that $\tzetat = \zeta_0 \circ \phi_t^{-1}$ by definition, we may write
\begin{multline*}
\KNbeta(\tmut, \tzetat) \geq 
 \int_{(\R^2)^N} \exp\left(- \frac{\beta}{2} \left( \FN(\XN, \mueq) + 2N \sum_{i=1}^N \zeta_0(x_i) \right) \right. \\  \left. - \frac{t}{2} \Ani(\psi, \XN, \mueq) + \left(\sum_{i=1}^N \log \det D \phi_t(x_i) - \frac{1}{4} t\,\div \psi(x_i) \right)\right) \exp(\Error_t) d\XN.
\end{multline*}
with $\Error_t$ as in \eqref{errortransinglepourt}.

Using the a priori bound on the fluctuations given by Proposition \ref{prop:fluctenergy} we see that
\begin{equation}\label{00}
\sum_{i=1}^N \log \det D\phi_t (x_i) = N \int_{\R^2} \log \det D\phi_t  d\mueq + \Error_{\psi}, 
\end{equation}
with an $\Error_{\psi}$ term controlled, as in \eqref{controlfluctuations2}, by
$$
\Error_{\psi} \preceq \|\psi\|_{C^{1,1}} \left(|U_N|^{\hal} \|\nabla H^{\mueq}_{N, \vec{r}}\|_{L^2(U_N)} + N^{\hal} |U_N|\|\mueq\|_{L^{\infty}} \right).
$$
Since $\phi_t$ transports $\mueq$ on $\tmut$, we have  $\det D\phi_t= \frac{\mueq}{\tmut \circ \phi_t}$ and thus 
\begin{equation}\label{22}
\int_{\R^2} \log  \det D\phi_t \, d\mueq=  \int_{\R^2} \log \mueq \,  d\mueq - \int_{\R^2} \log \tmut(\phi_t)\, d\mueq=   \Ent(\mueq)- \Ent(\tmut).
\end{equation}
We also have, by Taylor expansion
\begin{equation}\label{33}
\frac{t}{\beta} \,\div \psi(x_i) = \log \det  D\phi_t(x_i) + O(t^2) \|\psi\|^2_{C^{0,1}}.
\end{equation}
Hence we get, using \eqref{00} and \eqref{22}
\begin{multline*}
\sum_{i=1}^N \log \det D \phi_t(x_i) - \frac{t}{4} \div \psi(x_i)= \( 1- \frac{\beta}{4}\) \sum_{i\in I_N} \log \det  D\phi_t(x_i) + \# I_N O(t^2) \|\psi\|^2_{C^{0,1}}.   \\ 
=  N\left(1 - \frac{\beta}{4}\right) \(\Ent(\mueq)- \Ent(\tmut)\) + \#I_N O(t^2) \|\psi\|^2_{C^{0,1}} + \Error_{\psi}.
\end{multline*}
Finally, we may write
\begin{multline*}
\frac{\KNbeta(\tmut, \tzetat)}{\KNbeta(\mueq, \zeta_0)} \geq  \exp\left( N  \left(1 - \frac{\beta}{4}\right)  \left(\Ent(\mueq) - \Ent(\tmut)  \right) \right) \\ 
\times \Esp_{\PNbetaz} \left[ \exp\left(  - \frac{t}{2}  \Ani(\psi, \XN, \mueq) + \Error \right) \right],
\end{multline*}
with an $\Error$ term obtained as the sum 
$$ \Error_t + \Error_\psi + \#I_N O(t^2) \|\psi\|^2_{C^{0,1}},$$
which yields
\begin{multline}  \label{ErrorComparaisonaprestransport}
\Error \preceq \\
t^2 \|  \psi\|_{C^{0,1}}^2  \left( \int_{U_N} |\nab H_{N,s\vecr}^{\mueq}|^2 + \#I_N\right)
+ s^2 t \|\psi\|_{C^{0,1}}  \left(\# I_N +  \int_{U_N} |\nab H_{N, s\vecr}^{\mueq}|^2 \right)
+ \# I_N \frac{st}{\sqrt{N}} \|\psi\|_{C^{1,1}}  \\ 
+ t \|\psi\|_{C^{1,1}} \left(|U_N|^{\hal} \|\nabla H^{\mueq}_{N, \vecr}\|_{L^2(U_N)} + N^{\hal} |U_N| \|\mueq\|_{L^{\infty}} \right)  +  \# I_{\partial,\psi } \, t  \|\psi\|_{C^{0,1}}.
\end{multline}
Using \eqref{contrnbpoints}, \eqref{nbpbord} and the control on the exponential moments of $\int_{U_N} |\nab H_{N,s\vecr}^{\mueq}|^2$  in Lemma~\ref{lem:expmoment2}, we obtain
\ed{
\begin{multline*}
\left| \log \Esp_{\PNbetaz} \left[  \exp\left(-\frac{\beta}{2} \Error \right) \right] \right| 
\\
\preceq 
t^2 \|  \psi\|_{C^{0,1}}^2 N\L_N^2 (1+|\log s|) + t s^2 \|\psi\|_{C^{0,1}} N \L_N^2 (1+| \log s|) + st \sqrt{N} \L_N^2 \|\psi\|_{C^{1,1}} \\
+  t \|\psi\|_{C^{1,1}}  \sqrt{N} \L_N^2  + tN^{\frac{2}{3}}   \indic_{\partial}\|\psi\|_{C^{0,1}}   .
\end{multline*}}
We may thus write
\begin{multline} \label{premieremoitie}
\frac{\KNbeta(\tmut, \tzetat)}{\KNbeta(\mueq, \zeta_0)} \geq  \exp\left(   N\left(1 - \frac{\beta}{4}\right)  \left(\Ent(\mueq) - \Ent(\tmut)  \right) \right) \\ 
\times \Esp_{\PNbetaz} \left[ \exp\left(  - \frac{t}{2}  \Ani(\psi, \XN, \mueq) \right) \right] \times \exp(\ErrorLin),
\end{multline}
with $\ErrorLin$ as in \eqref{ErrortransportI}. Exchanging the roles of $\mueq$ and $\tmut$ we also obtain
\begin{multline} \label{deuxiememoitie}
\frac{\KNbeta(\mueq, \zeta_0)}{\KNbeta(\tmut, \tzetat)} \geq  \exp\left(   N\left(1 - \frac{\beta}{4}\right)  \left(\Ent(\tmut) - \Ent(\mueq)  \right) \right) \\ 
\times \Esp_{\PNbetat} \left[ \exp\left(  - \frac{t}{2}  \Ani(-\psi, \XN, \tmut) \right) \right] \times \exp(\ErrorLin).
\end{multline}
Taking the logarithm of \eqref{premieremoitie}, \eqref{deuxiememoitie}, we obtain \eqref{comparaisontransport1A} and \eqref{comparaisontransport1B}. 
\end{proof}

We may now control the size of the anisotropy. \comTT{As explained in the introduction, this control is a consequence of the comparison of the two ways of computing the relative $ K_{N,\beta}$, the fact that the anisotropy is linear, and H\"older's inequality.}

\begin{coro}[The anisotropy is small]  \label{coro:Aniso}
Let $N \geq 1$ and $\tau$ be fixed. We have, for $t \in [0,1]$,
\begin{equation} \label{anismallagain}
\log \Esp_{\PNbetat} \left[ \exp\left(- \frac{\tau}{N} \Ani(\psi, \XN, \mueq)  \right) \right] = o_N(1), 
\end{equation}
and the convergence is uniform for bounded $\tau$.
\end{coro}
\begin{proof} 
\comTT{Comparing \eqref{comparaisonmacro} (in the macroscopic cases) or \eqref{comparaisonmeso} (in the mesoscopic case) with \eqref{comparaisontransport1A} we see that, for $|\ep| \leq \min(\tmax, \ttmax)$ and $s \in (0, \hal)$.
\begin{equation} \label{avantHolder}
\log \Esp_{\PNbetaz} \left[ \exp\left(  - \frac{\ep}{2}   \Ani(\psi, \XN, \mueq) \right)\right] \le N \L_N^2  r_N  + \ErrorLin(\ep),
\end{equation}
with $\ErrorLin$ as in \eqref{ErrortransportIExplicite} and $\lim_{N \ti} r_N = 0$. 
Let us now use H\"older's inequality with exponent $p= \frac{\ep N}{2\tau}$ where $\tau$ is a fixed number.  If   $p>1$, we may write 
\begin{equation} \label{Holderunefois}
\log \Esp_{\PNbetaz} \left[ \exp\left(  - \frac{\tau}{N} \Ani(\psi, \XN, \mueq) \right)\right] \leq \frac{2\tau}{\ep N} \log \Esp_{\PNbetaz} \left[ \exp\left(  - \frac{\ep}{2} \Ani(\psi, \XN, \mueq) \right)\right].
\end{equation}
Combining this with \eqref{avantHolder} and inserting \eqref{ErrortransportIExplicite}, we find
\begin{multline}
 \log \Esp_{\PNbetaz} \left[ \exp\left(  - \frac{\alpha \tau}{N} \Ani(\psi, \XN, \mueq) \right)\right] \\  \le 2\tau \(\frac{ \L_N^2 r_N}{\ep}+
  \ep \L_N^{-2}  (1 + |\log s|) +  s^2 (1 + |\log s|) +  N^{-\hal} \L_N^{-1} +  N^{-\frac{1}{3}} \indic_{\partial}\).
\end{multline}
Since $N\L_N^2 \gg 1$, we may choose $\delta_N$ such that 
$$1\gg \delta_N \gg r_N\qquad 1\gg \delta_N\gg N^{-1} \L_N^{-2}.$$
Choosing $\ep = \L_N^2 \delta_N$ and $s=\delta_N$, we do have $p>1$ for $N$ large enough and
we obtain the result.}
Strictly speaking, we have only obtained one inequality, but since $\psi$ is general we can apply it to $-\psi$, which is the same as changing $\tau$ into $-\tau$, which ensures that \eqref{anismallagain} holds. Since we can start from $\tmut$ instead of $\mueq$ and apply a transport, \eqref{anismallagain} extends to any $\PNbeta^{(t)}$ with the same arguments.
\end{proof}

\begin{coro}[Ratio of partition functions] \label{coro:Ratiotmuversmu}
For any fixed $\tau$, we have
\begin{multline} \label{ratiotmuversmu}
\log \KNbeta(\tilde{\mu}_{\tau/N}, \tilde{\zeta}_{\tau/N}) - \log \KNbeta(\mueq, \zeta_0) \\
= \left(1 - \frac{\beta}{4} \right) N \left(\Ent(\mueq) - \Ent(\tilde{\mu}_{\tau/N})\right) + o_N(1),
\end{multline}
and the convergence is uniform for bounded $\tau$.
\end{coro}
\begin{proof}
We combine \eqref{comparaisontransport1A} and \eqref{comparaisontransport1B} with the control on the error $\ErrorLin$ as in \eqref{ErrortransportIExplicite}, and the bound on the anisotropy  in Corollary \ref{coro:Aniso}.
\end{proof}

\subsection{Conclusion: proof of Theorem \ref{theo:CLT} }
Combining the results of Corollary  \ref{coro:LaplaceTransform}, Corollary \ref{coro:tmuetmut} and Corollary \ref{coro:Ratiotmuversmu}, we obtain, for any fixed $\tau$
\begin{multline}
\log \Esp_{\PNbetaV} \left[\exp \left( \tau  \Fluct_N(\xi_N) \right) \right] \\ 
=  \left(1 - \frac{\beta}{4} \right) N \left(\Ent(\mueq) - \Ent(\tilde{\mu}_{\tau/N})\right) 
+ \frac{\tau^2}{4 \pi \beta}  \int_{\R^2} |\nabla \xi^{\Sigma} |^2 + o_N(1),
\end{multline}
and the convergence is uniform for $\tau$ bounded. 

It remains to use the following result
\begin{lem} \label{lem:comparerlesentropies} Let $\tau$ be fixed.

In the macroscopic interior case, we have
\begin{equation}
\label{253v2}
\int_{\R^2} \tilde{\mu}_{\tau/N} \log \tilde{\mu}_{\tau/N} - \int_{\R^2} \mueq\log \mueq = \frac{-\tau}{2\pi  N \beta} \left(\int_{\R^2} \Delta \xi \log \Delta V  \right) + N^{-1}  o_N(1).
\end{equation} 

In the macroscopic boundary case, we have
\begin{equation}
\label{253v2bis}
\int_{\R^2} \tilde{\mu}_{\tau/N} \log \tilde{\mu}_{\tau/N} - \int_{\R^2} \mueq\log \mueq = \frac{-\tau }{2\pi N\beta} \left(\int_{\R^2} \Delta \xi \left( \mathbf{1}_{\Sigma} + \left( \log \Delta V \right)^{\Sigma}  \right) \right) + N^{-1}  o_N(1).
\end{equation} 

In the mesoscopic cases, we have
\begin{equation} \label{253bis}
\int_{\R^2} \tilde{\mu}_{\tau/N} \log \tilde{\mu}_{\tau/N} - \int_{\R^2} \mueq\log \mueq = N^{-1} o_N(1).
\end{equation} 

Moreover the terms $o_N(1)$ converge to $0$ as $N \ti$ uniformly for $\tau$ bounded.
\end{lem}
The proof of Lemma \ref{lem:comparerlesentropies} is given in Section \ref{sec:preuvecomparerlesentropies}.

We have thus proven the convergence of the Laplace transform of the fluctuations to that of a Gaussian of mean $\Mean(\xi)$ (respectively $0$ in the mesoscopic cases) and variance $\Var(\xi)$.  It is well-known \comTT{(see \cite[Chap XIII.1 Theorem 2a]{MR0270403})} that such a convergence implies convergence in law of the fluctuations, and this can be rephrased in terms of convergence of $\Delta^{-1} \fluct_N$ to a Gaussian Free Field in $\Sigma$. 

We note that the non-explicit error term $r_N$ in  Proposition \ref{lem:comparaisonmacro} is what prevents us from obtaining an explicit convergence rate. 

\section{Proofs of the additional results}\label{concl}
\subsection{Moderate deviations: proof of Theorem \ref{theo:ModDev2}}
\begin{proof}[Proof of Theorem \ref{theo:ModDev2}]
The proof is the same as that of Theorem \ref{theo:CLT}, except that instead of showing that the anisotropy term is small, we use the rougher control 
$$\left|\log \Esp_{\PNbetaV} \left[\exp\left(\frac{\tau}{N} \Ani(\psi, \XN,\mueq)\right)\right]\right|\le  C \tau \L_N^2 (1+|\log s|).$$
which is an immediate consequence of the definition \eqref{def:Ani} and the control \eqref{controleexpmoment2}.
Choosing $s= \hal$ and inserting this control into  \eqref{LTttot0} as above, we find that if $\tau \ll N\L_N^2$, we always have
$$
\left|\log \Esp_{\PNbeta} \left( \exp \tau  \Fluct_N(\xi)   \right)   \right| \le C\( \tau^2 + \tau \L_N^2  \),
$$where the  constant depends only on $V$ and $\|\xi\|_{C^{3,1}}$ (resp. $\|\xi\|_{C^{2,1}}$ in the interior cases).
Applying Markov's inequality we obtain
\begin{equation*}
\PNbeta\left(|\Fluct_N(\xi_N)| \geq c \tau \right)  \leq \exp\left(-\frac{c^2}{2} \tau^2\right),
\end{equation*}
for any $\tau$ such that $1 \ll \tau \ll N \L_N^2$ and $c$ large enough independent of $\tau$ and $N$. 

This proves Theorem \ref{theo:ModDev2}.
\end{proof}

\subsection{Joint law of linear statistics: proof of Corollary \ref{theo:independance}}
Let $\xi^{(1)}, \dots, \xi^{(m)}$ be appropriately regular test functions. For any $\alpha_1, \dots, \alpha_m$ in $\R$, Theorem \ref{theo:CLT} implies that
$$
\Fluct_N\left[ \sum_{k=1}^m \alpha_k \xi^{(k)} \right]
$$
converges to a Gaussian random variable with the same law as $\sum_{k=1}^m \alpha_k \mathcal{G}^{(k)}$ where the $\mathcal{G}^{(k)}$ are the limit laws of $\Fluct_N[\xi^{(k)}]$ with covariance matrix as given in the statement of Corollary~\ref{theo:independance}. It implies that the vector $\left( \Fluct_N[\xi^{(1)}], \dots, \Fluct_N[\xi^{(m)}]\right)$ is jointly Gaussian in the limit $N \to \infty$, with the correct covariance matrix. Of course the corresponding result also holds in the mesoscopic regime. 

\subsection{Fluctuations for minimizers: proof of Theorem \ref{thmini}}
We use the same notation as in Theorem \ref{theo:CLT} and its proof.
\coT{Considering energy minimizers formally correspond to taking $\beta = + \infty$. Although some factors $\beta$ appear in the argument above, they are always compensated by the fact that the transport map is of the form $\id + \frac{t}{\beta} \cdot$. In the first step of the proof, we use $\beta =2$ in a purely formal way, in order to write an algebraic identity expressing the fluctuations as a difference of energies. The important part is that we apply this identity to \textit{minimizers} of the energy, which is really taking $\beta = + \infty$.}

{\bf Step 1.} {\it Re-expressing the fluctuations.}
\\
Let $\XN$ be a minimizer of $\HN$. Applying \eqref{relfinale} with $\beta =2$ (which means that $\mueqt$ is associated to the potential $V - t \xi_N$ as in Definition \ref{def:mueqt}),   we have
\begin{multline*}
Nt \Fluct_N(\xi_N) =  F_N (\XN,\mueq) - F_N(\XN,\mueqt)
-  2N \sum_{i=1}^N (\zetat(x_i)-\zetaz(x_i))\\ - 2N^2 \int_{\R^2} \zetaz \, d\mueqt
+ \frac{N^2t^2}{8\pi } \int_{\R^2} |\nab \xi_N|^2 - \frac{N^2}{2\pi} \int_{\R^2}\left| \nabla \hmueqt - \nabla \hmut \right|^2 .
\end{multline*}
It is also shown in \cite[Theorem 3]{nodari2014renormalized} that when $\XN$ minimizes $\HN$, all the points $x_i$ belong to $\Sigma$, so that $\zetaz(x_i)=0$, and since $\zetat\ge 0$ we may write
\begin{multline}\label{rexpf}
N t \Fluct_N(\xi_N) \le F_N (\XN,\mueq) - F_N(\XN,\mueqt)
\\ - 2N^2 \int_{\R^2} \zetaz \, d\mueqt
+ \frac{N^2t^2}{8\pi } \int_{\R^\d} |\nab \xi_N|^2 - \frac{N^2}{2\pi} \int_{\R^2}\left| \nabla \hmueqt - \nabla \hmut \right|^2.
\end{multline}
The terms in the second line of the right-hand side of  this relation will be estimated by Lemma~\ref{lem:termesfactorises}, so there remains to estimate $F_N(\XN,\mueq)- F_N(\XN,\mueqt)$.

We will need the following a priori bounds, similar to Lemma \ref{lem:expmoment2}
\begin{equation}\label{apmini}
\int_{U_N}|\nab H_{N,s\vec{\rr}}|^2 \le CN |U_N| (1+|\log s|) \qquad \#I_N\le C N |U_N|,\end{equation}
which  come as  a consequence of the analysis of \cite{nodari2014renormalized}, using the local  bound on the energy proved there, as well as the point separation result which states that we always have 
$$\rr(x_i) \ge c N^{-1/2},$$ 
for some $c>0$ uniform.

{\bf Step 2.} {\it Using Proposition \ref{prop:comparaison2} and showing that the anisotropy is small. }
\\ Let us use an arbitrary regular transport map $\psi$ and $\Phi= \id +  \psi$, $\nu = \Phi\# \mu$. 
Applying Proposition \ref{prop:comparaison2}  we obtain that, if $\|\psi\|_{C^{1,1}}$ is small enough,
\begin{equation} \label{applicomparaison3}
  F_N(\Phi(\XN),\nu) - F_N(\XN,\mueq)\le \Ani(\psi, \XN, \mueq)  + \hal \sum_{i \in I} \div \psi(x_i) +\Error 
\end{equation}
with 
\begin{multline}
 |\Error| \preceq
 \|  \psi\|_{C^{0,1} }^2  \int_{U_N} |\nab H_{N,s\vecr}^{\mu}|^2
+ s^2 \|\psi\|_{C^{0,1}}  \left(\# I_N +  \int_{U_N} |\nab H_{N, s\vecr}^{\mu}|^2 \right)
\\ + \# I_N \frac{s}{\sqrt{N}} \|\psi\|_{C^{1,1}} + \# I_{\partial,\psi } \,  \|\psi\|_{C^{0,1}}\\
\preceq N|U_N| \( \|  \psi\|_{C^{0,1} }^2  (1+|\log s|)+  s^2 (1+|\log s|) \|\psi\|_{C^{0,1}}+ \frac{s}{\sqrt{N}} \|\psi\|_{C^{1,1}} \) + \indic_{\partial } N^{\frac{2}{3}} \|\psi\|_{C^{0,1}},
\end{multline}
 where we used \eqref{apmini} and \eqref{nbpbord}.
Since $F_N(\Phi(\XN),\nu)\ge \min F_N(\cdot,\nu)$, we deduce that
\begin{multline} \label{applicomparaison4}
\min F_N(\cdot, \mueq)\ge \min F_N(\cdot,\nu)- \Ani(\psi, \XN, \mueq)  - \hal \sum_{i \in I_N} \div \psi(x_i) \\ - C N|U_N|\left(   (1+|\log s|) \|  \psi\|_{C^{0,1} }^2+ s^2  (1+|\log s|)  \|\psi\|_{C^{0,1}}+\frac{s}{\sqrt{N}} \|\psi\|_{C^{1,1}}\right) -C   \indic_{\partial}\|\psi\|_{C^{0,1}}   N^{\frac{2}{3}} .
\end{multline}
Next, we claim that, as long as $N^{-1/2}\ll \L_N\ll 1$, we have
\begin{equation}\label{compminim}
\min F_N(\cdot, \mueq) - \min F_N(\cdot, \nu)=-\hal N\(\Ent(\mueq)- \Ent(\nu) \)+ r_N  N |U_N| ,\end{equation} with $\lim_{N\to \infty } r_N =0$. 
In the macroscopic case, this is a direct consequence of the energy expansion of \cite{SS2d} and the explicit scaling of the renormalized energy as in \cite[(1.37)]{SS2d}. In the mesoscopic case, to prove it we may apply the screening procedure of \cite{nodari2014renormalized} on the boundary of a square $K_N$ of sidelength $\in [4\L_N, 6\L_N]$ centered at $\bar x_N$, which is valid down to the microscopic scale. This allows to show that if $\XN$ is a minimizer of $F_N(\cdot, \mueq)$ and $\YN$ a minimizer of $F_N(\cdot,\mueqt)$ with $\mueq =\mueqt$ outside of $\D(\bar x_N, 2\L_N)$, their energies differ by at most the difference of the minimal energies in $K_N$, which is in turn known (combining  \cite[Theorem 3]{nodari2014renormalized} with the scaling of the renormalized energy)   to be the right-hand side of \eqref{compminim}.

Combining \eqref{compminim} with \eqref{applicomparaison4} and using the arguments of \eqref{00}--\eqref{33}, it follows that
\begin{multline*} -\Ani(\psi, \XN,\mu) \preceq N|U_N| \left( \|  \psi\|_{C^{0,1} }^2  (1+|\log s|) + s^2 (1+|\log s|)  \|\psi\|_{C^{0,1}}+\frac{s}{\sqrt{N}} \|\psi\|_{C^{1,1}}+r_N \right)  \\+  \indic_{\partial}\|\psi\|_{C^{0,1}}  N^{\frac{2}{3}} )
.\end{multline*}
Changing $\psi$ into $-\psi$ and using that $\Ani$ is linear in $\psi$, we obtain equality:
\begin{multline}\label{anismall}|\Ani(\psi, \XN,\mu) | = N|U_N| O\left( (1+|\log s|)  \|  \psi\|_{C^{0,1} }^2+ s^2 (1+|\log s|)  \|\psi\|_{C^{0,1}}+\frac{s}{\sqrt{N}} \|\psi\|_{C^{1,1}}+ r_N\right) \\+  \indic_{\partial}\|\psi\|_{C^{0,1}}  O( N^{\frac{2}{3}} ).
\end{multline}

{\bf Step 3.} {\it Conclusion}.\\
We now consider the map $\psi$ constructed in Section \ref{sec:transport}, and let 
$$
\Phi := \id+ \frac{\tau}{N}\psi= \id + \frac{\tau }{N\ep} \ep \psi
$$
 with $\tau, \ep \le 1$,  and $\tilde \mu_{\tau/N} = \Phi \# \mueq$.
We are now in a position to evaluate $F_N(\XN,\mueq)- F_N(\XN,\tilde\mu_{\tau/N})$.
Applying Proposition \ref{prop:comparaison2} again, this time to $\Phi^{-1}(\XN) $, and using that $F_N(\XN,\mueq) \le F_N(\Phi^{-1} (\XN),\mueq)$ by minimality,  it follows in view of \eqref{anismall} applied with $\ep \psi$, the linearity of $\Ani$, and the arguments in \eqref{00}--\eqref{33} that
\begin{multline}\label{fnm} F_N(\XN,\mueq) - F_N(\XN,\tilde \mu_{\tau/N}) \le -\hal   N \(\Ent(\mueq)-\Ent(\mu_{\tau/N})\)
 \\ +\tau \L_N^2 \left(  \ep  (1+|\log s|) \|\xi\|_{C^{1,1}}^2 +s^2 (1+|\log s|)  \|\xi\|_{C^{1,1}} + \frac{s }{\sqrt{N}}\|\xi\|_{C^{2,1}}+\frac{r_N}{\ep} \right)\\+ \indic_{\partial } O( \tau N^{-\frac{1}{3}} \|\xi\|_{C^{1,1}}) .
\end{multline}
Moreover 
\begin{multline*}
F_N(\XN, \tilde \mu_{\tau/N})-F_N(\XN, \mu_{\tau/N})= O\left( \tau^4 N^{-2}\L_N^4 M_{\xi}^2\)
 + O\left(\tau^2 N^{-1/2} \L_N^3 M_{\xi} \right)  +O\(  \tau^2   N^{ -\frac{1}{3}}\)\indic_{\partial} 
\end{multline*} as in \eqref{mutomut}.
Inserting the last two relations into \eqref{rexpf}, since $N\L_N^2 \gg 1$, we obtain that
\begin{multline}\label{relfinale3}
\tau  \Fluct_N(\xi_N) \le  -\hal   N (\Ent[\mueq]-\Ent[\mu_{\tau/N} ]) - 2N^2 \int_{\R^2} \zeta_0 \, d\mu_{\tau/N}
+ \frac{1}{8\pi } \int_{\R^2} |\nab \xi_N|^2\\ - \frac{N^2}{2\pi} \int_{\R^2}\left| \nabla h^{\mu_{\tau/N}}  - \nabla h^{\overline{\mu_{\tau/N}}} \right|^2
+o_N(1)+ \tau O( \ep (1+|\log s|)  + s^2 (1+|\log s|)  ).
  \end{multline}
Combining with the results of Lemma \ref{lem:termesfactorises}, we get in the macroscopic case
\begin{equation*}
\tau  \Fluct_N(\xi_N) \le  \frac{-\tau}{8\pi} \int_{\R^2} \Delta \xi \left( \indic_{\Sigma} + (\log \Delta V)^\Sigma)\right) +o_N(1)+ \tau O( \ep (1+|\log s|)  + s^2 (1+|\log s|)  ).
\end{equation*}
as $N \to \infty $ and $s \to 0$; and in the mesoscopic case
\begin{equation*}
\tau  \Fluct_N(\xi_N) \le o_N(1)+ \tau O( \ep (1+|\log s|)  + s^2 (1+|\log s|)  ).
\end{equation*}
Dividing by $\tau$ and changing $\tau $ into $-\tau$, then letting $N\to \infty$, $s\to 0$  and $\ep \to 0$, we conclude that
$$
\lim_{N\to \infty}\Fluct_N(\xi) = \frac{-1}{8\pi} \int_{\R^2} \Delta \xi \left( \indic_{\Sigma} + (\log \Delta V)^\Sigma)\right)
$$
 in the macroscopic case, and $\lim_{N\to \infty } \Fluct_N(\xi_N)=0$ otherwise.
This concludes the proof.

\appendix
\section{Proof of Proposition \ref{prop:comparaison2}} 
\label{sec:proofcomparaison}
\subsection{A preliminary bound on the potential near the charges} Let $\mu$ be a bounded probability density on $\R^2$, let $N \geq 1$ and $\XN$ be in $(\R^2)^N$.  For any $i = 1, \dots, N$ we let
\begin{equation} \label{def:tH}
\tH_{N, i}^{\mu}(x)  :=  H_N^{\mu}(x) + \log |x-x_i|,
\end{equation}
where $H_N^{\mu}$ was defined in \eqref{def:HNmuABC}.
\begin{lem} 
We have for $i = 1 , \dots, N$
\begin{equation} \label{HNeta}
H_{N, \vecr}^{\mu}  =\begin{cases}
H_N^{\mu} & \text{ outside } \D(x_i, \rr(x_i)) \\
\tH_{N, i}^{\mu} \text{ (up to a constant)} & \text{ in each } \D(x_i, \rr(x_i)).
\end{cases}
\end{equation}

In particular, it holds, for $s \in (0, \hal)$
\begin{equation}\label{idenhh}
\int_{\R^2}|\nab H_{N, s\vecr} ^{\mu}|^2=
\int_{\R^{2}\backslash \cup_{i=1}^N \D(x_i, \rr(x_i) )} |\nab H_{N}^{\mu} |^2 + \sum_{i=1}^N \int_{\D(x_i, \rr(x_i) )} |\nab \tH_{N,i}^{\mu}|^2.
\end{equation}
\end{lem}
\begin{proof}
The first point follows from \eqref{def:HNmutrun} and the fact that  the disks $\D(x_i,\rr(x_i))$ are disjoint by definition. The second point is a straightforward consequence of the first one.
\end{proof}

We let for $i=1, \dots, N$
\begin{equation} \label{def:lambdai}
\lambda_i(\XN, \mu):= \int_{ \D(x_i, \rr(x_i) )} |\nab \tH_{N,i}^\mu|^2.
\end{equation}
\comTT{We also recall that truncated potentials have been defined in \eqref{def:HNmutrun}.}
\begin{lem} We have for $i = 1, \dots, N$
\begin{equation}\label{contrhi}
\|\nab \tH_{N,i}^\mu \|_{L^\infty\left(\D\left(x_i, \hal \rr(x_i) \right)\right) }\le C  \sqrt{N} \left( \|\mu\|_{L^{\infty}} \sqrt{N} \rr(x_i) + \frac{1}{\sqrt{N} \rr(x_i)}  \sqrt{\lambda_i(\XN,\mu)} \right),
\end{equation}
for some universal constant $C$.
\end{lem}

\begin{proof}
We exploit the fact that $\tH_{N,i}^\mu$ is almost harmonic in each $\D(x_i, \rr(x_i))$. Since the disks $\D(x_i, \rr(x_i))$ are disjoint, $\tH_{N,i}^\mu$ satisfies
\begin{equation*}
-\Delta \tH_{N,i}^\mu = -2\pi N \mu \quad \text{in} \ \D(x_i, \rr(x_i)).
\end{equation*}
We may thus write $\tH_{N,i}^\mu$ as $\tH_{N,i}^\mu=u+v$ where
\begin{equation*}
\left\lbrace \begin{array}{ll}
\Delta u=0& \text{in} \ \D(x_i, \rr(x_i))\\
u=\tH_{N,i}^\mu & \text{on} \ \p \D(x_i, \rr(x_i)),
\end{array}
\right.
\end{equation*}
\begin{equation*}
\left\lbrace \begin{array}{ll}
-\Delta v=- 2\pi N \mu & \text{in} \ \D(x_i, \rr(x_i))\\
v=0& \text{on} \ \p \D(x_i, \rr(x_i)).
\end{array}\right.
\end{equation*}
Standard regularity estimates for harmonic functions yield that
\begin{equation}\label{regu}
\|\nab u\|_{L^\infty \left( \D\left(x_i, \hal  \rr(x_i)\right)\right)} \leq C \left( \frac{1}{r(x_i)^2} \int_{\D\left(x_i, \rr(x_i)\right)} |\nab \tH_{N,i}^\mu |^2\right)^{\hal} \leq C\frac{1}{\rr(x_i)} \sqrt{\lambda_i(\XN, \mu) }
\end{equation}
where $C$ is universal.
On the other hand, setting $w(x)=v(x_i+ \rr(x_i) x)$, we have that
\begin{equation*}
\left\lbrace \begin{array}{ll}
\Delta w = 2\pi N \rr(x_i)^2  \mu\left(x_i + \rr(x_i) x\right)   & \text{in} \ \D(0,1)
\\
w=0 & \text{on} \ \p \D(0,1).
\end{array}\right.
\end{equation*}
Since $\mu$ is in $L^{\infty}$, standard elliptic regularity estimates yield the bound
\begin{equation*}
\|w\|_{W^{2, p}\left(\D(0,1)\right)} \le C_p N  \|\mu\|_{L^{\infty}} \rr(x_i)^2
\end{equation*}
for all $p<\infty$ (with a constant $C_p$ depending only on $p$), and by Sobolev embedding we deduce
\begin{equation*}
\|\nab w\|_{L^\infty(\D(0,1))} \le C N  \|\mu\|_{L^{\infty}} \rr(x_i)^2,
\end{equation*}
which implies in turn
\begin{equation}\label{nabv2d}
\|\nab v \|_{L^\infty \left(\D(x_i,  \rr(x_i))\right) } \leq C N  \|\mu\|_{L^{\infty}} \rr(x_i).
\end{equation}
Combining \eqref{regu} and \eqref{nabv2d} we obtain \eqref{contrhi}.
\end{proof}

\subsection{Transporting electric fields}
\label{sec:transportedfields}
Let $\mu, \nu$ be two bounded probability densities and let $\Phi$ be a bi-Lipschitz map transporting $\mu$ onto $\nu$. For any $\XN \in (\R^2)^N$ we let $\YN := (\Phi(x_1), \dots, \Phi(x_N))$ and we define the transported electric fields $E$, $E_i$ ($i=1, \dots, N$) as
\begin{equation}
\label{defiE} 
\begin{cases}
E := & (D\Phi\circ\Phi^{-1})^T\nab H_N^{\mu} \circ\Phi^{-1}|\det D\Phi^{-1}|\\ 
 E_i := & ( D \Phi\circ \Phi^{-1})^T\nab \tH_{N,i}^{\mu} \circ \Phi^{-1}|\det D\Phi^{-1}|.
 \end{cases}
\end{equation} 
In particular, we may observe that $E = \nab H_N^{\mu}$, and $E_i = \nab \tH_{N,i}^{\mu}$ for any $i$, on the interior of the set $\{ \Phi \equiv \id \}$.

The following lemma expresses the fact that the transported electric fields are compatible with the point configuration $\YN$ and the background measure $\nu$.
\begin{lem}\label{lem:divE}
We have
\begin{equation*}
\div E = \div \nabla H_N^{\nu}
\end{equation*}
and for any $i =1, \dots, N$
\begin{equation*}
 \div E_i = \div \nabla \tH_{N,i}^{\nu}.
\end{equation*}
\end{lem}
\begin{proof}
Let $\phi$ be a smooth test-function, and let $h, f$ be such that $\Delta h = f$ (in the distributional sense).  We have $-\int \nab h \cdot \nab \phi=\int f \phi$ and so
changing variables, we find
\begin{equation*}
- \int \nab h \circ \Phi^{-1} \cdot \nab \phi \circ \Phi^{-1} |\det D\Phi^{-1}| = \int (\phi \circ \Phi^{-1}  ) (f \circ \Phi^{-1}) |\det D\Phi^{-1}|,
\end{equation*}
and writing $\nab \phi \circ \Phi^{-1} = (D\Phi\circ \Phi^{-1} )^T \nab( \phi \circ \Phi^{-1})$ we get
\begin{equation*}
-\int \nab h \circ \Phi^{-1}   \cdot ( D\Phi\circ \Phi^{-1})^T  \nab( \phi \circ \Phi^{-1})  |\det D\Phi^{-1}|= \int \phi \circ \Phi^{-1}  f\circ \Phi^{-1} |\det D\Phi^{-1}|.
\end{equation*}
Since this is true for any $\phi \circ \Phi^{-1}$ with $\phi $ smooth enough, we deduce that in the sense of distributions, we have
\begin{equation*}
\div \left( (D\Phi\circ \Phi^{-1})^T \nab h\circ \Phi^{-1} |\det D \Phi^{-1}| \right)= f \circ \Phi^{-1} |\det D \Phi^{-1}|.
\end{equation*}
Applying this to $h= H_N^{\mu}$ (respectively $\tH_{N,i}^{\mu}$) and $f= 2\pi \left( \sum_i \delta_{x_i} -N \mu\right)$, and using that $\mu = (\det D \Phi) (\nu \circ \Phi)$,  we obtain that
\begin{equation*}
\div E=  2\pi \left(\sum_{i=1}^N\delta_{\Phi(x_i) } - N \nu \right) \text{ and } \div E_i= 2\pi \left(\sum_{j \neq i}\delta_{ \Phi(x_j)} - N  \nu  \right)
\end{equation*}
and  the claim follows.
\end{proof}

\subsection{Proof of Proposition \ref{prop:comparaison2}}
\begin{proof}[Proof of Proposition \ref{prop:comparaison2}]
In the sequel, $\vec{\eta}$ denotes $s\vec{\rr}$ for some $s\in (0,\hal)$.
\setcounter{step}{0}
\begin{step}[Splitting the comparison.]
Applying Proposition \ref{prop:monoto}  to $\YN$ and $\nu$ yields
\begin{equation*}\label{fmueqt2}
F_N(\YN,\nu)=\frac{1}{2\pi} \int_{\R^2} |\nab H_{N,\veta}^{\nu}|^2
+\sum_{i=1}^N \log \eta_i+  2N  \sum_{i=1}^N \int \f_{\eta_i} (x-x_i) d\nu(x).
\end{equation*}
Let $E, E_i$ be the transported electric fields  as  in \eqref{defiE}. We define $E_{\veta}$ as the vector field
\begin{equation} \label{def:Eeta}
E_{\veta} = \begin{cases}
E_i & \text{ on } \D(y_i, \eta_i) \\
E & \text{ outside } \cup_{i =1}^N \D(y_i, \eta_i),
\end{cases}
\end{equation}where we let $y_i=\Phi(x_i)$.
The identity \eqref{HNeta} combined with Lemma \ref{lem:divE} shows that 
$$\div E_{\vec{\eta}} = \div \nab H_{N,\vec{\eta}}^{\nu}\ \text{ in $\R^2$}.$$
By $L^2$ projection property, cf. \cite[Proof of Prop. 5.12]{serfatyZur}, it implies
\begin{equation*} \label{minimaliteenergielocale}
\int_{\R^2} |\nab H_{N,\veta}^{\nu}|^2 \leq \int_{\R^2}  |E_{\veta}|^2.
\end{equation*}
Also, by definition of $\Phi$ and of the set $U_N$ we have
\begin{equation*}
E_{\veta} = \nab H_{N,\veta}^{\mu}  \  \text{ in }  \R^2 \setminus U_N.
\end{equation*}

Applying Proposition \ref{prop:monoto} to $\XN $ and $\mu$ yields
\begin{equation*}
F^{\mu}_N(\XN)=\frac{1}{2\pi } \int_{\R^2} |\nab H_{N,\veta}^{\mu}|^2  +\sum_{i=1}^N \log \eta_i + 2N  \sum_{i=1}^N \int \f_{\eta_i} (x-x_i) d\mu(x)
\end{equation*}
We thus see that
\begin{equation}\label{55}
F^{\nu}_N(\YN) - F^{\mu}_N(\XN) \leq A + B,
\end{equation}
where
\begin{align}
\label{A} A  := & \frac{1}{2\pi} \int_{U_N} \left( |E_{\vec{\eta}}|^2 - |\nab H^{\mu}_{N,\vec{\eta}}|^2 \right) \\
\label{B} B  :=  &  2N  \sum_{i=1}^N \int \f_{\eta_i} (x-y_i) d\nu(x)- 2N  \sum_{i=1}^N \int \f_{\eta_i} (x-x_i) d\mu(x).
\end{align}

We may split the error term $A$ further. First let us write, \footnote{Let us recall that $I_N$ denotes the set $\displaystyle{\{i \in \{1, \dots, N\}, x_i \in U_N\}}$.} using the definition \eqref{def:Eeta}
\begin{equation*}
\int_{U_N} |E_{\vec{\eta}}|^2 =  \int_{U_N \backslash \cup_{i \in I_N} \hD_i }  |E|^2 +\sum_{i \in I_N} \int_{\hD_i} |E_i|^2,
\end{equation*}
where we let $\hD_i := \D(y_i, \eta_i)$. Next, we may write
\begin{equation*}
\int_{U_N \backslash \cup_{i \in I_N} \hD_i }  |E|^2 = \int_{U_N \backslash \cup_{i \in I_N} \bD_i}   |E|^2 + \sum_{i \in I_N} \left(\int_{\bD_i} |E|^2 - \int_{\hD_i} |E|^2\right).
\end{equation*}
where $\bD_i := \Phi(D_i)$, with $D_i := \D(x_i, \eta_i)$. Let us summarize the new notation
\begin{equation} \label{lesDi}
D_i := \D(x_i, \eta_i), \quad \bD_i := \Phi(D_i), \quad \hD_i := \D(y_i, \eta_i).
\end{equation}

Using the same notation, we may write
\begin{equation*}
\sum_{i \in I_N} \int_{\hD_i} |E_i|^2 = \sum_{i \in I_N} \int_{\bD_i} |E_i|^2 + \sum_{i \in I} \left( \int_{\hD_i} |E_i|^2 -  \int_{\bD_i} |E_i|^2 \right).
\end{equation*}
Finally we may split $A$ as $A = \Main + \frac{1}{2\pi} \left(A_1 + A_2\right)$, with the main energy comparison term being defined as
\begin{equation*}
\Main := \frac{1}{2\pi } \left( \int_{U_N \backslash \cup_{i \in I} \bD_i}   |E|^2 +   \sum_{i \in I} \int_{\bD_i} |E_i|^2 - \int_{U_N} |\nab H_{N,\vec{\eta}}^{\mu}|^2\right),
\end{equation*}
and the error terms $A_1, A_2$ as
\begin{equation} \label{def:A1A2}
A_1 := \sum_{i \in I_N} \left(\int_{\bD_i} |E|^2 - \int_{\hD_i} |E|^2\right), \quad
A_2  := \sum_{i \in I_N} \left( \int_{\hD_i} |E_i|^2 -  \int_{\bD_i} |E_i|^2 \right).
\end{equation}
\end{step}

\begin{step}[The main energy term]
We claim that
\begin{equation} \label{mainenergyterm}
\Main =  \Ani +
O\left(\|  \psi\|_{C^{0,1} }^2  \int_{U_N} |\nab H_{N,s\vecr}^{\mu}|^2\right)
\end{equation}
with $\Ani$ as in \eqref{def:Ani} and with an implicit constant independent of $N$.
To prove \eqref{mainenergyterm}, we first use \eqref{HNeta} to split the integral of $|\nab H_{N,s\vecr}^{\mu}|^2$ and we get
\begin{equation*}
2\pi \Main = \int_{U_N \backslash \cup_{i \in I_N} \bD_i}   |E|^2 - \int_{U_N \backslash \cup_{i \in I_N} D_i} |\nabla H_{N}^{\mu}|^2 \\
+ \sum_{i \in I_N} \left( \int_{\bD_i} |E_i|^2 - \int_{D_i} |\nabla \tH_{N,i}^{\mu}|^2\right).
\end{equation*}
Inserting the definition of $E, E_i$ and changing variables we obtain
\begin{multline}\label{2pm}
2\pi \Main = \int_{U_N \backslash \cup_{i \in I_N} D_i} \left(|(D\Phi)^T \nab H_N^{\mu}|^2  |\det D \Phi^{-1}\circ  \Phi| - |\nab H_N^{\mu}|^2\right) \\
+ \sum_{i \in I_N} \int_{D_i}\left( |(D \Phi)^T\nab \tH_{N,i}^{\mu}|^2 |\det D \Phi^{-1}\circ  \Phi| - |\nab \tH_{N,i}^{\mu}|^2\right).
\end{multline}
Next writing $\Phi= \id + \psi$ and linearizing,  yields after some computation
\begin{multline} \label{mainestimate1}
2\pi \Main  = \int_{U_N \backslash \cup_{i \in I_N} D_i} \langle \mathcal A \nab H_N^{\mu}, \nab H_N^{\mu} \rangle  +   \sum_{i\in I_N} \int_{D_i}  \langle \mathcal A \nab \tH_{N,i}^{\mu}, \nab \tH_{N,i}^{\mu} \rangle \\
+ O\left(\|  \psi\|_{C^{0,1} }^2  \int_{U_N} |\nab H_{N,s\vec{\rr}}^{\mu}|^2\right) ,
\end{multline}where $\mathcal A$ is as in Definition \ref{defani}.

\end{step}

\begin{step}[The error term $A_2$]
Using the definition of $E$, $E_i$ and changing variables in \eqref{def:A1A2} as above we obtain
\begin{equation*}
A_2 = \sum_{i \in I_N} \left( \int_{\Phi^{-1}(\hD_i)}  |(D \Phi)^T\nab \tH_{N,i}^{\mu}|^2 |\det D \Phi^{-1}\circ  \Phi| - \int_{D_i}  |(D \Phi)^T\nab \tH_{N,i}^{\mu}|^2 |\det D \Phi^{-1}\circ  \Phi|\right).
\end{equation*}
We have by definition $D_i = \D(x_i, \eta_i)$ and $\Phi^{-1}(\hD_i) = \Phi^{-1}(\D(y_i, \eta_i))$. It is not hard to see that
\begin{equation*}
\D(x_i, \theta_1 \eta_i) \subset \Phi^{-1}(\D(y_i, \eta_i)) \subset \D(x_i, \theta_2 \eta_i),
\end{equation*}
with $\theta_1, \theta_2$ such that $\max(|\theta_1 -1|, |\theta_2-1|) = O \left(\|D\Phi - \id\|_{\infty}\right)$. Let us denote by $\C_i$ the annulus
\begin{equation*}
\C_i := \D(x_i, \theta_2 \eta_i) \backslash \D(x_i, \theta_1 \eta_i).
\end{equation*}
We obtain
\begin{equation*}
|A_2 |\le C (1+\|\psi\|_{C^{0,1}} ) \sum_{i \in I_N} \int_{\C_i} |\nab \tH_{N,i}^{\mu}|^2.
\end{equation*}
We now use \eqref{contrhi} and $\rr(x_i) \le \frac{1}{\sqrt{N}}$ (cf. \eqref{def:trxi}) to bound $|\nab \tH_{N,i}^{\mu}|^2$ by
\begin{equation*}
O \left( N + \frac{1}{\rr(x_i)^2} \lambda_i(\XN, \mu) \right)
\end{equation*}
uniformly on $\C_i$, as soon as $\|\psi\|_{L^\infty}$ is small enough. On the other hand, the area of $\C_i$ is bounded by
$
O\left(\eta_i^2 \|D\psi\|_{\infty} \right),$ so
summing over $i \in I_N$, and recalling that $\eta_i \leq s \rr(x_i) \leq \frac{s}{\sqrt{N}}$ we obtain
\begin{equation*}
A_2 = s^2 O\left( \|D \psi\|_{\infty} \right)  \left(\# I_N + \sum_{i \in I_N} \lambda_i(\XN, \mu) \right).
\end{equation*}
Using the definition of $\lambda_i$ and
  \eqref{idenhh}, we may write 
\begin{equation*}
\sum_{i\in I_N} \lambda_i(\XN,\mu) \le \sum_{i \in I_N} \int_{\D(x_i, \rr(x_i))} |\nab \tH_{N,i}^{\mu}|^2 \leq \int_{U_N} |\nab H_{N, \vec{\rr}}^{\mu}|^2
\end{equation*}
and we obtain
\begin{equation} \label{controleA2}
|A_2| \le C s^2 \|D\psi\|_{L^\infty}  \left(\# I_N +  \int_{U_N} |\nab H_{N, \vec{\rr}}^{\mu}|^2 \right),
\end{equation}
with a constant independent of $s, \psi, N$.
\end{step}

\begin{step}[The error term $A_1$]
Changing variables again in \eqref{def:A1A2} we have
\begin{equation*}
A_1 = \sum_{i \in I_N} \left(\int_{D_i}  |(D \Phi)^T \nab H_{N}^{\mu}|^2 |\det D \Phi^{-1}\circ  \Phi| - \int_{\Phi^{-1}(\hD_i)}  |(D \Phi)^T \nab H_{N}^{\mu}|^2 |\det D \Phi^{-1}\circ  \Phi|\right).
\end{equation*}
For each $i \in I_N$ we may decompose $\nab H_{N}^{\mu}$ on $D_i \cup \Phi^{-1}(\hD_i)$ as
\begin{equation}\label{58b}
\nab H_{N}^{\mu} = \nab \tH_{N,i}^{\mu} + \nab \log |\cdot - x_i|.
\end{equation}
Let us then first study
\begin{equation} \label{def:A1p}
A'_1 := \sum_{i \in I} \left(\int_{D_i} - \int_{\Phi^{-1}(\hD_i)}\right)\left( |(D \Phi)^T \nab \log |\cdot - x_i| |^2 |\det D \Phi^{-1}\circ  \Phi|\right),
\end{equation}
which is the leading order term, carrying the singularity near $x_i$. For shortness, we abuse notation and denote by $\left(\int_{D_i} - \int_{\Phi^{-1}(\hD_i)}\right)$ the difference of the two integrals. We may  write
\begin{equation*}
A'_1 = \left(1 +O(\|D\psi\|_{L^\infty}   )   \right) \sum_{i \in I_N} \left(\int_{D_i} - \int_{\Phi^{-1}(\hD_i)}\right) \left(|\nab \log |\cdot - x_i| |^2\right).
\end{equation*}
Since $\log |\cdot - x_i|$ is harmonic away from $x_i$ we have, using Green's formula
\begin{multline*}
\left(\int_{D_i} - \int_{\Phi^{-1}(\hD_i)}\right) \left(|\nab \log |\cdot - x_i| |^2\right)  \\ = \int_{\partial D_i}  \log |x-x_i| \frac{\partial \log |x-x_i|}{\partial \vec{n}} -  \int_{\partial (\Phi^{-1}(\hD_i))}  \log |x-x_i| \frac{\partial \log |x-x_i|}{\partial \vec{n}},
\end{multline*}
where $\vec{n}$ denotes the outer unit normal. We easily compute
\begin{equation} \label{Ap1easilycompute}
\int_{\partial D_i}  \log |x-x_i| \frac{\partial \log |x-x_i|}{\partial \vec{n}} = 2\pi \log \eta_i,
\end{equation}
and we turn to estimating the integral over $\partial (\Phi^{-1}(\hD_i))$. By a change of variables it is equal to 
$$ 2\pi \log \eta_i + \int_{\partial \mathcal {D}_i}\log |x|\frac{\p \log |x|}{\p \vec{n}}$$
where $\mathcal {D}_i= \frac{1}{\eta_i}\left( \Phi^{-1} (\hD_i)-x_i\right)$. 
Let us denote by $h$ the map
\begin{equation*}
h(x) = \log |x| \frac{x \cdot \vec{n}}{|x|^2}.
\end{equation*}
Changing variables again, we have
\begin{equation*}
\int_{\partial\mathcal{D}_i} h =\frac{1}{\eta_i} \int_{\partial \hD_i} h\left(\frac{\Phi^{-1}(y)  -x_i}{\eta_i}\right) |\det D\Phi^{-1}(y)|.
\end{equation*}

By Taylor expansion,   we may write that if $\hD_i $ does not intersect $\Phi(\partial \Sigma)$, i.e. if $D_i$ does not intersect $\partial \Sigma$, we have  on $ \hD_i$, 
\begin{equation*}
\Phi^{-1}(y) - \Phi^{-1}(y_i) = D \Phi^{-1}(y_i) (y-y_i) + O\left(\|\psi\|_{C^{1,1}}  \eta_i^{2}\right).
\end{equation*}
We also get, uniformly on $\partial \hD_i$
\begin{equation*}
\det |D\Phi^{-1}(y)| = \det|D\Phi^{-1}(y_i)| + O\left( \|\psi\|_{C^{1,1}} \eta_i \right).
\end{equation*}
Indeed, $\psi$ is $C^{1,1}$ in the interior of $\Sigma$ and in its complement.
 Therefore
\begin{equation} \label{avantellipse}
\int_{\partial \mathcal{D}_i} h =\frac{1}{\eta_i} \int_{\partial \hD_i}
h\left(D\Phi^{-1}(y_i)\frac{ y-y_i}{\eta_i} \right)|\det D\Phi^{-1}(y_i)| dy  
+ O \left( \eta_i \|\psi\|_{C^{1,1}}  \right).
\end{equation}
For the discs $\D_i$ such that $\D_i$ intersects $\partial \Sigma$ and the support of $\psi$, we may write  instead 
$$\int_{\partial \mathcal{D}_i} h =\frac{1}{\eta_i} \int_{\partial \hD_i}
h\left(D\Phi^{-1}(y_i)\frac{ y-y_i}{\eta_i} \right)|\det D\Phi^{-1}(y_i)| dy  
+ O \left(  \|\psi\|_{C^{0,1}}  \right).$$

Now, let $\E$ be the ellipse obtained as the image of the disk $\hD_i = \D(y_i, \eta_i)$ by the affine map $y\mapsto D\Phi^{-1}(y_i)\frac{y-y_i}{\eta_i} $. Let $l_2\sqrt{l_1}$ and $\frac{\sqrt{l_1}}{l_2}$ be the length of its axes. The formula for the area of an ellipse yields
\begin{equation*}
\pi l_1 = \pi  |\det D\Phi^{-1}(y_i)|.
\end{equation*}
On the other hand, the eccentricity $l_2$ is such that
\begin{equation*}
|l_2-1| = O \left(\|D\psi\|_{\infty}\right)
\end{equation*}
Making an affine change of variables, we have
\begin{equation*} \label{calculellipse}
\frac{1}{\eta_i}\int_{\partial \hD_i} h( D\Phi^{-1}(y_i)\frac{y-y_i}{\eta_i}  )|\det D\Phi^{-1}(y_i)|dy =  \int_{\partial \E}h(y) dy.
\end{equation*}
Setting $z := \frac{y}{\sqrt{l_1}}$ and parametrizing the ellipse in polar coordinates, we have
\begin{equation*}
 \int_{\partial \E} h(y) dy  =  + 2\pi \hal \log l_1 + \int_{0}^{2\pi}  \hal\frac{  \log \left( l_2^2 \cos ^2\theta+ \frac{1}{l_2^2} \sin^2 \theta \right)  }{\left( l_2^2 \cos ^2\theta+ \frac{1}{l_2^2} \sin^2 \theta \right)  }d \theta.
\end{equation*}
Linearizing near $l_2=1$, one finds that the order $1$ term vanishes and thus
\begin{equation*}
\int_{0}^{2\pi}  \hal\frac{  \log \left( l_2^2 \cos ^2\theta+ \frac{1}{l_2^2} \sin^2 \theta \right)  }{\left( l_2^2 \cos ^2\theta+ \frac{1}{l_2^2} \sin^2 \theta \right)  }d \theta = O\left((l_2-1)^2\right).
\end{equation*}
Using the value found for $l_1$,  we thus get
\begin{multline} \label{apresellipse}
\int_{\partial (\Phi^{-1}(\hD_i))}  \log |x-x_i| \frac{\partial \log |x-x_i|}{\partial \vec{n}}
 \\ =  2 \pi \log \eta_i - \pi \log |\det D\Phi(x_i)| + O\left(   \eta_i \|\psi\|_{C^{1,1}} \right).
\end{multline}
Combining \eqref{Ap1easilycompute}, \eqref{avantellipse} and \eqref{apresellipse} we obtain that
\begin{equation*}
A'_1 =  \pi \sum_{i \in I_N} \log |\det D\Phi(x_i)|  + O\left(  \sum_{i\in I_N}  \eta_i \|\psi\|_{C^{1,1}}   \right) +  O\( \# I_{\partial, \psi}\|\psi\|_{C^{0,1}}\).
\end{equation*}

We now return to $A_1$. Using \eqref{58b}, we see that
\begin{equation*}
A_1 - A'_1 = A_{11} + 2 A_{12},
\end{equation*}
with
\begin{eqnarray*}
A_{11} & := & \sum_{i \in I_N} \left(\int_{D_i} - \int_{\Phi^{-1}(\hD_i)} \right)\left( |(D \Phi)^T \nab \tH_{N,i}^{\mu} |^2 |\det D \Phi^{-1} \circ \Phi |\right)\\
A_{12} & := & \sum_{i \in I_N} \left(\int_{D_i} - \int_{\Phi^{-1}(\hD_i)}\right)\left( \left( (D \Phi)^T \nab \tH_{N,i}^{\mu} \cdot (D \Phi)^T \nab \log |\cdot - x_i|\right) |\det D \Phi^{-1} \circ  \Phi|\right)
\end{eqnarray*}
The error term $A_{11}$ is comparable to the term $A_2$ studied in the previous step of the proof. As for $A_{12}$, we may use the Cauchy-Schwarz inequality and the previous estimates to show that it is of sub-leading order. Finally, combining \eqref{mainenergyterm}, the estimates on $A_1$ and \eqref{controleA2},  we get that
\begin{multline}\label{controleA}
A =  \Ani(\psi, \XN, \mu)  + \hal \sum_{i \in I_N} \log |\det D\Phi(x_i)|  \\+
O\left(\|  \psi\|_{C^{0,1} }^2  \int_{U_N} |\nab H_{N,s\vec{\rr}}^{\mu}|^2\right)
+O\left( s^2 \|\psi\|_{C^{0,1}}  \left(\# I_N +  \int_{U_N} |\nab H_{N, \vec{\rr}}^{\mu}|^2 \right)\right)
\\+O\left(  \sum_{i\in I_N}  \eta_i \|\psi\|_{C^{1,1}}   \right) +  O\( \# I_{\partial, \psi}\|\psi\|_{C^{0,1}}\).
\end{multline} where $A$ is as in \eqref{A}.
\end{step}

\begin{step}[The error term \eqref{B}]
Using the fact that $\Phi$ transports $\mu$ onto $\nu$, we find
\begin{equation}\label{Bcontrole}
B = O \left( s^2 \# I_N  \|\psi\|_{C^{0,1}}  \right).
\end{equation}
\end{step}

\begin{step}[Conclusion]

Combining \eqref{55}, \eqref{controleA} and \eqref{Bcontrole} and linearizing $\log \det$, bounding $\eta_i$ by $\frac{s}{\sqrt{N}}$, we obtain the result of  Proposition  \ref{prop:comparaison2}.
\end{step}

\end{proof}

\section{Auxiliary results}\label{app2}
\subsection{Proof of Lemma \ref{lem:splitting}} \label{sec:preuvelemsplitting}
\begin{proof}[Proof of Lemma \ref{lem:splitting}]
Denoting $\triangle $ the diagonal in $\R^\d\times \R^\d$ we may write 
\begin{multline} \label{finh}
 \HN(\XN)  =  \sum_{i \neq j} \g(x_i- x_j) + N \sum_{i=1}^N V(x_i)\\
  =   \iint_{\triangle^c} \g(x-y)\left(\sum_{i=1}^N \delta_{x_i}\right) (x)  \left(\sum_{i=1}^N \delta_{x_i}\right) (y)+ N \int_{\R^\d} V \left(\sum_{i=1}^N \delta_{x_i}\right)(x) \\
 =  N^2 \iint_{\triangle^c}  \g(x-y) d\mueq(x) d\mueq(y) + N^2 \int_{\R^\d} V d\mueq
\\   + 2N \iint_{\triangle^c}  \g(x-y)d\mueq(x) d\fluct_N(y) + N \int_{\R^\d} V d\fluct_N \\ 
 + \iint_{\triangle^c} \g(x-y)d\fluct_N(x)d\fluct_N(y).
\end{multline}

We now recall that $\zeta_0$ was defined in \eqref{zeta} 
and that $\zeta_0=0$  in $\Sigma$.
With the help of this we may rewrite the medium line in the right-hand side of \eqref{finh} as  
\begin{multline*}
2N \iint_{\triangle^c} \g(x-y)d\mueq(x) d\fluct_N(y) + N \int_{\R^\d} V d\fluct_N \\
 = 2N  \int_{\R^\d} \left(h^{ \mueq} + \frac{V}{2}\right) d\fluct_N = 2N  \int_{\R^\d} (\zeta_0 + c) d\fluct_N \\
 = 2N \int_{\R^\d} \zeta_0    \Big(\sum_{i=1}^N \delta_{x_i}- N d\mueq\Big)= 2N \sum_{i=1}^N \zeta_0(x_i).
\end{multline*}
The last equalities are   due to the facts that  $\zeta_0 \equiv 0$ on the support of $\mueq$ and that  $\int \fluct_N=0$ since $\mueq$ is a  probability measure.   We also have to notice that since  $\mueq$ is absolutely continuous with respect to the Lebesgue measure, we may include the diagonal  back into the domain of integration.
By that same argument, one may recognize in the first line of the right-hand side of \eqref{finh} the quantity $N^2 \mathcal{I}_V(\mueq)$. 
\end{proof}

\subsection{Proof of Proposition \ref{prop:monoto}} \label{sec:preuvepropmonoto}
\begin{proof}[Proof of Proposition \ref{prop:monoto}]
For the proof, we drop the superscripts $\mu$. 
First we notice that $\int_{\R^\d}|\nab H_{N, \vec{\eta}}|^2 $ is a convergent integral and that 
\be\label{intdoub}\int_{\R^\d}|\nab H_{N, \vec{\eta}}|^2=\cd \iint \g(x-y)\left( \sum_{i=1}^N \delta_{x_i}^{(\eta_i)} -N d\mu\right)(x) \left( \sum_{i=1}^N \delta_{x_i}^{(\eta_i)} -N d \mu\right)(y).\ee
Indeed, we may choose $R$ large enough so that all the points of $\XN$ are contained in the ball $B_R = B(0, R)$.
    By Green's formula  and \eqref{eqhne}, we have
\be\label{greensplit1}
\int_{B_R} |\nabla H_{N,\vec{\eta}}|^2 \\= \int_{\partial B_R} H_{N,\vec{\eta}} \frac{\partial H_N}{\partial \vec{n}} - \cd\int_{B_R} H_{N,\vec{\eta}} \left(   \sum_{i=1}^N \delta_{x_i}^{(\eta_i)}-N d\mu\right)  .
\ee
Since $\int \sum_i \delta_{x_i}-N\mueq=0$, the function $H_N$ decreases like $1/|x|^{\d-1}$ and $\nab H_N$ like $1/|x|^{\d}$ as $|x|\to \infty$, 
hence  the boundary integral tends to $0$ as $R \to \infty$, and  we may write 
$$\int_{\R^\d}|\nab H_{N, \vec{\eta}} |^2= \cd \int_{\R^\d} H_{N,\vec{\eta}} \left(   \sum_{i=1}^N \delta_{x_i}^{(\eta_i)}-N d\mu\right)  $$ and thus by \eqref{eqhne},  \eqref{intdoub} holds.
We may next write 
\begin{multline}
\label{l3}\iint \g(x-y)\left( \sum_{i=1}^N \delta_{x_i}^{(\eta_i)} -N d\mu\right)(x) \left( \sum_{i=1}^N \delta_{x_i}^{(\eta_i)} -N d\mu\right)(y) 
\\
-\iint_{\triangle^c} \g(x-y)\, d\fluct_N(x)\, d\fluct_N(y)
 = \sum_{i=1}^N \g(\eta_i) \\
 + \sum_{i\neq j} \iint \g(x-y) \left(\delta_{x_i}^{(\eta_i)}(x) \delta_{x_j}^{(\eta_j)} (y)-   \delta_{x_i}(x)\delta_{x_j}(y)\right)+2 N\sum_{i=1}^N\iint \g(x-y)\left( \delta_{x_i}-\delta_{x_i}^{(\eta_i)} \right)(x) d\mu(y).\end{multline}
Let us now observe that $\int \g(x-y)\delta_{x_i}^{(\eta_i)} (y)$, the potential generated by $\delta_{x_i}^{(\eta_i)}$ is equal to 
$$
\int \g(x-y) \delta_{x_i}
$$
outside of $B(x_i,\eta_i)$, and is smaller otherwise. Since its  Laplacian is $-\cd \delta_{x_i}^{(\eta_i)}$, a negative measure,  this  is also a superharmonic function, so by the maximum principle, its value at a point $x_j$ is larger or equal to its average on a sphere centered at $x_j$. Moreover, outside $B(x_i,\eta_i)$ it is a harmonic function, so its values are equal to its averages. We deduce from these considerations, and reversing the roles of $i $ and $j$,  that for each $i\neq j$,
$$\int\g(x-y) \delta_{x_i}^{(\eta_i)}(x) \delta_{x_j}^{(\eta_j)}(y) \le \int\g(x-y) \delta_{x_i}(x) \delta_{x_j}^{(\eta_j)}(y)
\le \int\g(x-y) \delta_{x_i} (x)\delta_{x_j}(y)$$
with equality if $B(x_i,\eta_i) \cap B(x_j,\eta_j) = \varnothing.$
We may also obviously  write  
\begin{multline*}
\g(x_i-x_j) - \min (\g(\eta_i), \g(\eta_j))\le 
\int\g(x-y) \delta_{x_i} (x) \delta_{x_j}(y) - \int\g(x-y) \delta_{x_i}^{(\eta_i)}(x) \delta_{x_j}^{(\eta_j)}(y)
\\
\le \g(x_i-x_j) \indic_{ |x_i-x_j |\le \eta_i+\eta_j}.
\end{multline*}
We conclude that the second term  in the right-hand side of \eqref{l3} is nonnegative, equal to  $0$ if all the balls are disjoint,  bounded above by $\sum_{i\neq j} \g(x_i-x_j)  \indic_{ |x_i-x_j |\le \eta_i+\eta_j}$ and below by $\sum_{i \neq j}\left( \g(x_i-x_j) - \min (\g(\eta_i), \g(\eta_j))\right)  \indic_{ |x_i-x_j |\le \eta_i+\eta_j}.$
By the above considerations,  since 
$$\int \g(x-y) \delta_{x_i}^{(\eta_i)}(y) = \int \g(x-y)\delta_{x_i}(y)$$ outside $B(x_i,\eta_i)$, 
we may rewrite the last term in the right-hand side of \eqref{l3} as 
$$2 N\sum_{i=1}^N\int_{B(x_i,\eta_i)} ( \g(x-x_i)- \g(\eta_i)) d\mu(x).$$
Finally, if $\mu\in L^\infty$ then 
\begin{equation}\label{fdmu}
\left|\sum_{i=1}^N \int_{\R^\d}\f_{\eta_i} d\mu\right|\le  C_\d \|\mu\|_{L^\infty} \sum_{i=1}^N\eta_i^2.\end{equation}
Indeed, 
it suffices to observe that 
\begin{equation}\label{intfeta}\int_{B(0,\eta)} \f_{\eta}= \int_0^\eta (\g(r)-\g(\eta))r^{\d-1}\, dr= -\int_0^r \g'(r) r^\d\, dr ,\end{equation}
with an integration by parts (and an abuse of notation, i.e. viewing $\g$ as a function on $\R^\d$) and using  the explicit form of $\g$ it follows that 
\be\label{bfeeta}\int_{B(0,\eta)} |\f_{\eta}|\le C_\d \eta^2.\ee
Combining the above, by the definition \eqref{def:truncation}, we obtain the result. 
\end{proof}

\subsection{Proof of Lemma \ref{lem:contrdist}}
\label{sec:preuvecontrdist}
\def\fae{\f_{\alpha_i,\eta_i}}
\def\faej{\f_{\alpha_j, \eta_j}}

We again drop the $\mu$ superscripts.
For any $\alpha \le \eta$, let us denote $\f_{\alpha,\eta}= \f_{\alpha}-\f_{\eta}$ and note that it vanishes outside $B(0, \eta)$ and 
$$\g(\eta)-\g(\alpha) \le \f_{\alpha ,\eta}\le 0$$ 
while 
\begin{equation}\label{eqfae}
- \Delta \f_{\alpha, \eta}= \cd (\delta_0^{(\eta)}- \delta_0^{(\alpha)}).\end{equation}
 Let us choose $\vec{\eta}$ such that $\eta_i= \eta= N^{-1/\d}$ for each $i$, and $\vec{\alpha}$ such that 
\begin{equation}\label{defal}
\alpha_i=\begin{cases} 
\rr(x_i) \quad &  \text{if} \ B(x_i, N^{-1/\d}) \subset U_N=B(\bar x_N,\L_N) \\
N^{-1/\d} \quad & \text{otherwise}\end{cases}\end{equation}
 Let us denote by $I_N$ the set of $i$'s such that $B(x_i, N^{-1/\d}) \subset U_N$. We recall that  by assumption, $\a_i= \eta_i$ for the points outside $U_N$.

Noting that by \eqref{def:HNmutrun}, we have
$$H_{N,\vec{\eta}}(x) - H_{N,\vec{\alpha}}(x) = \sum_{i\in I_N} \fae(x-x_i),$$
we may compute
\begin{multline*}
T:= \int_{\R^\d}| \nab H_{N,\vec{\eta}}|^2 - \int_{\R^\d}  |\nab H_{N,\vec{\alpha}}|^2
= 2 \int  (\nab H_{N,\vec{\eta}}- \nab H_{N,\vec{\alpha}} ) \cdot \nab H_{N,\vec{\alpha}}+ \int  |\nab H_{N,\vec{\eta}}- \nab H_{N,\vec{\alpha}}|^2 \\
= 2\sum_{i\in I_N } \int \nab \fae(x-x_i)  \cdot \nab H_{N,\vec{\alpha}}+ \sum_{i,j \in I_N} \int  \nab \fae(x-x_i) \cdot \nab \faej (x-x_j).
\end{multline*}
Using an integration by parts and  \eqref{dhne} and \eqref{eqfae}, we obtain 
\begin{multline}\label{ii1i2}
T \\ 
=  2 \cd\sum_{i\in I_N}\int  \fae(x-x_i) \Big( \sum_{j=1}^N \delta_{x_j}^{(\alpha_j)} - N d\mu\Big)(x) 
 +\cd \sum_{i,j\in I_N } \int  \fae(x-x_i) \( \delta_{x_j}^{(\eta_j)}-\delta_{x_j}^{(\alpha_j)} \) (x)\\
= \cd\sum_{i\in I_N} \int  \fae(x-x_i) \( \sum_{j\in I_N} \delta_{x_j}^{(\alpha_j)} + \delta_{x_j}^{(\eta_j)}\) (x) - \cd\sum_{i\in I_N} \int  \fae(x-x_i)  2N d\mueq(x)
.\end{multline}
The first term is nonpositive  since $\fae $ is.  For the diagonal terms, we note that 
  $$  \int  \fae(x-x_i)   \(  \delta_{x_i}^{(\alpha_i)} + \delta_{x_i}^{(\eta_i)}\) (x)    =  - (\g(\alpha_i)-\g(\eta_i))$$ by definition of $\f_{\alpha,\eta}$ and the fact that $\delta_0^{(\alpha)}$ is a measure of mass $1$ on $\partial B(0,\alpha)$.
  For the off diagonal terms,  noting that $|x_i-x_j| \ge 2( \rr(x_i )+ \rr(x_j))$ by definition \eqref{def:trxi}, we may bound 
\begin{multline*} \int  \fae(x-x_i)  \sum_{j\in I_N} ( \delta_{x_j}^{(\alpha_j)} + \delta_{x_j}^{(\eta_j)}) (x)
\le  \int  \fae(x-x_i) \sum_{x_j\  \text{nearest neighbor to}\  x_i}  \delta_{x_j}^{(\rr(x_j))}
\\ \le   -  \g(|x_i-x_j|- \rr(x_i)- \rr(x_j)) +\g(\eta)\le \g(\eta)- \g(\rr(x_i)),\end{multline*} 
where we used the monotonicity of $\g$.

 For the last term in \eqref{ii1i2} we may use \eqref{bfeeta} to bound it by 
$CN \sum_{i\in I_N} \eta_i^2\|\mueq\|_{L^\infty}$.
We have thus obtained  that
\begin{equation}\label{i1}
T \le  -\cd \sum_{i\in I_N}  (\g(\alpha_i)- \g(\eta_i)) + \cd \sum_{i\in  I_N} \g(\eta_i)- \g(\rr(x_i))  
   \\+  C  N^{1-\frac{2}{\d}}  \|\mu\|_{L^\infty}\# I_N     .\end{equation}
Inserting \eqref{contrnbpoints}, we conclude that 
\begin{multline}\label{numer}
\int_{\R^\d} |\nab H_{N,\vec{\eta}}|^2  - \int_{\R^\d} |\nab H_{N,\vec{\alpha}}|^2\\ \le  -\cd \sum_{i\in I_N}  (\g(\alpha_i)- \g(N^{-\frac{1}{\d}})) + \cd \sum_{i \in I_N} \g(N^{-\frac{1}{\d}})- \g(\rr(x_i))  
\\+ C N^{2-\frac{2}{\d}}|U_N|  \|\mueq\|_{L^\infty}^2 + C N^{1-\frac{2}{\d}} \|\mueq\|_{L^\infty} |U_N|^{\frac{\d-2}{2\d}} \|\nab H_{N,\vec{\eta}}\|_{L^2(U_N)}.
\end{multline}
Recalling that $\a_i= \eta_i$ for the points outside $U_N$  we obtain 
\begin{multline*}
\cd \sum_{i\in I_N}  N^{1-\frac{2}{\d}}\(\g(\rr(x_i) N^{1/\d})- \g(1)  \)
\le  \int_{U_N} |\nab H_{N,\vec{\a}}|^2
- \cd \sum_{i\in I_N}  N^{1-\frac{2}{\d}}\(\g(\rr(x_i) N^{1/\d})- \g(1)  \)\\+ C_{\mu_0} N^{1-\frac{2}{\d}}
\( N |U_N| +   N^{1-\frac{2}{\d}} |U_N|^{1-\frac{2}{\d}}+ N^{\frac{2}{\d}-1}\int_{U_N} |\nab H_{N,\vec{\eta}}|^2\)\end{multline*}
In view of  the definitions \eqref{defal} and \eqref{deflocen}, we have 
$$ \int_{U_N} |\nab H_{N,\vec{\a}}|^2- \cd  N^{1-\frac{2}{\d}}\sum_{i\in I_N} \g (\rr(x_i) N^{1/\d} ) 
=N^{1-\frac{2}{\d}} \EnergieLoc_{\vec{\a}}(U_N)$$
and we obtain the result.

\subsection{Proof of Proposition \ref{prop:fluctenergy} and Corollary \ref{coro:nombredepointspresdubord}}
\label{sec:preuvefluctenergy}
\ed{
\begin{proof}[Proof of Proposition \ref{prop:fluctenergy}]

 We may take  $\chi$  a smooth cutoff function equal to $1$ in $ U_N\cap \{\dist (x, \p U_N \ge \delta/2)\}$ (we note that this set contains $\Supp \varphi+B(0,N^{-1/\d})$) and equal to $0$ outside $U_N $, such that $\|\nab \chi\|_{L^\infty}\le C\delta^{-1}$, this way 
 $$
 \| \nab \chi \|_{L^2} \leq C |\p U_N|^{\hal} \delta^{-\hal}.
 $$
Integrating \eqref{dhne} against $\chi$ we get 
\begin{equation}\label{wothci}
\left|\int \chi  \Big( \sum_{i=1}^N \delta_{x_i}^{(\eta_i)} - Nd\mu\Big) \right|\le  C \|\nab \chi\|_{L^2}\|\nab H_{N,\vec{\eta}}^{\mu}\|_{L^2(U_N )}.\end{equation}
Letting $\#I_{\varphi}$ denote the number of balls $B(x_i, N^{-1/\d})$ intersecting the support of $\varphi$, we may write 
\begin{equation}\label{Ip}
\# I_{\varphi}\le \int \chi \sum_{i=1}^N \delta_{x_i}^{(\eta_i)}\le   N\int_{U_N}d \mu+ C\( \frac{|\p U_N|}{\delta}\)^{\hal} \|\nab H_{N,\vec{\eta}}^{\mu} \|_{L^2(U_N)}.\end{equation}
Secondly, in view of \eqref{dhne}, we  have 
\begin{equation}\label{rel3}
\left|\int  \Big( \sum_{i=1}^N \delta_{x_i}^{(\eta_i)} - N\mu\Big) \varphi\right|= \frac{1}{\cd}\left| \int_{\R^\d}\nabla H_{N,\vec{\eta}}^{\mu} \cdot \nab \varphi\right|\le \frac{1}{\cd} \|\nab \varphi\|_{L^2(U_N)} \|\nab H_{N,\vec{\eta}}^{\mu}\|_{L^2(U_N)}.
\end{equation}
Thirdly,
\begin{equation}
\left|\int \left(d\fluct_N- \Big( \sum_{i=1}^N \delta_{x_i}^{(\eta)} - N d\mu\Big) \right) \varphi\right|= 
\left|\int  \Big( \sum_{i=1}^N( \delta_{x_i} - \delta_{x_i}^{(\eta)}) \Big) \varphi\right|\le \# I_\varphi  N^{-\frac{1}{\d}} \|\nab \varphi\|_{L^\infty} .\end{equation}
Combining with \eqref{Ip}, we get the result. Finally, the proof of the estimate \eqref{contrnbpoints} is completely analogous to that of \eqref{Ip}.
\end{proof}

\begin{proof}[Proof of Corollary \ref{coro:nombredepointspresdubord}] 
Let $\chi$ be a nonnegative cutoff function such that $\chi(x)= 1 $ if $x\in \Sigma $ and $\dist(x , \partial \Sigma) \ge \max (2r, N^{- \frac{1}{3}})$, $\chi(x)=0$ if $\dist(x, \partial \Sigma) \le r$ or $ x \notin \Sigma$, and such that $|\nab \chi|\le C \min\( N^{\frac{1}{3}}, \frac{1}{r}\right)$ and $|\chi| \leq 1$.
By definition of $\# I_{\partial}^r$ and $\chi$, we have
$$
\# I_{\partial}^r = \int_{\{\chi=0\}} \sum_{i=1}^N \delta_{x_i}.
$$
Since $\partial \Sigma$ is regular, for any $x_i$ in $\{\chi=0\}$,  if $r$ is small enough, there is at least $1/C$ of the mass of $\delta_{x_i}^{(\eta_i)}$ which still belongs to $\{\chi=0\}$ (where $C$ depends only on $\p \Sigma$), thus
\begin{multline*}
\int_{\{\chi=0\}} \sum_{i=1}^N \delta_{x_i}  \le 3 \int_{\{\chi=0\}} \sum_{i=1}^N \delta_{x_i}^{(\eta_i)} \\
 \le C\left( \int_{\R^{\d}} (1-\chi)  \left(\sum_{i=1}^N \delta_{x_i}^{(\eta_i)} - N  d\mu\right) +\int_{\{1-\chi >0\}} Nd \mu\right)
\\
\le C \int_{\R^{\d}} - \chi \( \sum_{i=1}^N \delta_{x_i}^{(\eta_i)} - N d \mu\) + C|\partial \Sigma| \|\mu\|_{L^\infty} \max (N^{1-\frac{1}{3}}, 2Nr),
\end{multline*}
where we have used that the total mass of $\left( \sum_{i=1}^N \delta_{x_i}^{(\eta_i)} - N d\mu\right)$ is $0$.
On the other hand, \eqref{wothci} gives
$$
\left|\int_{\R^\d} \chi \, \left(\sum_{i=1}^N \delta_{x_i}^{(\eta_i)}- Nd \mu \right) \right| \le C \min \( N^{\frac{1}{6}}, r^{-\hal}\right)
 \|\nab H_{N,\vec{\eta}}^{\mu} \|_{L^2(\Sigma)}.
 $$
  The conclusion follows.
\end{proof}

}

\subsection{Proof of Lemma \ref{lem:termesfactorises} and Lemma \ref{lem:comparerlesentropies}}
\subsubsection{A preliminary result}
Let us recall that for an arbitrary function $f$, $f^{\Sigma}$ is the harmonic extension of $f$ outside $\Sigma$ and let us first prove a preliminary result.
\begin{lem}\label{lemS}
For any $f \in C^{0,1}(\R^\d)$ we have
\begin{equation}\label{claim76}
\int_{\R^\d} f d(\mueqt-\mueq) = -  \frac{t}{\cd \beta} \int_{\R^\d}  f^{\Sigma}\Delta \xi   +  O(\|f\|_{C^{0,1}}   t^2),
\end{equation}where $O$ depends only on the $C^{3,1}$ norms of $V$ and $\xi$, and  the lower bound on $\mueq$ on its support.
\end{lem}
\begin{proof}  Since $\mueq$ and $\mueqt$ are supported near $\Sigma$, we may try to replace $f$ by a function which coincides with $f$ in $\Sigma$ and has a convenient behavior outside $\Sigma$. To do so, we first write
\begin{equation} \label{claim76a}
\int_{\R^\d} f d(\mueqt-\mueq) = \int_{\R^\d}  f^{\Sigma}  d (\mueqt-\mueq) + \|f\|_{C^{0,1}} O( t^2).
\end{equation}
Indeed, $f-f^{\Sigma}$ is supported in $\Sigma^c$ and $\mueqt$ is supported in $\Sigma_t$, hence
\begin{equation*}
\int_{\R^\d} (f-f^{\Sigma})  d(\mueqt-\mueq) = \int_{\Sigma_t \backslash \Sigma} (f-f^{\Sigma})  d\mueqt .
\end{equation*}
From Proposition \ref{proserser} we know that $\Sigma_t$ is contained in a $O(t)$-tubular neighborhood of $\Sigma$ with $C\le \frac{\|\nab \xi\|_{L^\infty}}{\mueq}$. Since $f$ and $f^{\Sigma}$ coincide on $\Sigma$ (up to the boundary) and are Lipschitz in $\Sigma^c$, we have $f-f^{\Sigma} = O(\|f\|_{C^{0,1}} t)$ on $\Sigma_t \backslash \Sigma$ as $t \t0$,  we deduce
\begin{equation*}
\left|\int_{\R^\d} (f-f^{\Sigma}) d (\mueqt-\mueq)  \right|\le  C\|f\|_{C^{0,1}} 
\end{equation*}with $C$ depending only on the $C^3$ norms of $V$, $\xi$ and the lower bound on  $\mueq$,
which proves \eqref{claim76a}.
Next, using the result of Proposition \ref{proserser}, we may write 
$$\int_{\R^\d}  f^{\Sigma} d(\mueqt-\mueq)=-\frac{t}{\cd \beta} \int_{\Sigma } f^\Sigma \Delta \xi+\int_{\Sigma_t\backslash \Sigma} f^\Sigma d \mueq  + O(t^2),$$
with 
$$\int_{\Sigma_t\backslash \Sigma} f^\Sigma d \mueq= \frac{t}{\cd\beta} \int_{\p \Sigma}  f^\Sigma \left[\nab \xi^\Sigma\right] \cdot \vec{n}+ O(t^2).$$
By definition of $\Delta \xi^\Sigma$, the result follows.
  \end{proof}
\subsubsection{Proof of Lemma \ref{lem:termesfactorises}}
\label{sec:preuvetermesfactorises}
We now turn to proving Lemma \ref{lem:termesfactorises},  which we recall only needs to be proven in the boundary macroscopic case where $\xi_N= \xi$.
\begin{proof}
Using an integration by parts, the definition of $\zeta_t$ and \eqref{def:mut}, we have
\begin{multline*}
\frac{1}{\cd} \int_{\R^\d}|\nab h^{\mut- \mueqt}|^2 =\int_{\R^\d} (h^{\mut - \mueqt})  d(\mut-  \mueqt)= \i (\zeta_0 -\zetat+c_0-c_t)d (\mut-  \mueqt)\\
= \int_{\R^\d} \zeta_0 \left(d\mueq - \frac{t}{\cd \beta} \Delta \xi\right)- \int_{\R^\d} \zeta_0 d \mueqt + \zetat  \left(d \mueq -\frac{t}{\cd \beta}\Delta\xi\right)\\
=-\frac{t}{\cd \beta}  \int_{\R^\d}(\zeta_0-\zetat)\Delta \xi- \int_{\R^\d} (\zeta_0 d \mueqt+ \zetat d \mueq)
=
-\frac{t}{\cd \beta}  \int_{\R^\d} (\zeta_0-\zetat)\Delta \xi +O(t^3),
\end{multline*}
where the last term is negligible by the same reasoning as above.
On the other hand, integrating by parts and using \eqref{zeta} we get
\begin{equation*}
-\frac{t}{\cd \beta}  \int_{\R^\d} (\zeta_0-\zetat)\Delta \xi = \frac{t}{\beta} \int_{\R^\d} \xi d(\mueq-\mueqt) -\frac{t^2}{\cd\beta^2} \int_{\R^\d} \xi \Delta \xi.
\end{equation*}
Finally, integrating $\int_{\R^\d} \xi \Delta \xi$ by parts, we obtain
\begin{equation}\label{exif}
 \frac{t^2 }{2\cd \beta}\int_{\R^\d} |\nab \xi|^2 - \frac{ \beta  }{2\cd} \int_{\R^\d} |\nab h^{\mut-\mueqt}|^2 = \frac{ t}{2} \int_{\R^\d} \xi d (\mueq-\mueqt) +O(t^3 ).
 \end{equation}
Now we apply \eqref{claim76} to $f = \xi$ and get
\begin{equation*}
 \frac{t^2 }{2\cd\beta}\int_{\R^\d} |\nab \xi|^2 - \frac{ \beta  }{2\cd} \int_{\R^\d} |\nab h^{\mut-\mueqt}|^2 = -\frac{t^2}{2\cd} \int_{\R^\d} \xi^{\Sigma} \Delta \xi^{\Sigma} + O( t^3)
\end{equation*}  which  yields \eqref{251}.

\ed{
Next, we turn to \eqref{252}. Since $\zetaz$ vanishes in $\Sigma$ we have
\begin{equation*}
\i \zetaz d\mueqt = \int_{\Sigma_t \backslash \Sigma} \zetaz d\mueqt.
\end{equation*}
But $\mueqt= \left( \frac{\Delta V}{4\pi}-\frac{t}{2\pi \beta} \Delta \xi\right) \indic_{\Sigma_t}$ hence $\|\mueqt\|_{L^\infty} \le C $ when $|t|<1$.
In view of Proposition \ref{proserser} and \eqref{bzetaA}, we deduce that \eqref{252} holds. }\end{proof}

\subsubsection{Proof of Lemma \ref{lem:comparerlesentropies}}
\label{sec:preuvecomparerlesentropies}
\begin{proof}
In the boundary case we use \eqref{tmut}, i.e. 
\begin{equation}\label{denmut}
\tilde \mueqt = \frac{1}{2\cd} \left( \Delta V - \frac{2t}{\beta} \Delta \xi +u  \right) \mathbf{1}_{\phi_t(\Sigma)}, \quad \mueq = \frac{1}{2\cd} \Delta V \mathbf{1}_{\Sigma}
\end{equation}
with 
$$\|u\|_{L^\infty} \le C t^2 \|\xi\|_{C^{1,1}}^2 + C t^2 \|\xi\|_{C^{2,1}}\|\xi\|_{C^{0,1}}.$$
In the interior case we have \eqref{eq3}. 

We find after an elementary computation in all cases,
\begin{multline*}
\int_{\R^\d} \tilde \mueqt \log \tilde \mueqt - \int_{\R^\d} \mueq \log \mueq \\ = \int_{\R^\d} (\tilde \mueqt - \mueq) \log \Delta V - \frac{t}{\cd \beta} \int_{\phi_t(\Sigma)} (\Delta \xi  +u)  + O\left(t^2 \|\xi\|_{C^{0,1}}+ \|u\|_{L^\infty})|U_N|\right).
\end{multline*}
Since $\phi_t(\Sigma)$ and $\Sigma$ are $O(t\|\xi\|_{C^{0,1}})$-close by Proposition \ref{proserser} and $\Delta \xi$ is continuous, we also have
\begin{equation*}
\int_{\phi_t(\Sigma)} \Delta \xi = \int_{\Sigma} \Delta \xi + O(t).
\end{equation*}
Finally, applying \eqref{claim76} to $\log \Delta V$, which is $C^{0,1}$ in  a neighborhood of $\Sigma$, yields the result.
\end{proof}

\section{Comparison of partition functions - some more detail}
\label{sec:comparaisonmieux}
In this section, we provide a more detailed explanation of how to obtain Proposition \ref{comparaisonmacro} from the results of \cite{ls1} and \cite{loiloc}. It is intended for an hypothetical interested reader who would already have some familiarity with those papers on which we strongly rely.

 \subsection{The macroscopic case}
 \def \M{\mathcal{M}}
Let us explain why the expansion of the partition function $\log \KNbeta(\mu)$ up to order $N$ can be done uniformly (i.e. with a uniform error term $o(N)$) over subsets $K$ of probability densities $\mu$ satisfying some uniform assumptions, as described below.

Let $K \subset \R^2$ be a compact set and let $\M$ be a set of probability densities supported on $K$ such that
\begin{enumerate}
\item There exists $c > 0$ such that the density $\mu$ is bounded below by $c$ on its support for all $\mu$ in $\M$.
\item There exists $\kappa > 0$ and $C_{\kappa} > 0$ such that $\| \mu \|_{C^{0, \kappa}(K)} \leq C_{\kappa}$ for all $\mu$ in $\M$, where $\| \cdot \|_{C^{0, \kappa}(K)}$ denotes the usual Hölder norm on $K$.\end{enumerate}
We want to prove that for all $\epsilon > 0$, there exists $N$ large enough such that for all $\mu, \nu$ in $\M$ we have
\begin{equation} \label{uniformitecompar1}
\left| \log \KNbeta(\mu) - \log \KNbeta(\nu) + N\left(1 - \frac{\beta}{4} \right) \left(\Ent(\mu) - \Ent(\nu) \right) \right| \leq \epsilon N.
\end{equation}
\def \bEmp{\overline{P}_{N}}
\def \config{\mathcal{X}}
\def \Kun{\mathcal{K^0}}
\def \KM{\mathcal{K}_M}
\def \bP{\overline{P}}
\def \bPx{\overline{P}^x}
\def \W{\overline{\mathbb{W}}}
\def \Ww{\mathbb{W}}
\def \ERS{\mathsf{ent}}
\def \bERS{\overline{\mathsf{ent}}}
\def \fbarbeta{\overline{\mathcal{F}}_{\beta}}
\def \Amu{\mathcal{A}_{\mu}}

Let us first observe that if $\M$ is finite then the result follows directly from \cite[Corollary~1.5.]{ls1}, where an expansion of $\log \KNbeta(\mu)$ is proven for any $\mu$ which is bounded below and Hölder. This expansion can be written as
$$
\log \KNbeta(\mu) = - N \min \fbarbeta^{\mu} + o(N),
$$
where $\fbarbeta^{\mu}$ is a functional on the space of random tagged point configurations $\bP \in \probas(\Sigma \times \config)$. Let us briefly describe  these objects  (we refer to \cite{ls1} for details).
\begin{itemize}
\item $\Sigma \times \config$ is the space of “tagged point processes”, where the first component is the “tag” (a point in $\Sigma$) and the second one is a (locally finite) point configuration
\item If $\bP$ is in $\probas(\Sigma \times \config)$ (probability measures on $\Sigma \times \config$) we consider, for any $x \in \Sigma$, the associated random point process $\bPx \in \probas(\config)$.
\item For $m > 0$ we denote by $\Ww^{m}$ the renormalized energy (of random point process) as in \cite{ls1}, computed with respect to a background of intensity $m$. We extend it to a functional on $\probas(\Sigma \times \config)$ by
$$
\W^{\mu}(\bP) := \int_{\Sigma} \Ww^{\mu(x)}(\bPx) dx.
$$
In particular, it is observed in \cite{ls1} that $\W^{\mu}(\bP)$ is finite only if $\bPx$ has intensity $\mu(x)$ for Lebesgue-a.e. $x$.
\item There is a “specific relative entropy” functional $\ERS$ defined on $\probas(\config)$ and extended on $\probas(\Sigma \times \config)$ by
$$
\bERS[\bP | \Pi^1] := \int_{\Sigma} \ERS[\bP^x|\Pi^1] dx
$$
\item For any $m > 0$ there is a functional $\fbeta^m$ defined on $\probas(\config)$ by
$$
\fbeta^m := \frac{\beta}{2} \Ww^m + \ERS[ \cdot | \Pi^1]
$$
 and we extend it on $\probas(\Sigma \times \config)$ by
$$
\fbarbeta^{\mu}(\bP) := \int_{\Sigma} \fbarbeta^{\mu(x)}(\bPx) dx.
$$

\end{itemize}

The scaling property of $\min \fbarbeta^{\mu}$ under variations of $\mu$ allows to express $\log \KNbeta(\mu) - \log \KNbeta(\nu)$ only in terms of $\mu$ and $\nu$, as in \eqref{uniformitecompar1}.

In order to argue for the uniformity of the error term on (infinite) sets $\M$ as above, we need to follow the lines of the proof of \cite{ls1}. The central object is the \textit{tagged empirical field} $\bEmp$, which encodes the microscopic behavior of the particles. It is a random variable with value in $\probas(\Sigma \times \config)$, where $\Sigma$ depends on $\mu$ but we may see $\Sigma$ as included in $K$ for all $\mu$ in $\M$ and not worry about it anymore.

\textbf{Step 1.} \textit{Uniform upper bound.} \ \\
For a fixed $\mu \in \M$, it is proven in \cite[Lemma 4.1.]{ls1} that the law of $\bEmp$ is exponentially tight in $\probas(K \times \config)$. In fact the argument of \cite{ls1} shows that there exists a compact set $\Kun \subset \probas(K \times \config)$ such that
$$\bEmp \in \Kun,\, \PNbeta^{\mu}-\text{almost surely},$$
and the \textit{same compact set} works for any $\mu \in \M$.

For any $P \in \Kun$, for any $\delta > 0$, there exists $\epsilon > 0$ such that
\begin{enumerate}
\item If $\{ \mu_N \}_N$ is a sequence in $\M$ (converging, up to extraction, to $\mu$) and $\XN$ is such that $\bEmp \in B(\bP, \epsilon)$ for $N$ large enough, then
$$\liminf_{N \ti} \frac{1}{N} \FN^{\mu_N}(\XN) \geq \W^{\mu}(\bP) - \delta,$$
this follows from the $\Gamma$-$\liminf$ property.
\item The exponential volume of $B(\bP, \epsilon)$ under the reference measure (in the sense of \cite{ls1}) is less than $- N\bERS[\bP|\Pi^1] + \delta N$. This follows from the definition of $\bERS[\bP|\Pi^1]$.
\end{enumerate}
For a fixed $\delta > 0$, we may cover $\Kun$ by a finite union of such balls $\{B(\bP_i, \epsilon_i)\}_i$. Let $\{\mu_N\}_N$ be a sequence in $\M$ (up to extraction, we may assume that it converges to some $\mu \in \M$). We can then estimate $\log \KNbeta^{\mu_N}$ by covering the domain of integration ($(\R^2)^N)$) by the sets
$$
\Omega_i := \left\lbrace \bEmp \in B(\bP_i, \epsilon_i) \right\rbrace_{i \in I}.
$$
By construction,  on $\Omega_i$ we have  $\FN^{\mu_N} \geq N \W^{\mu}(\bP) - \delta N,$
and moreover the exponential volume of $\Omega_i$ is less than $- N \bERS[\bP|\Pi^1] + \delta N$. Since by definition $\fbarbeta := \frac{\beta}{2} \W^{\mu} + \bERS[\cdot|\Pi^1]$, we deduce that
$$
\limsup_{N \ti} \log \KNbeta^{\mu_N} \leq -N  \fbarbeta^{\mu}(\bP)  + 2 \delta N.
$$
This provides a uniform upper bound on $\log \KNbeta^{\mu}$ in terms of minimization of $\fbarbeta^{\mu}$.

Thus the uniformity of the \textit{upper bound} relies mostly on the $\Gamma$-$\liminf$ connecting $\FN^{\mu}$ and $\W_{\mu}$, which can be proven using standard functional analysis (see \cite[Lemma 3.1]{ls1} or \cite[Proposition 5.2]{PetSer}). In particular, the regularity of $\mu$ plays a negligible role, and we could replace both assumptions on $\M$ by the fact that the probability densities in $\M$ are uniformly bounded.

\textbf{Step 2.} \textit{Uniform lower bound} \ \\
The lower bound on the partition function is obtained in \cite{ls1} by the mean of an “explicit” construction of point configurations. The proof is quite lengthy, and here we only sketch the reasons why the lower bound could be made \textit{uniform} on $\M$.
\def \bmi{\overline{m}_i}

To simplify, for a given $\mu$ the argument goes as follows: first we cut $\Sigma$ (the support of $\mu$) into rectangles $\{K_i\}_i$ of sidelength $\approx \frac{R}{\sqrt{N}}$ (with $R$ large), such that $N_i := N \int_{K_i} d\mu \text{ is an integer}$. It appears that the rectangles will have a uniformly controlled aspect ratio (i.e we can find $C$ uniform such that the sidelengths are in $(R, R + C/R)$) if the lower bound on $\mu$ is uniform on $\M$. We also let $\bmi := \frac{1}{|K_i|} \int_{K_i}  d\mu$. This allows to get from a varying background measure to a piecewise constant one, but this has a cost proportional to
\begin{equation} \label{coutholder}
N C^2_{\kappa} R^3 \left(\frac{R}{\sqrt{N}}\right)^{2 \kappa},
\end{equation}
where $C_{\kappa}$ is a bound on $\|\mu\|_{C^{0, \kappa}}$. It is thus important to have a \textit{uniform} bound on the Hölder norms in order to control this error term uniformly. However, we may observe (this will be of particular importance below) that for $R$ fixed, the error term is $o(N)$ as long as
\begin{equation} \label{Holderpastropgrand}
C_{\kappa}  \ll N^{\kappa/2}.
\end{equation}

Then, roughly speaking, in each $K_i$ we paste a copy of the point process minimizing $\fbeta^{\bmi}$. Minimizers of $\fbeta^{m}$ are rescaled versions of each other (as $m$ varies), and here the scaling factor is uniformly bounded (above and below) on $i \in I$ and $\mu \in \M$ if the lower/upper bounds on $\mu$ are uniform on $\M$ (which is true by assumption). This ensures that the “screening-then-regularization” procedure of \cite{ls1} can be done “uniformly”.

This construction yields a set of $N$-point configurations whose exponential volume is close $N\bERS[\bP|\Pi^1]$ and whose energy is bounded above by $N \W^{\mu}(\bP)$, where $\bP$ is some minimizer of $\fbarbeta^{\mu}$. This yields a lower bound on $\log \KNbeta^{\mu}$ in terms of $\min \fbarbeta^{\mu}$ with an error term which is uniform for $\mu \in \M$.

\subsection{The mesoscopic case}
The paper \cite{loiloc} shows that the analysis of \cite{ls1} can be done at any mesoscopic scale. For a fixed $\mu$, any  $\delta \in (0, \hal)$ and any square $C(z_0, N^{-\delta})$ of sidelength $N^{-\delta}$ centered at $z_0 \in \Sigma$, the contribution of $\log \KNbeta(\mu)$ “due to” the points in  $C(z_0, N^{-\delta})$ is approximately
$$
- N^{1-2\delta} \min \fbarbeta^{\mu(z_0)}  + o\left(N^{1-2\delta} \right).
$$
In order to prove \eqref{comparaisonmeso}, we want to compare $\log \KNbeta(\mu)$ with $\log \KNbeta(\mueqt)$, where $\mueqt$ can be written as
$$
\mueqt = \mu + t \rho\left( \frac{\cdot}{\L_N} \right),
$$
where $\rho$ has mean zero, is supported in $\Sigma$, is bounded below by a positive constant and belongs to some Hölder space $C^{0, \kappa}$. The family of measures
$$\left\lbrace\mu + t \rho\left( \frac{\cdot}{\L_N} \right) \right\rbrace_{t \in [0,1], N\geq 1}$$ is thus uniformly bounded below, and the Hölder norm (on the mesoscopic square) scales like
$$
\left\| \mu + t \rho\left( \frac{\cdot}{\L_N} \right) \right\|_{C^{0, \kappa}(C(0, \L_N))} = O\left( N^{\delta \kappa} \right) \text{ for } t \in [0,1].
$$
The fact that $\delta < \hal$ ensures that \eqref{Holderpastropgrand} is satisfied, and the procedure of the macroscopic case can be applied \textit{uniformly} in the mesoscopic square.

The other step of \cite{loiloc} is to prove that the interior, mesoscopic square of sidelength $\L_N$ and the macroscopic exterior can be decoupled. This relies on a screening argument, which depends only on the exterior part for which the density is constant in $N$ (and so this step is essentially independent on $N$).

\bibliographystyle{alpha}
\bibliography{CLTbib}

\end{document}